\documentclass[a4paper,11pt]{article}
\usepackage{amsmath,amssymb,amsthm}
\usepackage{graphicx, color}
\usepackage[top=80pt, bottom=90pt, left=70pt, right=70pt]{geometry}
\usepackage{multirow}
\usepackage[colorlinks=true,linkcolor=blue,citecolor=blue]{hyperref}
\usepackage{cleveref}

\def\qqed{\hfill $\blacksquare$}

\newcommand{\argmin}{\operatornamewithlimits{argmin}}
\newcommand{\ave}[1]{\ensuremath{\langle#1\rangle} }
\newcommand{\TEST}{{\sc Testing Quadratic M${}_2$-Representability}}
\newcommand{\DECOMP}{{\sc Decomposition}}
\newcommand{\LAM}{{\sc Laminarization}}
\DeclareMathOperator{\dom}{dom}

\newcommand{\R}{\mathbf{R}}
\newcommand{\ER}{\overline{\mathbf{R}}}
\newcommand{\Z}{\mathbf{Z}}
\newcommand{\DF}{U_{\mathcal{A}}}
\newcommand{\HP}{U_{n,r}}

\newtheorem{thm}{\bfseries Theorem}[section]
\newtheorem{thm1}{\bfseries Theorem}
\newtheorem{lem}[thm]{\bfseries Lemma} 
\newtheorem{prop}[thm]{\bfseries Proposition} 
\newtheorem{cor}[thm]{\bfseries Corollary}
\newtheorem{cl}[thm1]{\bfseries Claim}
\newtheorem*{cl*}{\bfseries Claim}

\theoremstyle{definition}
\newtheorem{remark}[thm]{\bfseries Remark} 
\newtheorem{exmp}[thm]{\bfseries Example}

\crefname{thm}{Theorem}{Theorems}
\crefname{prop}{Proposition}{Propositions}
\crefname{lem}{Lemma}{Lemmas}
\crefname{exmp}{Example}{Examples}
\crefname{cor}{Corollary}{Corollarys}
\crefname{cl}{Claim}{Claims}
\crefname{remark}{Remark}{Remarks}
\crefname{section}{Section}{Sections}

\usepackage{bm}
\makeatletter
\renewcommand*{\@fnsymbol}[1]{\ensuremath{\ifcase#1\or *\or 1\or 2\or
		3\or 4\or \dag\or \ddag \else\@ctrerr\fi}}
\makeatother
\makeatletter

\@addtoreset{equation}{section}
\makeatother
\allowdisplaybreaks

\begin{document}
\title{A tractable class of binary VCSPs via M-convex intersection\thanks{A preliminary version of this paper~\cite{STACS/HIMZ18} has appeared in the proceedings of the 35th International Symposium on Theoretical Aspects of Computer Science (STACS 2018).
The work was done while Yuni Iwamasa was at the University of Tokyo.}}
\author{Hiroshi Hirai\thanks{Department of Mathematical Informatics,
		Graduate School of Information Science and Technology,
		University of Tokyo, Tokyo, 113-8656, Japan.
		Email: \texttt{hirai@mist.i.u-tokyo.ac.jp}}
	\and Yuni Iwamasa\thanks{
		National Institute of Informatics, Tokyo, 101-8430, Japan.
		Email: \texttt{yuni\_iwamasa@nii.ac.jp}}
	\and Kazuo Murota\thanks{Department of Business Administration,
		Tokyo Metropolitan University, Tokyo, 192-0397, Japan.
		Email: \texttt{murota@tmu.ac.jp}} \and Stanislav \v{Z}ivn\'y\thanks{Department of Computer Science,
		University of Oxford, Oxford, OX1 3QD, United Kingdom.
		Email: \texttt{standa.zivny@cs.ox.ac.uk}}}
\maketitle

\begin{abstract}
	A binary VCSP is a general framework for the minimization problem of a function represented as the sum of unary and binary cost functions.
	An important line of VCSP research is to investigate what functions can be solved in polynomial time.
	Cooper and \v{Z}ivn\'{y} classified the tractability of binary VCSP instances according to the concept of ``triangle,''
	and showed that the only interesting tractable case is the one induced by the joint winner property (JWP).
	Recently, Iwamasa, Murota, and \v{Z}ivn\'{y} made a link between VCSP and discrete convex analysis, showing that a function satisfying the JWP can be transformed into a function represented as the sum of two quadratic M-convex functions, which can be minimized in polynomial time via an M-convex intersection algorithm if the value oracle of each M-convex function is given.
	
	In this paper,
	we give an algorithmic answer to a natural question: What binary finite-valued CSP instances can be represented as the sum of two quadratic M-convex functions and can be solved in polynomial time via an M-convex intersection algorithm?
	We solve this problem by devising a polynomial-time algorithm for obtaining a concrete form of the representation in the representable case.
	Our result presents a larger tractable class of binary finite-valued CSPs, which properly contains the JWP class.
\end{abstract}
\begin{quote}
\end{quote}

\section{Introduction}\label{sec:intro}
The {\it valued constraint satisfaction problem (VCSP)} provides a general framework for discrete optimization (see \cite{book/Zivny12} for details).
Informally, the VCSP framework deals with the minimization problem of a function represented as the sum of ``small'' arity functions,
which are called {\it cost functions}.
It is known that various kinds of combinatorial optimization problems can be formulated in the VCSP framework.
In general, the VCSP is NP-hard.
An important line of research is to investigate what restrictions on classes of VCSP instances ensure polynomial time solvability.
Two main types of VCSPs with restrictions are {\it structure-based VCSPs} and {\it language-based VCSPs} (see e.g.,~\cite{book/KrokhinZivny17}).
Structure-based VCSPs deal with restrictions on graph structures representing the appearance of variables in a given instance.
For example,
it is known (e.g.,~\cite{book/BerteleBrioschi72}) that if the graph (named the {\it Gaifman graph}) corresponding to a VCSP instance has a bounded treewidth
then the instance can be solved in polynomial time.
Language-based VCSPs deal with restrictions on cost functions that appear in a VCSP instance.
Kolmogorov, Thapper, and \v{Z}ivn\'{y}~\cite{SICOMP/KTZ15} gave a precise characterization of tractable valued constraint languages via the basic LP relaxation.
Kolmogorov, Krokhin, and Rol\'{i}nek~\cite{SICOMP/KKR17} gave a dichotomy for all language-based VCSPs
(see also~\cite{FOCS/B17,FOCS/Z17} for a dichotomy for all language-based CSPs).

{\it Hybrid VCSPs},
which deal with a combination of structure-based and language-based restrictions,
have emerged recently~\cite{incollection/CZ17}.
Among many kinds of hybrid restrictions,
a {\it binary VCSP}, VCSP with only unary and binary cost functions, is a representative hybrid restriction that includes numerous fundamental optimization problems.
Cooper and \v{Z}ivn\'{y}~\cite{AI/CZ11} showed that
if a given binary VCSP instance satisfies the {\it joint winner property (JWP)},
then it can be minimized in polynomial time.
The same authors classified in~\cite{JAIR/CZ12} the tractability of binary VCSP instances according to the concept of ``triangle,''
and showed that the only interesting tractable case is the one induced by the JWP (see also~\cite{incollection/CZ17}).
Furthermore,
they introduced {\it cross-free convexity} as a generalization of JWP,
and devised a polynomial-time minimization algorithm for cross-free convex instances $F$
when a ``cross-free representation'' of $F$ is given; see related works below for details.

In this paper,
we introduce a novel tractability principle going beyond triangle and cross-free representation for binary finite-valued CSPs,
denoted from now on as binary VCSPs.
A binary VCSP is formulated as follows,
where $D_1, D_2, \dots, D_r$ ($r \geq 2$) are finite sets.
\begin{description}
	\item[Given:] Unary cost functions $F_p : D_p \rightarrow \R$ for $p \in \{1,2, \dots, r \}$ and binary cost functions
	$F_{pq} : D_p \times D_q \rightarrow \R$
	for $1 \leq p < q \leq r$.
	\item[Problem:] Find a minimizer of
	$F : D_1 \times D_2 \times \cdots \times D_r \rightarrow \R$
	defined by
	\begin{align}\label{eq:binary VCSP}
	F(X_1,X_2,\dots,X_r) := \sum_{1 \leq p \leq r} F_p(X_p) + \sum_{1\leq p < q \leq r}F_{pq}(X_p, X_q).
	\end{align}
\end{description}
Our tractability principle is built on {\it discrete convex analysis (DCA)}~\cite{MP/M98, book/Murota03},
which is a theory of convex functions on discrete structures.
In DCA, {\it L-convexity} and {\it M-convexity} play primary roles;
the former is a generalization of submodularity, and the latter is a generalization of matroids.
A variety of polynomially solvable problems in discrete optimization can be understood within the framework of L-convexity/M-convexity (see e.g.,~\cite{book/Murota03, incollection/M09,  JMID/M16}).
Recently, it has also turned out that discrete convexity is deeply linked to tractable classes of VCSPs.
L-convexity is closely related to the tractability of language-based VCSPs.
Various kinds of submodularity induce tractable classes of language-based VCSP instances~\cite{SICOMP/KTZ15},
and a larger class of such submodularity can be understood as L-convexity on certain graph structures~\cite{MPSA/H16}; see also~\cite{JORSJ/H18}.
On the other hand, Iwamasa, Murota, and \v{Z}ivn\'{y}~\cite{DO/IMZ18} have pointed out that M-convexity plays a role in hybrid VCSPs.
They revealed the reason for the tractability of a VCSP instance satisfying the JWP from a viewpoint of M-convexity.
We here continue this line of research,
and explore further applications of M-convexity in hybrid VCSPs.

A function $f : \{0,1\}^n \rightarrow \R \cup \{ +\infty\}$ is called {\it M-convex}~\cite{AM/M96,book/Murota03} if it satisfies the following generalization of the matroid exchange axiom: for $x = (x_1, x_2, \dots, x_n)$ and $y = (y_1, y_2, \dots, y_n)$ with $f(x) < +\infty$ and $f(y) < +\infty$,
and $i \in \{ 1,2,\dots,n \}$ with $x_i > y_i$, there exists $j \in \{ 1,2,\dots,n \}$ with $y_j > x_j$ such that
\begin{align*}
f(x) + f(y) \geq f(x - \chi_i + \chi_j) + f(y + \chi_i - \chi_j),
\end{align*}
where $\chi_i$ is the $i$th unit vector.
Although M-convex functions are defined on $\mathbf{Z}^n$ in general,
we only need functions on $\{0,1\}^n$ here.
M-convex functions on $\{0,1\}^n$ are equivalent to the negative of {\it valuated matroids} introduced by Dress and Wenzel~\cite{AML/DW90, AM/DW92}.
An M-convex function can be minimized in a greedy fashion similarly to the greedy algorithm for matroids.
Furthermore, a function $f : \{0,1\}^n \rightarrow \R \cup \{+\infty\}$ that is representable
as the sum of two M-convex functions is called {\it M${}_2$-convex}.
In particular,
$f$ is called {\it quadratically representable M${}_2$-convex (QR-M${}_2$-convex)}
if $f$ is representable as the sum of two quadratic M-convex functions.
As a generalization of matroid intersection,
the problem of minimizing an M${}_2$-convex function,
called the {\it M-convex intersection problem},
can also be solved in polynomial time if the value oracle of each constituent M-convex function is given~\cite{SIDMA/M96_1, SIDMA/M96_2}; see also~\cite[Section~5.2]{book/Murota00}.
Our proposed tractable class of VCSPs is based on this result.

Let us return to binary VCSPs.
The starting observation for relating VCSP to DCA is that
the objective function $F$ on $D_1 \times D_2 \times \cdots \times D_r$ can be regarded as a function $f$ on $\{0, 1\}^n$ with $n := \sum_{1 \leq p \leq r} |D_p|$ by the following correspondence between the domains:
\begin{align}\label{eq:D to {0,1}}
D_p := \{1,2,\dots, n_p \} \ni i\  \longleftrightarrow\ (\underbrace{0, \dots, 0,\overset{i}{\check{1}},0, \dots, 0}_{n_p}).
\end{align}
With this correspondence,
the minimization of $F$ can be transformed to that of $f$.
A binary VCSP instance $F$ is said to be {\it quadratic M${}_2$-representable}
if the function $f$ obtained from $F$ via the correspondence~\eqref{eq:D to {0,1}}
is QR-M${}_2$-convex.

It is shown in~\cite{DO/IMZ18} that a binary VCSP instance satisfying the JWP can be transformed to a quadratic M${}_2$-representable instance,\footnote[6]{In~\cite{DO/IMZ18},
	a binary VCSP instance satisfying the JWP was transformed into the sum of two quadratic M${}^\natural$-convex functions.
	It can be easily seen that this function can also be transformed into the sum of two quadratic M-convex functions.}
and two M-convex summands can be obtained in polynomial time.
Here the following natural question arises:
{\it What binary VCSP instances are quadratic M${}_2$-representable?}
In this paper, we give an algorithmic answer to this question by considering the following problem:
\begin{description}
	\item[\underline\TEST]
	\item[Given:] A binary VCSP instance $F$.
	\item[Problem:] Determine whether $F$ is quadratic M${}_2$-representable or not.
	If $F$ is quadratic M${}_2$-representable,
	obtain a decomposition $f = f_1 + f_2$ of the function $f$ into two quadratic M-convex functions $f_1$ and $f_2$,
	where $f$ is the function transformed from $F$ via~\eqref{eq:D to {0,1}}.
\end{description}

Our main result is the following:
\begin{thm}\label{thm:main}
	\TEST\ can be solved in $O(n^4)$ time.
\end{thm}
An M${}_2$-convex function $f$ can be minimized in $O(nr^3 + nr \log n)$ time if such a decomposition is given
(the time complexity can be easily derived from a minimization algorithm for M${}_2$-convex functions in~\cite{SIDMA/M96_2}).
Thus we obtain the following corollary of \cref{thm:main}.
\begin{cor}
	A quadratic M${}_2$-representable binary VCSP instance can be minimized in $O(n^4)$ time.
\end{cor}

\paragraph{Overview.}
We outline our approach to \TEST\
via taking a small concrete example of a quadratic M${}_2$-representable binary VCSP instance.
	Suppose that $D_1 = D_2 = D_3 = D_4 = \{0, 1\}$.
	Unary cost functions $F_1, F_2, F_3, F_4$ and binary cost functions $F_{pq}$ ($1 \leq p < q \leq 4$)
	are given by
	\begin{align}\label{eq:F instance}
	\begin{split}
	F_1 :=
	\left[
	\begin{array}{c}
	1\\
	0
	\end{array}
	\right],\quad
	F_2 :=
	\left[
	\begin{array}{c}
	0\\
	1
	\end{array}
	\right],\quad
	F_3 :=
	\left[
	\begin{array}{c}
	1\\
	0
	\end{array}
	\right],
\quad
	F_4 :=
	\left[
	\begin{array}{c}
	0\\
	1
	\end{array}
	\right],\\
	F_{12} :=
	\left[
	\begin{array}{cc}
	3 & 0 \\
	1 & 4
	\end{array}
	\right],
	\qquad
	F_{13} :=
	\left[
	\begin{array}{cc}
	2 & 0 \\
	1 & 3
	\end{array}
	\right],
	\qquad
	F_{14} :=
	\left[
	\begin{array}{cc}
	1 & 0 \\
	0 & 1
	\end{array}
	\right],\\
	F_{23} :=
	\left[
	\begin{array}{cc}
	2 & 0 \\
	0 & 2
	\end{array}
	\right],
	\qquad
	F_{24} :=
	\left[
	\begin{array}{cc}
	3 & 0 \\
	1 & 1
	\end{array}
	\right],
	\qquad
	F_{34} :=
	\left[
	\begin{array}{cc}
	2 & 0 \\
	0 & 0
	\end{array}
	\right],
	\end{split}
	\end{align}
	where $F_{pq}$ is regarded as a $2 \times 2$ matrix with the $(i,j)$-component $F_{pq}(i-1, j-1)$ for $1 \leq i,j \leq 2$,
	and
	$F_p$ is also regarded as a two-dimensional vector in a similar way.
	Based on the correspondence~\eqref{eq:D to {0,1}},
	the function $f$ on $\{0,1\}^8$ is constructed as follows
	(this construction will be introduced formally in \cref{subsec:representation}):
	\begin{align}\label{eq:f instance}
	f(x) :=
	\frac{1}{2} x^\top
	\left[
	\begin{array}{cc|cc|cc|cc}
	0 & \infty & \multicolumn{2}{c|}{\multirow{2}{*}{$F_{12}$}} & \multicolumn{2}{c|}{\multirow{2}{*}{$F_{13}$}} & \multicolumn{2}{c}{\multirow{2}{*}{$F_{14}$}} \\
	\infty & 0 & \multicolumn{2}{c|}{} & \multicolumn{2}{c|}{} & \multicolumn{2}{c}{} \\ \hline
	 \multicolumn{2}{c|}{\multirow{2}{*}{$F_{12}^\top$}} & 0 & \infty & \multicolumn{2}{c|}{\multirow{2}{*}{$F_{23}$}} & \multicolumn{2}{c}{\multirow{2}{*}{$F_{24}$}} \\
	 \multicolumn{2}{c|}{} & \infty & 0 & \multicolumn{2}{c|}{} & \multicolumn{2}{c}{} \\ \hline
	 \multicolumn{2}{c|}{\multirow{2}{*}{$F_{13}^\top$}} & \multicolumn{2}{c|}{\multirow{2}{*}{$F_{23}^\top$}} & 0 & \infty  & \multicolumn{2}{c}{\multirow{2}{*}{$F_{34}$}} \\
	 \multicolumn{2}{c|}{} & \multicolumn{2}{c|}{} & \infty & 0 & \multicolumn{2}{c}{} \\ \hline
	 \multicolumn{2}{c|}{\multirow{2}{*}{$F_{14}^\top$}} & \multicolumn{2}{c|}{\multirow{2}{*}{$F_{24}^\top$}} & \multicolumn{2}{c|}{\multirow{2}{*}{$F_{34}^\top$}} & 0 & \infty \\
	 \multicolumn{2}{c|}{} & \multicolumn{2}{c|}{} & \multicolumn{2}{c|}{} & \infty & 0
	\end{array}
	\right] x
	+
	\left[
	\begin{array}{c}
	\multirow{2}{*}{$F_1$} \\ \\ \hline
	\multirow{2}{*}{$F_2$} \\ \\ \hline
	\multirow{2}{*}{$F_3$} \\ \\ \hline
	\multirow{2}{*}{$F_4$} \\ \\
	\end{array}
	\right]^\top x
	\end{align}
	for $x \in \{0,1\}^8$ with $\sum_{1 \leq i \leq 8} x_i = 4$
	and $f(x) := +\infty$ for other $x$.
	Recall that $F$ is quadratic M${}_2$-representable if and only if $f$ is QR-M${}_2$-convex
	and that $F$ is efficiently minimizable if and only if $f$ is.
	
	Our algorithm constructs the following two M-convex summands $f_1$ and $f_2$ of $f$:
	\begin{align}
	f_1(x) :=
	\frac{1}{2}
	\left[
	\begin{array}{c}
	x_1\\
	x_3\\
	x_5\\
	x_7\\
	x_4\\
	x_8\\
	x_2\\
	x_6
	\end{array}
	\right]^\top
	\left[
	\begin{array}{cccccccc}
	6 & \multicolumn{1}{c|}{6} & \multicolumn{1}{c|}{4} & \multicolumn{1}{c|}{2} & \multicolumn{4}{c}{\multirow{4}{*}{{\huge $0$}}} \\
	6 & \multicolumn{1}{c|}{6} & \multicolumn{1}{c|}{4} & \multicolumn{1}{c|}{2} & \multicolumn{4}{c}{} \\ \cline{1-2}
	4 & 4 & \multicolumn{1}{c|}{4} & \multicolumn{1}{c|}{2} & \multicolumn{4}{c}{} \\ \cline{1-3}
	2 & 2 & 2 & \multicolumn{1}{c|}{2} & \multicolumn{4}{c}{} \\ \cline{1-4} \cline{5-6}
	\multicolumn{4}{c}{\multirow{4}{*}{{\huge $0$}}} & \multicolumn{1}{|c}{2} & \multicolumn{1}{c|}{2} & & \\
	\multicolumn{4}{c}{} & \multicolumn{1}{|c}{2} & \multicolumn{1}{c|}{2} & & \\ \cline{5-6}
	\multicolumn{4}{c}{} & & & 0 &  \\
	\multicolumn{4}{c}{} & & & & 0
	\end{array}
	\right] \left[
	\begin{array}{c}
	x_1\\
	x_3\\
	x_5\\
	x_7\\
	x_4\\
	x_8\\
	x_2\\
	x_6
	\end{array}
	\right]\label{eq:f_1 instance}
	\end{align}
	and
	\begin{align}
	f_2(x) :=
	\frac{1}{2} \left[
	\begin{array}{c}
	x_1\\
	x_2\\
	x_3\\
	x_4\\
	x_5\\
	x_6\\
	x_7\\
	x_8
	\end{array}
	\right]^\top
	\left[
	\begin{array}{cccccccc}
	0 & \multicolumn{1}{c|}{\infty} & & & \multicolumn{4}{c}{\multirow{4}{*}{{\huge $0$}}} \\
	\infty & \multicolumn{1}{c|}{0} & & & \multicolumn{4}{c}{} \\ \cline{1-2} \cline{3-4}
	& & \multicolumn{1}{|c}{0} & \multicolumn{1}{c|}{\infty} & \multicolumn{4}{c}{} \\
	& & \multicolumn{1}{|c}{\infty} & \multicolumn{1}{c|}{0} & \multicolumn{4}{c}{} \\ \cline{3-4} \cline{5-6}
	\multicolumn{4}{c}{\multirow{4}{*}{{\huge $0$}}} & \multicolumn{1}{|c}{0} & \multicolumn{1}{c|}{\infty} & & \\
	\multicolumn{4}{c}{} & \multicolumn{1}{|c}{\infty} & \multicolumn{1}{c|}{0} & & \\ \cline{5-6} \cline{7-8}
	\multicolumn{4}{c}{} & & & \multicolumn{1}{|c}{0} & \multicolumn{1}{c}{\infty}  \\
	\multicolumn{4}{c}{} & & & \multicolumn{1}{|c}{\infty} & \multicolumn{1}{c}{0}
	\end{array}
	\right] \left[
	\begin{array}{c}
	x_1\\
	x_2\\
	x_3\\
	x_4\\
	x_5\\
	x_6\\
	x_7\\
	x_8
	\end{array}
	\right]
	+
	\left[
	\begin{array}{c}
	-3\\
	7\\
	-8\\
	-1\\
	-3\\
	2\\
	0\\
	1
	\end{array}
	\right]^\top \left[
	\begin{array}{c}
	x_1\\
	x_2\\
	x_3\\
	x_4\\
	x_5\\
	x_6\\
	x_7\\
	x_8
	\end{array}
	\right]\label{eq:f_2 instance}
\end{align}
for $x \in \{0,1\}^8$ with $\sum_{1 \leq i \leq 8} x_i = 4$,
and $f_1(x) := +\infty$ and $f_2(x) := +\infty$ for other $x$.
The first function $f_1$ in~\eqref{eq:f_1 instance} is a {\it laminar convex function}~\cite[Section~6.3]{book/Murota03},
which is a typical example of M-convex functions.
Indeed, by using a laminar family $\mathcal{L} = \{ \{1,3,5,7\}, \{1, 3, 5\}, \{1, 3\}, \{4, 8\} \}$,
$f_1$ is written as
\begin{align}\label{eq:f_1 quadratic}
f_1(x) = \sum_{X \in \mathcal{L}} \left( \sum_{i \in X} x_i \right)^2.
\end{align}
The second function $f_2$ in~\eqref{eq:f_2 instance} is nothing but a linear function on the base family of the partition matroid with partition $\{ \{1,2\}, \{3,4\}, \{5, 6\}, \{7, 8\} \}$,
and hence $f_2$ is also M-convex.

We establish a representation theorem (\cref{thm:representation}), which says that QR-M${}_2$-convex functions arising from binary VCSP instances always admit the above type of the decomposition.
For a set $X \subseteq \{1,2,\dots,n\}$,
let $\overline{\ell}_X$ be the quadratic function defined on $\{0,1\}^n$ by
\begin{align}\label{eq:overline l}
\overline{\ell}_X(x) := \left(\sum_{i \in X} x_i\right)^2.
\end{align}
The theorem states that
a function $f$ arising from a binary VCSP instance is QR-M${}_2$-convex
if and only if
$f$ is
a laminar convex function restricted to the base family of the partition matroid
with partition $\mathcal{A}$ of $\{1,2,\dots,n\}$, i.e.,
\begin{align*}
f = \sum_{X \in \mathcal{L}}c_X \overline{\ell}_X + h + \delta_\mathcal{A},
\end{align*}
where $\mathcal{L}$ is a laminar family,
$c_X$ is a positive weight on $X \in \mathcal{L}$,
$h$ is a linear function,
and
$\delta_\mathcal{A}$ is the $\{0,+\infty\}$-function taking $0$ on the bases
and $+\infty$ on the non-bases.

The main difficulty in solving \TEST\ is that a representation of quadratic functions
on the base family of the partition matroid is not unique.
Indeed,
we see that
the coefficients in~\eqref{eq:f instance}
do not equal the sum of coefficients in~\eqref{eq:f_1 instance} and~\eqref{eq:f_2 instance}.
In particular,
$\overline{\ell}_X$ satisfies the following relations:
\begin{align}
\overline{\ell}_X + \delta_\mathcal{A} &= \overline{\ell}_{X \cup A_p} + h + \delta_\mathcal{A} \qquad \textrm{if $A_p \in \mathcal{A}$ and $A_p \cap X = \emptyset$},\label{eq:X + A_p}\\
\overline{\ell}_X + \delta_\mathcal{A} &= \overline{\ell}_{\{1,2,\dots,n\} \setminus X} + h' + \delta_\mathcal{A},\label{eq:n - X}
\end{align}
where $h$ and $h'$ are linear functions.
This means that
$f$ can be QR-M${}_2$-convex
even if $f$ is written as
\begin{align}\label{eq:f intro}
f = \sum_{X \in \mathcal{F}}c_X \overline{\ell}_X + h + \delta_\mathcal{A}
\end{align}
for a non-laminar family $\mathcal{F}$.
Based on this consideration,
we divide \TEST\ into two subproblems named
\DECOMP\ and \LAM.

\DECOMP\ is the problem of obtaining a representation~\eqref{eq:f intro} of a given QR-M${}_2$-convex function $f$
for some family $\mathcal{F}$ not necessarily laminar but {\it laminarizable}
by repeating the following transformations corresponding to~\eqref{eq:X + A_p} and~\eqref{eq:n - X}:
\begin{align}
&X \mapsto X \cup A_p \textrm{ or } X \setminus A_p\label{eq:set X + A_p}\\
&X \mapsto \{1,2,\dots,n\} \setminus X\label{eq:set n - X}
\end{align}
for $X \in \mathcal{F}$,
where $A_p \cap X = \emptyset$ or $A_p \subseteq X$.
We present a polynomial-time algorithm for \DECOMP\ in \cref{sec:(i)}.
\LAM\ is the problem of constructing a laminar family $\mathcal{L}$ from the family $\mathcal{F}$ obtained in \DECOMP\
by repeating the transformations~\eqref{eq:set X + A_p} and~\eqref{eq:set n - X}.
\LAM\ can be seen as a purely combinatorial problem for a set system.
We present a polynomial-time algorithm for \LAM\ in \cref{sec:(ii)}.

\begin{figure}
	\begin{center}
		\includegraphics[width=12cm]{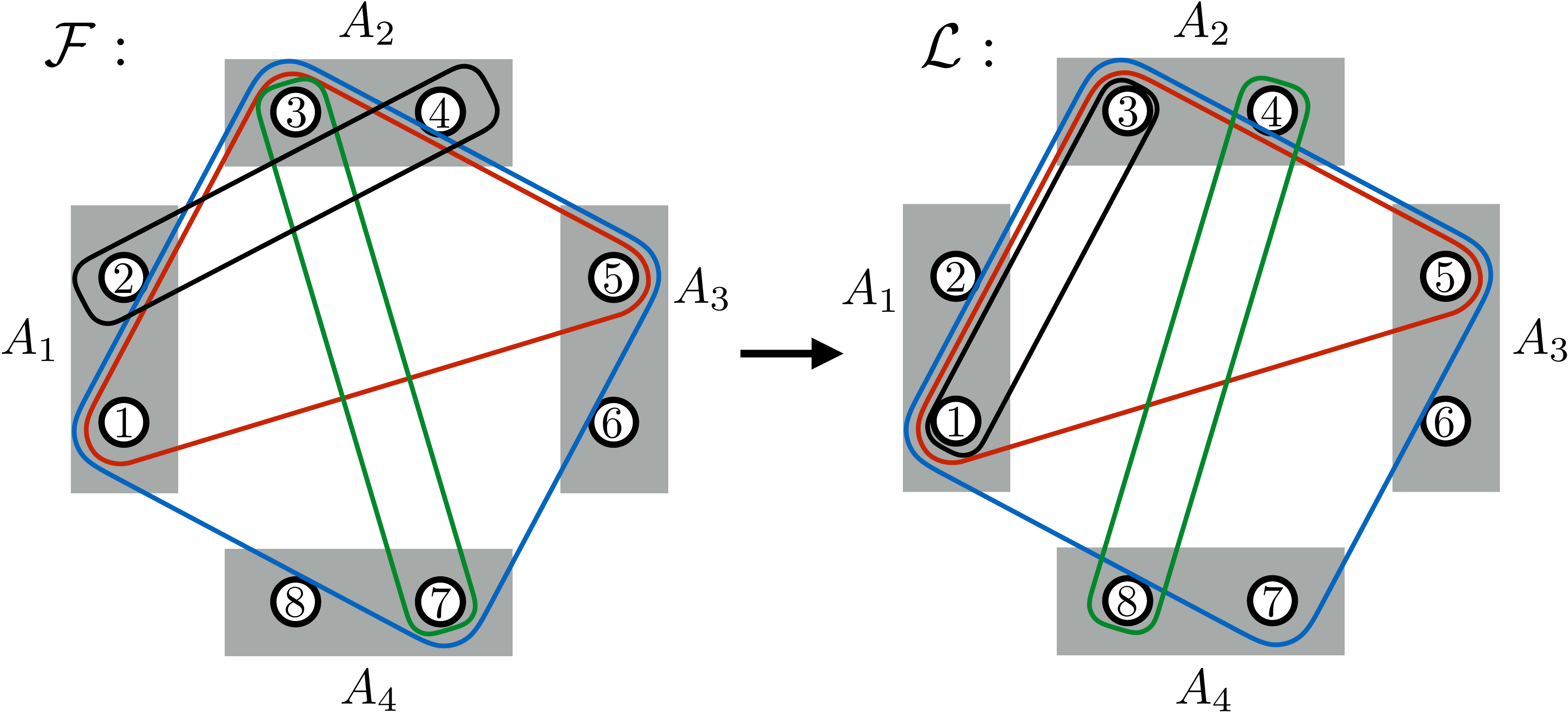}
		\caption{The left figure illustrates the input $\mathcal{F}$ of \LAM\ and the right figure illustrates an output laminar family $\mathcal{L}$, where black nodes indicate elements of $\{1,2,\dots,8\}$,
			gray rectangles indicate members in $\mathcal{A}$, and solid curves indicate four sets in $\mathcal{F}$ and $\mathcal{L}$, respectively.}
		\label{fig:Laminarization}
	\end{center}
\end{figure}

If we apply our \DECOMP\  algorithm to $f$ in~\eqref{eq:f instance},
we obtain a representation~\eqref{eq:f intro}
for a partition $\mathcal{A} := \{ \{1,2\}, \{3,4\}, \{5,6\}, \{7,8\} \}$ of $\{1,2,\dots, 8\}$
and a non-laminar family $\mathcal{F} := \{ \{1, 3, 5, 7\}, \{1, 3, 5\}, \{2, 4\}, \{3, 7\} \}$.
Then, by solving \LAM\ for $\mathcal{F}$,
we obtain a laminar family $\mathcal{L} := \{\{1,3,5,7\}, \{1, 3, 5\}, \{1, 3\}, \{4, 8\} \}$.
Indeed, we can transform $\{2,4\}$ to $\{1,3\}$ by repeating transformations~\eqref{eq:set X + A_p} and~\eqref{eq:set n - X}
since $\{1,3\} = \left(\{ 1,2,\dots,8 \} \setminus \{ 2,4 \}\right) \setminus \{ 5,6,7,8 \}$.
See Figure~\ref{fig:Laminarization}.
Thus we can verify the QR-M${}_2$-convexity of $f$
by constructing two M-convex summands of $f$.

\paragraph{Application to quadratic pseudo-Boolean function minimization.}
Pseudo-Boolean function minimization is a fundamental and well-studied problem in theoretical computer science (see e.g.,~\cite{DAM/BH02,book/CramaHammer01}).
Our result provides a new tractable class of quadratic pseudo-Boolean functions minimization.
Consider a pseudo-Boolean function $F : \{0,1\}^n \rightarrow \R$ represented as
\begin{align*}
F(x_1, x_2, \dots, x_n) = \sum_{1 \leq i < j \leq n} a_{ij} x_i x_j + \sum_{1 \leq i \leq n} a_i x_i.
\end{align*}
Then $F$ is {\it lifted} to $f : \{0,1\}^{2n} \rightarrow \R \cup \{+\infty\}$ defined by the following:
For $x \in \{0,1\}^{2n}$ with $\sum_{1 \leq i \leq 2n} x_i = n$,
\begin{align*}
f(x_1, \dots, x_n, x_{n+1}, \dots, x_{2n}) := \sum_{1 \leq i < j \leq n} a_{ij} x_i x_j + \sum_{1 \leq i \leq n}\infty \cdot x_i x_{n+i} + \sum_{1 \leq i \leq n} a_i x_i,
\end{align*}
and for other $x$, $f(x) := +\infty$.
Then $F(x_1, \dots, x_n) = f(x_1, \dots, x_n, 1 - x_1, \dots, 1 -x_n)$ for any $x \in \{0,1\}^n$.
Hence minimizing $F$ is equivalent to minimizing $f$.

We can regard $f$ as a function arising from the binary VCSP instance $F$ with the partition $\mathcal{A} := \{A_1, A_2, \dots, A_n\}$ of $\{1,2,\dots,2n\}$ given by $A_i = \{ i, n+i \}$ for $i =1,2,\dots,n$.
Therefore, if $f$ is QR-M${}_2$-convex,
then we can obtain two M-convex functions $f_1$ and $f_2$ satisfying $f = f_1 + f_2$ by our proposed algorithm,
and we can minimize $f$ (and hence $F$) in polynomial time.

To the best of our knowledge,
our new tractable class is incomparable with the existing ones,
and we are not aware of any nontrivial known tractable class contained in ours.
Tractable classes of (exactly minimizable) pseudo-Boolean functions introduced in~\cite{DAM/BH02,book/CramaHammer01} are related to
(i) bounded treewidth, (ii) submodularity, or (iii) a switching reduction
(which flips the values of a subset of the variables) to~(ii).
These tractable classes are incomparable with ours.
The minimum weight perfect bipartite matching problem constitutes another tractable class of quadratic pseudo-Boolean function minimization.
Although this problem can be formulated as a matroid intersection problem for two partition matroids,
it is outside our class since $a_{ij}$ take only finite values in our model.

\paragraph{Related works.}
\begin{itemize}
	\item
	Cooper and \v{Z}ivn\'{y}~\cite{AI/CZ11} introduced the {\it joint winner property (JWP)} for binary VCSP instances as a sufficient condition for tractability.
	A binary VCSP instance $F$ of the form~\eqref{eq:binary VCSP} is said to satisfy the {\it JWP}
	if
	\begin{align*}
	F_{ij}(a,b) \geq \min\{F_{ik}(a,c), F_{jk}(b, c)\}
	\end{align*}
	for all distinct $i,j,k \in [r]$ and all $a \in D_i, b \in D_j, c \in D_k$.
	It is shown in~\cite{AI/CZ11} that if $F$ satisfies the JWP,
	then $F$ can be transformed, in polynomial time, into a function $F'$ satisfying the JWP, $\argmin F' \subseteq \argmin F$, and
	the additional special property named the {\it Z-freeness},
	and that
	Z-free instances can be minimized in polynomial time.
	Thus, if $F$ satisfies the JWP,
	then $F$ can be minimized in polynomial time.
	Furthermore, Iwamasa, Murota, and \v{Z}ivn\'{y}~\cite{DO/IMZ18} revealed that Z-free instances are quadratic M${}_2$-representable.
	
	The tractability based on quadratic M${}_2$-representability depends solely on the function values,
	and is independent of how the function $F$ is given.
	Indeed, a quadratic M${}_2$-representable instance $F$ can be characterized 
	by the existence of a Z-free instance $F'$ that satisfies $F'(X) = F(X)$ for all $X$.
	This stands in sharp contrast with the tractability based on the JWP,
	which depends heavily on the representation of $F$.
	For example, let 
	$F(X) = \sum F_p(X_p) + \sum F_{pq}(X_p, X_q)$ 
	be a binary VCSP instance satisfying the JWP.
	By choosing a pair of distinct $p,q \in \{1,2,\dots,r\}$,
	$d \in D_p$, and $\alpha \in \R$ arbitrarily,
	replace $F_p(d)$ and $F_{pq}(d, X_q)$ by $F_p(d) + \alpha$ and $F_{pq}(d, X_q) - \alpha$, respectively.
	Then $F$ does not change but violates the JWP in general.
	Although the binary VCSP instance $F$ in~\eqref{eq:F instance} does not satisfy the JWP
	by $F_{12}(1,1) = 4$, $F_{13}(1,0) = 1$, and $F_{23}(1,0) = 0$,
	$F$ is quadratic M${}_2$-representable.
	Thus our result can explore such hidden M${}_2$-convexity.
	
	\item
	Cooper and \v{Z}ivn\'{y}~\cite{JAIR/CZ12} introduced a generalization of JWP, named {\it cross-free convexity}, for not necessarily binary VCSP instances.
	A VCSP instance $F : D_1 \times D_2 \times \cdots \times D_r \rightarrow \R \cup \{+\infty\}$ is said to be {\it cross-free convex}
	if the function $f : \{0,1\}^n \rightarrow \R \cup \{+\infty\}$ obtained from $F$ via correspondence~\eqref{eq:D to {0,1}} can be represented as
	\begin{align}\label{eq:cross-free convex}
	f(x) = \sum_{X \in \mathcal{F}} g_X\left(\sum_{i \in X} x_i\right),
	\end{align}
	where $\mathcal{F} \subseteq 2^{\{1,2,\dots,n\}}$ is {\it cross-free} and, for each $X \in \mathcal{F}$, $g_X$ is a univariate function on $\Z$ satisfying $g_X(m-1) + g_X(m+1) \geq 2g_X(m)$ for all $m \in \Z$.
	Here the equality~\eqref{eq:cross-free convex} is required for every
	$x \in \{0,1\}^n$ that corresponds to some $X \in D_1 \times D_2 \times \cdots \times D_r$ via~\eqref{eq:D to {0,1}}
	and $f(x) = +\infty$ for other $x$.
	A pair $X, Y \subseteq \{1,2,\dots,n\}$ is said to be {\it crossing}
	if $X \cap Y$, $\{1,2,\dots,n\} \setminus (X \cup Y)$, $X \setminus Y$, and $Y \setminus X$ are all nonempty,
	and a family $\mathcal{F} \subseteq 2^{\{1,2,\dots,n\}}$ is said to be {\it cross-free}
	if there is no crossing pair in $\mathcal{F}$.
	
	Cross-free convexity is a special class of M${}_2$-representability,
	where a VCSP instance $F$ is {\it M${}_2$-representable}
	if the function $f$ obtained from $F$ via correspondence~\eqref{eq:D to {0,1}} is M${}_2$-convex.
	Indeed,
	it follows from a similar argument to the M-convexity of laminar convex functions that $f$ in~\eqref{eq:cross-free convex} is M${}_2$-convex.
	Hence, a cross-free convex instance $F$ is M${}_2$-representable.
	
	Our result provides, for binary finite-valued CSPs, a polynomial-time minimization algorithm for special cross-free convex instances (quadratic M${}_2$-representable instances) even when the expression~\eqref{eq:cross-free convex} is not given.
	
	\item
	Our representation theorem (\cref{thm:representation}) is inspired by the {\it polyhedral split decomposition} due to Hirai~\cite{DCG/H06}.
	This general decomposition principle decomposes, by means of polyhedral geometry, a function on a finite set $\mathcal{D}$ of points of $\R^n$
	into a sum of simpler functions, called {\it split functions}, and a residue term.
	This aspect can be explained for our function $f$ in~\eqref{eq:f instance} roughly as follows.
	The expression $\sum_{X \in \mathcal{L} } c_X\overline{\ell}_X + f_2$ of $f$ can be viewed as the polyhedral split decomposition of $f$,
	where $\mathcal{D}$ is equal to the effective domain of $f$,
	$c_X\overline{\ell}_X$ on $\mathcal{D}$ is a sum of split functions and a linear function (cf.~\eqref{eq:l_X def}) for each $X \in \mathcal{L}$,
	and $f_2$ defined by~\eqref{eq:f_2 instance} is a residue term.
	
	\item
	Interestingly,
	\LAM\ appears in a different problem in computational biology.
	A {\it phylogenetic tree} is a graphical representation of an evolutionary history in a set of taxa in which the leaves correspond to taxa and the non-leaves correspond to speciations.
	One of the important problems in phylogenetic analysis is to assemble  
	a global phylogenetic tree
	from smaller pieces of phylogenetic trees, particularly, {\it quartet trees}.
	{\sc Quartet Compatibility} is to decide whether there is a phylogenetic tree inducing a given collection of quartet trees,
	and to construct such a phylogenetic tree if it exists.
	It is known~\cite{JC/S92} that {\sc Quartet Compatibility} is NP-hard.
	
	As a subsequent work to the present paper, Hirai and Iwamasa~\cite{ISAAC/HI18} have introduced two novel classes of quartet systems,
	named {\it complete multipartitite quartet systems} and {\it full multipartite quartet systems},
	and showed that {\sc Quartet Compatibility} for these quartet systems can be solved in polynomial time.
	In their algorithms,
	the algorithm proposed in this paper for \LAM\ is utilized for the polynomial-time solvability.
\end{itemize}

\paragraph{Notation.}
Let $\mathbf{Z}$, $\R$, $\R_+$, and $\R_{++}$ denote the sets of integers, reals, nonnegative reals, and positive reals, respectively.
In this paper, functions can take the infinite value $+ \infty$, where $a < + \infty$, $a + \infty = + \infty$ for $a \in \R$, and $0 \cdot (+\infty) = 0$.
Let $\ER := \R \cup \{+ \infty\}$.
For a function $f : \{0,1\}^n \rightarrow \overline{\R}$,
the effective domain is denoted as $\dom{f} := \{ x \in \{0,1\}^n \mid f(x) < + \infty \}$.
For a positive integer $k$, we define $[k] := \{1,2,\dots,k\}$.
We often abbreviate a set $\{i_1, i_2, \dots, i_k\}$ as $i_1 i_2 \cdots i_k$.
For $f: \{0,1\}^n \rightarrow \ER$ and $U \subseteq \{0,1\}^n$,
the function $f$ on $U$ means the ``restriction'' of $f$ obtained from $f$ by
redefining $f(x)$ as $+\infty$ for each $x \not\in U$.

\section{Representation of QR-M${}_2$-convex functions}\label{sec:toward TM2C}
For a partition $\mathcal{A} := \{A_1, A_2, \dots, A_r \}$ of $[n]$,
let $\delta_{\mathcal{A}} : \{0,1\}^n \rightarrow \ER$ be the indicator function of the base family of a partition matroid with partition $\mathcal{A}$,
that is,
$\delta_\mathcal{A}(x) := 0$ if $\sum_{i \in A_p} x_i = 1$ for each $p \in [r]$
and $\delta_\mathcal{A}(x) := +\infty$ otherwise.
Let $\DF$ be the set of characteristic vectors of the bases of a partition matroid with partition $\mathcal{A}$,
i.e., $\DF := \{ x \in \{0,1\}^n \mid \sum_{i \in A_p} x_i = 1 \ (p \in [r]) \} = \dom{\delta_\mathcal{A}}$.
Let $\HP$ be the set of characteristic vectors of the bases of the uniform matroid on $[n]$ of rank $r$,
i.e., $\HP := \{ x \in \{0,1\}^n \mid \sum_{i \in [n]} x_i = r \}$.
Note that $\DF \subsetneq \HP$ for $r \geq 2$.

\subsection{Representation theorem}\label{subsec:representation}
We introduce a class of quadratic functions on $\{0, 1\}^n$ that has a bijective correspondence to binary VCSP instances.
Let $\mathcal{A} := \{ A_1, A_2, \dots, A_r \}$ be a partition of $[n]$ with $|A_p| \geq 2$ for $p \in [r]$.
We say that $f : \{0, 1\}^n \rightarrow \ER$ is a {\it VCSP-quadratic function of type $\mathcal{A}$}
if $f$ is represented as
\begin{align}\label{eq:f}
f(x) :=
\begin{cases}
\displaystyle
\sum_{1 \leq i < j \leq n} a_{ij} x_i x_j + \sum_{1 \leq i \leq n} a_i x_i & \textrm{if $x \in \HP$},\\
+\infty & \textrm{otherwise}
\end{cases}
\end{align}
for some $a_i \in \R$ and $a_{ij} \in \ER$ such that $a_{ij} = +\infty$ for $i,j \in A_p$ ($p \in [r]$) and $a_{ij} < +\infty$ for $i \in A_p$ and $j \in A_q$ ($p,q \in [r]$, $p \neq q$).
We assume $a_{ij} = a_{ji}$ for distinct $i,j \in [n]$;
see~\eqref{eq:f instance}.

Suppose that a binary VCSP instance $F$ of the form~\eqref{eq:binary VCSP} is given,
where we assume $F_{pq} = F_{qp}$ for distinct $p,q \in [r]$.
The transformation of $F$ to $f$ based on \eqref{eq:D to {0,1}}
in Section~\ref{sec:intro} is formalized as follows.
Choose a partition $\mathcal{A} := \{A_1, A_2, \dots, A_r\}$ of $[n]$
with $|A_p| = n_p (= |D_p|)$ and identify $A_p$ with $D_p$.
Define
\begin{align*}
a_i &:= F_p(i) \qquad (i \in [n]),\\
a_{ij} &:=
\begin{cases}
F_{pq}(i, j) & (i \in A_p \textrm{ and } j \in A_q \textrm{ for some distinct }p,q \in [r]),\\
+\infty & (i,j \in A_p \textrm{ for some } p \in [r]).
\end{cases}
\end{align*}
Then the function $f$ in~\eqref{eq:f}
is a VCSP-quadratic function of type $\mathcal{A}$.

We introduce two functions that will serve as the M-convex summands of an M${}_2$-convex VCSP-quadratic function of type $\mathcal{A}$.
A function $h : \{0,1\}^n \rightarrow \ER$ is said to be {\it $\mathcal{A}$-linear}
if $h$ is a linear function on $\DF$,
that is,
if $h$ can be represented as $h(x) = \delta_\mathcal{A}(x) + \sum_{1 \leq i \leq n} u_i x_i + \gamma$
for some linear coefficient $(u_i)_{i \in [n]}$ and constant $\gamma \in \R$.
We use such $h$ as an M-convex summand.
As the other function, for technical reasons,
we use the following $\ell_X$ instead of $\overline{\ell}_X$ in~\eqref{eq:overline l};
the difference of
$\overline{\ell}_X$ and $2 \ell_X$ is linear.
For $X \subseteq [n]$, let $\ell_X : \{0,1\}^n \rightarrow \R$ be defined by
\begin{align}\label{eq:l_X def}
\ell_X(x) := \sum_{i,j \in X, i<j} x_ix_j. 
\end{align}

The following lemma guarantees the M-convexity of the two functions (like $f_1$ in~\eqref{eq:f_1 instance} and $f_2$ in~\eqref{eq:f_2 instance}) obtained in our algorithm.
Here a family $\mathcal{F} \subseteq 2^{[n]}$ is said to be {\it laminar}
if $X \subseteq Y$, $X \supseteq Y$, or $X \cap Y = \emptyset$ holds for all $X, Y \in \mathcal{F}$.
\begin{lem}\label{lem:M2 representation}
	\begin{description}
		\item[{\rm (1)}]
		An $\mathcal{A}$-linear function is M-convex.
		\item[{\rm (2)}]
		For any laminar family $\mathcal{L}$ and positive weight $c$ on $\mathcal{L}$, the function $\sum_{X \in \mathcal{L}} c(X)\ell_X$ on $\HP$ is M-convex.
	\end{description}
\end{lem}
\begin{proof}
	(1).
	An $\mathcal{A}$-linear function $h$ can be viewed as a linear function on the base family of a partition matroid with partition $\mathcal{A}$.
	Hence $h$ is M-convex.
	
	(2).
	We can see that
	the quadratic coefficient of $\sum_{X \in \mathcal{L}} c(X) \ell_X$ satisfies $a_{ij} + a_{kl} \geq \min \{ a_{ik}+a_{jl}, a_{il}+a_{jk} \}$
	for every distinct $i,j,k,l \in [n]$ (see also \cref{lem:anti-ultrametric} below).
	Hence, by \cite[Theorem~3.1]{DAM/I18} (or \cref{lem:Type}~(I) below),
	$\sum_{X \in \mathcal{L}} c(X) \ell_X$ on $\HP$ is M-convex.
\end{proof}

\cref{lem:M2 representation} gives a sufficient condition for the QR-M${}_2$-convexity of a VCSP-quadratic function $f$;
if $f$ can be represented as the sum of $\sum_{X \in \mathcal{L}} c(X) \ell_X$ on $\HP$ for some laminar $\mathcal{L}$ and a linear function on $\DF$,
then $f$ is QR-M${}_2$-convex.
Our representation theorem (\cref{thm:representation}) says that
this is also a necessary condition,
that is,
a QR-M${}_2$-convex VCSP-quadratic function is always representable as the sum of $\sum_{X \in \mathcal{L}} c(X)\ell_X$ on $\HP$ for some laminar $\mathcal{L}$ and a linear function on $\DF$.

A laminar family inducing the given QR-M${}_2$-convex VCSP-quadratic function
possesses some kind of uniqueness,
which ensures the validity of our proposed algorithm.
To describe the uniqueness in \cref{thm:representation},
we introduce an equivalence relation on functions:
\begin{itemize}
	\item For two functions $f$ and $f'$ on $\{0,1\}^n$,
	we say that $f$ and $f'$ are {\it $\mathcal{A}$-linear equivalent} (or $f \simeq f'$)
	if the difference between $f$ and $f'$ is a linear function on $\DF$,
	that is,
	$f + \delta_\mathcal{A} = f' + h$ holds for some $\mathcal{A}$-linear function $h$.
\end{itemize}

The $\mathcal{A}$-linear equivalence on $\ell_X$'s can be regarded as a combinatorial property on sets $X$
by using the following notations.
\begin{itemize}
	\item We say that a set $X \subseteq [n]$ {\it cuts} $A_p$ if both $X$ and $[n] \setminus X$ have a nonempty intersection with $A_p$, i.e., $\emptyset \neq (X \cap A_p) \neq A_p$.
	\item A set $X \subseteq [n]$ is called an {\it $\mathcal{A}$-cut}
	if $X$ cuts at least two elements in $\mathcal{A}$.
	\item For $X \subseteq [n]$, the {\it cutting support} of $X$, denote by $\ave{X}$, is defined as the union of $A_p$'s cut by $X$.
	That is,
	\begin{align}\label{eq:ave}
	\ave{X} := \bigcup \{A_p \in \mathcal{A} \mid \emptyset \neq (X \cap A_p) \neq A_p \}.
	\end{align}
\end{itemize}
\begin{lem}\label{lem:sim}
	\begin{description}
		\item[{\rm (1)}]
		For $X \subseteq [n]$,
		$\ell_X + \delta_\mathcal{A}$ is not $\mathcal{A}$-linear if and only if $X$ is an $\mathcal{A}$-cut.
		\item[{\rm (2)}]
		For two $\mathcal{A}$-cuts $X$ and $Y$,
		functions $\ell_X$ and $\ell_Y$ are $\mathcal{A}$-linear equivalent if and only if
		\begin{align}\label{eq:sim}
		\{\ \ave{X} \cap X,\ \ave{X} \setminus X\ \} = \{\ \ave{Y} \cap Y,\ \ave{Y} \setminus Y\ \},
		\end{align}
		that is,
		$X$ and $Y$ have the same cutting support and yield the same bipartition on it.
	\end{description}
\end{lem}
\begin{proof}
	As in~\eqref{eq:overline l} in the introduction, 
	define $\overline{\ell}_X : \{0,1\}^n \rightarrow \R$ by $\overline{\ell}_X(x) := \left( \sum_{i \in X} x_i \right)^2$.
	Then it holds
	$\ell_X \simeq \overline{\ell}_X/2$ by $x_i^2 = x_i$ for $i \in [n]$.
	Hence it suffices to show the statements for $\overline{\ell}_X$.
	As mentioned in~\eqref{eq:X + A_p} and~\eqref{eq:n - X},
	it holds (i) $\overline{\ell}_X \simeq \overline{\ell}_{X \cup A_p}$ if $X \cap A_p = \emptyset$, 
	and (ii) $\overline{\ell}_X \simeq \overline{\ell}_{[n] \setminus X}$.
	The former follows from  $\overline{\ell}_{X \cup A_p}(x) = \left( \sum_{i \in X} x_i + \sum_{i \in A_p} x_i \right)^2
	\simeq \left( \sum_{i \in X} x_i + 1 \right)^2
	\simeq \overline{\ell}_{X}(x)$, and
	the latter follows from $\overline{\ell}_{[n] \setminus X}(x) \simeq \left( r - \sum_{i \in X} x_i \right)^2
	\simeq \overline{\ell}_{X}(x)$.
	
	(Only-if part of (1)).
	Suppose that $X$ is not an $\mathcal{A}$-cut.
	Then $\ave{X} \subseteq A_p$ holds for some $A_p \in \mathcal{A}$.
	By (i), we may assume $X \subseteq A_p$.
	Then it holds $\overline{\ell}_{X}(x) = \left( \sum_{i \in X} x_i \right)^2 = \sum_{i \in X} x_i$ for all $x \in \DF$, implying that $\overline{\ell}_{X}$ is $\mathcal{A}$-linear.
	
	(If part of (2)).
	Suppose that (\ref{eq:sim}) holds.
	Then we can construct $Y$ from $X$ by repeating the transformation $X \mapsto [n] \setminus X$,
	$X \cup A_p$, or $X \setminus A_p$ for $A_p$ with $\ave{X} \cap A_p = \emptyset$.
	Hence
	$\overline{\ell}_X \simeq \overline{\ell}_Y$ by (i) and (ii) above.

	(If part of (1)).
	To detect the non-linearity,
	we consider the following four points $x^{su}, x^{sv}, x^{tu}, x^{tv} \in \DF$
	for distinct $s,t \in A_p$ and $u,v \in A_q$ with distinct $p,q \in [r]$:
	\begin{itemize}
		\item $x^{ij}_i = x^{ij}_j = 1$ for $i = s,t$ and $j = u,v$, and
		\item $x^{su}_i = x^{sv}_i = x^{tu}_i = x^{tv}_i$ for $i \in [n] \setminus (A_p \cup A_q)$.
	\end{itemize}
	Since  $x^{su} + x^{tv} = x^{sv} + x^{tu}$, the inequality
	$\overline{\ell}_X(x^{su}) + \overline{\ell}_X(x^{tv}) \neq \overline{\ell}_X(x^{sv}) + \overline{\ell}_X(x^{tu})$ 
	implies that $\overline{\ell}_X$ is not linear on the four points.
	Let $\kappa_X := \left(\overline{\ell}_X(x^{su}) + \overline{\ell}_X(x^{tv})\right) - \left(\overline{\ell}_X(x^{sv}) + \overline{\ell}_X(x^{tu})\right)$.
	By $\overline{\ell}_X(x^{ij}) = \left( |X \cap \{i\}| + |X \cap \{j\}| + k \right)^2$ with a constant $k$, we have
	\begin{align}\label{eq:non-linearity}
	\kappa_X =
	\begin{cases}
	2 & \textrm{if $X \cap \{s,t,u,v\} = \{s,u\}$ or $\{t,v\}$},\\
	-2 & \textrm{if $X \cap \{s,t,u,v\} = \{s,v\}$ or $\{t,u\}$},\\
	0 & \textrm{otherwise}.
	\end{cases}
	\end{align}
	If $X$ is an $\mathcal{A}$-cut,
	we can choose distinct $s,t \in A_p, u,v \in A_q$
	for distinct $A_p, A_q \subseteq \ave{X}$ such that $|X \cap \{s,t\}| = |X \cap \{u,v\}| = 1$,
	and it holds $\kappa_X \neq 0$.
	
	(Only-if part of (2)).
	This can be shown in a similar way as the proof of the if part of~(1).
	Suppose that \eqref{eq:sim} does not hold.
	Then we can choose distinct $s,t \in A_p, u,v \in A_q$ with $p \neq q$
	such that $\kappa_X \neq \kappa_Y$,
	which implies that $\overline{\ell}_X$ and $\overline{\ell}_Y$ are not
	${\cal A}$-linear equivalent.
	Indeed,
	by $\{ \ave{X} \cap X, \ave{X} \setminus X \} \neq \{ \ave{Y} \cap Y, \ave{Y} \setminus Y \}$,
	there are $A_p$ and $A_q$ cut by (say) $X$ such that $\{ (A_p \cup A_q) \cap X, (A_p \cup A_q) \setminus X \} \neq \{ (A_p \cup A_q) \cap Y, (A_p \cup A_q) \setminus Y \}$.
	Hence we can choose points $s \in A_p \cap X$, $t \in A_p \setminus X$, $u \in A_q \cap X$, and $v \in A_q \setminus X$ such that $Y \cap \{s,t,u,v\} \neq \{s,u\}$ and $Y \cap \{s,t,u,v\} \neq \{t,v\}$.
	Then~\eqref{eq:non-linearity} shows $\kappa_X \neq \kappa_Y$.
\end{proof}

According to \cref{lem:sim},
we introduce the equivalence relations on sets, families, and positive weights on families,
and also introduce the concept of {\it laminarizability} as follows.
\begin{itemize}
	\item For two $\mathcal{A}$-cuts $X$ and $Y$,
	we say that $X$ and $Y$ are {\it $\mathcal{A}$-equivalent} (or $X \sim Y$)
	if $X$ and $Y$ satisfy~\eqref{eq:sim}.
	That is, $X \sim Y$ if and only if $\ell_X \simeq \ell_Y$.
	\item The $\mathcal{A}$-equivalence relation is naturally extended to $\mathcal{A}$-cut families $\mathcal{F}, \mathcal{G}$ by:
	$\mathcal{F}$ and $\mathcal{G}$ are {\it $\mathcal{A}$-equivalent} (or $\mathcal{F} \sim \mathcal{G}$) if the set of the equivalence classes of all $\mathcal{A}$-cuts in $\mathcal{F}$ coincides with that of $G$.
	\item We define the $\mathcal{A}$-equivalence relation $\sim$ between a positive weight $c$ on $\mathcal{F}$ and a positive weight $d$ on $\mathcal{G}$ by:
	$c \sim d$ if $\mathcal{F} \sim \mathcal{G}$ and $\sum_{Y \in \mathcal{F} : Y \sim X} c(Y) = \sum_{Y \in \mathcal{G} : Y \sim X} d(Y)$ for all $X \subseteq [n]$,
	where $c(X) := 0$ (resp. $d(Y) := 0$) if $X \not\in \mathcal{F}$ (resp. if $Y \not\in \mathcal{G}$).
	It is clear, by the definition of $\sim$, that if $\mathcal{F} \sim \mathcal{G}$ and $c \sim d$,
	then $\sum_{X \in \mathcal{F}} c(X) \ell_X \simeq \sum_{X \in \mathcal{G}} d(X) \ell_X$.
	\item An $\mathcal{A}$-cut family $\mathcal{F}$ is said to be {\it laminarizable}
	if there is a laminar family $\mathcal{L}$ with $\mathcal{F} \sim \mathcal{L}$.
\end{itemize}

The formal description of our representation theorem is the following.
\begin{thm}\label{thm:representation}
	Let $f$ be a VCSP-quadratic function of type $\mathcal{A} = \{ A_1, A_2, \dots, A_r\}$.
	Then $f$ is QR-M${}_2$-convex if and only if
		there exist a laminarizable $\mathcal{A}$-cut family $\mathcal{F}$ and a positive weight $c$ on $\mathcal{F}$
		such that
		\begin{align}\label{eq:(L, c)}
			f \simeq \sum_{X \in \mathcal{F}} c(X) \ell_{X}.
		\end{align}
	In addition,
	$\mathcal{F}$ and $c$ are uniquely determined up to $\sim$.
\end{thm}
The proof of \cref{thm:representation} is given in \cref{subsec:chara,subsec:unique}.

\subsection{Two subproblems: \DECOMP\ and \LAM}\label{subsec:subproblems}
By \cref{thm:representation},
\TEST\ can be divided into the following two problems:
(i) if $f$ is QR-M${}_2$-convex,
then output a laminarizable $\mathcal{A}$-cut family $\mathcal{F}$ and a positive weight $c$ on $\mathcal{F}$ satisfying the equation~\eqref{eq:(L, c)},
and (ii) if the output $\mathcal{F}$ of (i) is laminarizable,
then find a laminar family $\mathcal{L}$ with $\mathcal{L} \sim \mathcal{F}$.
(i) and (ii) can be formulated as \DECOMP\ and \LAM,
respectively.
An $\mathcal{A}$-cut family $\mathcal{F}$ is said to be {\it non-redundant}
if no distinct $X, Y$ with $X \sim Y$ are contained in $\mathcal{F}$.
\begin{description}
	\item[\underline\DECOMP]
	\item[Given:] A VCSP-quadratic function $f$ of type $\mathcal{A}$.
	\item[Problem:]
	Either detect the non-QR-M${}_2$-convexity of $f$,
	or obtain some non-redundant $\mathcal{A}$-cut family $\mathcal{F}$ and positive weight $c$ on $\mathcal{F}$
	satisfying
	\begin{align}\label{eq:(L, c)inDEC}
	f \simeq \sum_{X \in \mathcal{F}} c(X) \ell_X.
	\end{align}
	In addition,
	in case where $f$ is QR-M${}_2$-convex,
	$\mathcal{F}$ is required to be laminarizable.
\end{description}
We emphasize that \DECOMP\ may possibly output
the decomposition~\eqref{eq:(L, c)inDEC}
even when the input $f$ is not QR-M${}_2$-convex.
However, if \DECOMP\ detects the non-QR-M${}_2$-convexity then we can conclude that the input $f$ is not QR-M${}_2$-convex.

\begin{description}
	\item[\underline\LAM]
	\item[Given:] A non-redundant $\mathcal{A}$-cut family $\mathcal{F}$.
	\item[Problem:] Determine whether
	$\mathcal{F}$ is laminarizable or not.
	If it is laminarizable,
	obtain a non-redundant laminar $\mathcal{A}$-cut family $\mathcal{L}$ with 
	$\mathcal{F} \sim \mathcal{L}$.
\end{description}

With these procedures, \TEST\ is solved as follows.
\begin{itemize}
	\item Suppose that $f$ is QR-M${}_2$-convex.
	First, by solving \DECOMP,
	we obtain a non-redundant laminarizable $\mathcal{A}$-cut family $\mathcal{F}$ and a positive weight $c$ on $\mathcal{F}$ satisfying~\eqref{eq:(L, c)inDEC} (and hence~\eqref{eq:(L, c)}).
	Then, by solving \LAM\ with $\mathcal{F}$ as its input, we obtain a non-redundant laminar $\mathcal{A}$-cut family $\mathcal{L}$ with $\mathcal{L} \sim \mathcal{F}$.
	Thus we can obtain two M-convex summands $f_1 := \sum_{X \in \mathcal{L}} c^*(X) \ell_X$ on $\HP$
	and $f_2 := f - \sum_{X \in \mathcal{L}} c^*(X) \ell_X$,
	where $c^* \sim c$.
	Such $c^*$ can easily be constructed
	as $c^*(X) := c(Y)$ for $X \in \mathcal{L}$ and $Y \in \mathcal{F}$ with $X \sim Y$.
	\item Suppose that $f$ is not QR-M${}_2$-convex.
	By solving \DECOMP,
	we can detect the non-QR-M${}_2$-convexity of $f$ or we obtain some $\mathcal{A}$-cut family $\mathcal{F}$, positive weight $c$ on $\mathcal{F}$, and $\mathcal{A}$-linear function $h$ that demonstrates~\eqref{eq:(L, c)inDEC}.
	In the former case, we are done.
	In the latter case,
	by solving \LAM\ with $\mathcal{F}$ as its input,
	we can detect the non-laminarizability of $\mathcal{F}$,
	which denies the QR-M${}_2$-convexity of $f$.
\end{itemize}
We devise an $O(rn^2)$-time algorithm for \DECOMP\ in \cref{sec:(i)} and an $O(n^4)$-time algorithm for \LAM\ in \cref{sec:(ii)}.
Thus we obtain \cref{thm:main}.

By \cref{lem:sim}~(2),
\LAM\ can be regarded as the problem of transforming a given family $\mathcal{F}$ to a laminar family
by repeating the following operation:
replace $X \in\mathcal{F}$ with $[n] \setminus X$,
$X \cup A_p$, or $X \setminus A_p$ with some $A_p$ satisfying $\ave{X} \cap A_p = \emptyset$.
Figure~\ref{fig:Laminarization} illustrates an example of the input (left) and an output (right) of \LAM.

\subsection{Proof of \cref{thm:representation}: Characterization}\label{subsec:chara}
In this subsection,
we prove the if-and-only-if part of \cref{thm:representation},
i.e.,
a VCSP-quadratic function $f$ of type $\mathcal{A}$ is QR-M${}_2$-convex if and only if
\eqref{eq:(L, c)} holds for some laminarizable $\mathcal{A}$-cut family $\mathcal{F}$ and positive weight $c$ on $\mathcal{F}$.

We first review fundamental facts about a general quadratic 
(not necessarily VCSP-quadratic) function
$g : \{0,1\}^n \rightarrow \ER$ 
represented as
\begin{align}
g(x_1,x_2,\dots,x_n) =
\begin{cases}
\displaystyle
\sum_{1 \leq i < j\leq n}a_{ij}x_ix_j + \sum_{1 \leq i \leq n} a_i x_i & \text{if $x \in \HP$},\\
+\infty & \text{otherwise},\label{eq:g}
\end{cases}
\end{align}
where $r \in \mathbf{Z}$ with $r \geq 2$, $a_i \in \R$, and $a_{ij}= a_{ji} \in \ER$.
We assume the following regularity condition (R) for $g$.
\begin{itemize}
	\item[(R):] For all $i \in [n]$,
	there is $x = (x_1, x_2, \dots, x_n)$ such that $g(x) < +\infty$ and $x_i = 1$.
\end{itemize}
Denote the indicator function of $\dom{g}$ by $\delta_{g}$,
which is defined as $\delta_g(x) := 0$ for $x \in \dom{g}$ and $\delta_g(x) := +\infty$ for $x \not\in \dom{g}$.

Let $G_g^\infty$ be the graph on node set $[n]$ such that edge $\{i,j\}$ ($i \neq j$) exists if and only if $a_{ij} = +\infty$.
Define $m(G_g^\infty)$ as the number of connected components of $G_g^\infty$.
A connected component with at least one edge is said to be {\it non-isolated}.
The number of non-isolated connected components of $G_g^\infty$ will be denoted by $m^* = m^*(G_g^\infty)$.
Let $B_1, B_2, \dots, B_{m^*}$ be the node sets of the non-isolated connected components of $G_g^\infty$.

Then the M-convexity of $g$ is characterized by the following lemma, which is a refinement of 
the results of~\cite{JJIAM/HM04} and~\cite{AAM/MS04}.
\begin{lem}[{\cite[Theorem~3.1]{DAM/I18}}]\label{lem:Type}
	A function $g$ of the form~{\rm \eqref{eq:g}} satisfying condition {\rm (R)} is M-convex if and only if
	each connected component of $G_g^\infty$ is a complete graph
	and one of the following conditions {\rm (I)}, {\rm (II)}, and {\rm (III)} holds:
	\begin{description}
		\item[(I):] $m(G_g^\infty) \geq r + 2$ and
		\begin{align}\label{eq:Type I}
		a_{ij} + a_{kl} \geq \min \{ a_{ik}+a_{jl}, a_{il}+a_{jk} \}
		\end{align}
		holds for every distinct $i,j,k,l \in [n]$.
		\item[(II):] $m(G_g^\infty) = r + 1$ and
		\begin{align}\label{eq:Type II}
		a_{ij} + a_{kl} = a_{il} + a_{jk}
		\end{align}
		holds for every $p \in [m^*]$, distinct $i,k \in B_p$, and distinct $j,l \in [n] \setminus B_p$.
		\item[(III):] $m(G_g^\infty)=r$ and
		\begin{align}\label{eq:Type III}
		a_{ij} + a_{kl} = a_{il} + a_{jk}
		\end{align}
		holds for every distinct $p,q \in [m^*]$, distinct $i,k \in B_p$, and distinct $j,l \in B_q$.
	\end{description}
	Moreover,
	{\rm (II)} or {\rm (III)} holds if and only if
	$g$ is represented as $g(x) = \delta_g(x) + \sum_{i \in [n]} u_ix_i + \gamma$
	for some $u \in \R^n$ and $\gamma \in \R$.
\end{lem}

We say that $(a_{ij})_{i,j \in [n]}$ satisfies the {\it anti-tree metric property} if \eqref{eq:Type I} holds,
and that $(a_{ij})_{i,j \in [n]}$ satisfies the {\it anti-ultrametric property} if
\begin{align}\label{eq:anti-ultrametric}
	a_{ij} \geq \min \{ a_{ik}, a_{jk} \}
\end{align}
holds for all distinct $i,j,k \in [n]$.
It is known~\cite{book/DezaLaurent97} that the anti-ultrametric property is stronger than the anti-tree metric property~\eqref{eq:Type I}.
The anti-ultrametric property is related with a laminar family as in \cref{lem:anti-ultrametric} below.
A {\it subpartition} of $[n]$ is a family of disjoint nonempty subsets of $[n]$.
For a subpartition $\mathcal{B}$,
a family $\mathcal{L}$ is said to be {\it $\mathcal{B}$-laminar} if $\mathcal{L}$ is laminar and $B \subsetneq X$ or $B \cap X = \emptyset$ holds for each $B \in \mathcal{B}$ and $X \in \mathcal{L}$,
that is, if $\mathcal{L}$ does not intersect with $\mathcal{B}$, $\mathcal{L} \cup \mathcal{B}$ is laminar, and each $B \in \mathcal{B}$ is minimal in $\mathcal{L} \cup \mathcal{B}$.
\begin{lem}[{\cite[Lemma~8]{DO/IMZ18}}]\label{lem:anti-ultrametric}
	Let $g$ be a quadratic function with a coefficient $(a_{ij})_{i,j \in [n]}$,
	and $\mathcal{B}$ be the family of the node sets of the non-isolated connected components of $G_g^\infty$.
	Then
	$(a_{ij})_{i,j \in [n]}$ satisfies the anti-ultrametric property if and only if
	$a_{ij}$ can be represented as
	\begin{align}\label{eq:a_ij laminar}
	a_{ij} =
	\begin{cases}
	+\infty & {\rm if\ } i,j \in B {\rm \ for\ some\ } B \in  \mathcal{B},\\
	\sum \{ c(L) \mid L \in \mathcal{L};\ i,j \in L \} + \alpha^* & {\rm otherwise}
	\end{cases}
	\end{align}
	for some $\mathcal{B}$-laminar family $\mathcal{L} \subseteq 2^{[n]} \setminus [n]$ and some positive weight $c$ on $\mathcal{L}$,
	where $\alpha^* := \min_{i,j \in [n]} a_{ij}$.
\end{lem}
\cref{lem:anti-ultrametric} follows from Lemma~8 of~\cite{DO/IMZ18}
by relating $\mathcal{B}$ to the set of complete graphs for $\alpha = +\infty$
and relating $\mathcal{L}$ to the union of the set of complete graphs for $\alpha < +\infty$,
where $\alpha$ is a parameter appearing in Lemma~8 of~\cite{DO/IMZ18}.

The following is a variation of a well-known technique (the Farris transform)
in phylogenetics~\cite{book/SempleSteel03}
to transform a tree metric to an ultrametric,
and is implied by the validity of Algorithm~I described in Section~4.1 of~\cite{DAM/I18}.
In particular, Steps~1 and~2 of Algorithm~I correspond to the following.
\begin{lem}[{\cite{DAM/I18}}]\label{lem:Type I anti-ultra}
	Suppose that $(a_{ij})_{i,j \in [n]}$ satisfies the anti-tree metric property.
	Let $\alpha^* := \min_{i,j \in [n]} a_{ij}$ and $b_{k} := \min_{j \in [n]} a_{kj} - \alpha^*$ for $k \in [n]$.
	Then $\min_{j \in [n]} a_{ij} = \alpha^*$ holds for all $i \in [n]$, and $(a_{ij} - b_i - b_j)_{i,j \in [n]}$ satisfies the anti-ultrametric property.
\end{lem}

We are now ready to show the characterization part of \cref{thm:representation}.
Note that, by the definition of laminarizability, \eqref{eq:(L, c)} holds for some laminarizable family $\mathcal{F}$
if and only if \eqref{eq:(L, c)} holds for some laminar family $\mathcal{L}$.
\begin{prop}[The characterization part of \cref{thm:representation}]\label{prop:Type I + Type III}
	Let $f$ be a VCSP-quadratic function of type $\mathcal{A}$.
	Then $f$ is QR-M${}_2$-convex if and only if
	\begin{align}\label{eq:f_1 desired}
	f \simeq \sum_{X \in \mathcal{L}} c(X) \ell_{X}
	\end{align}
	for some laminar $\mathcal{A}$-cut family $\mathcal{L}$ and positive weight $c$ on $\mathcal{L}$.
\end{prop}
\begin{proof}
	For a subpartiton $\mathcal{B}$,
	define $\delta_\mathcal{B} : \{0,1\}^n \rightarrow \ER$ by
	$\delta_\mathcal{B}(x) := 0$ if $x \in \HP$ and $\sum_{i \in B} x_i \leq 1$ for each $B \in \mathcal{B}$,
	and $\delta_\mathcal{B}(x) := +\infty$ otherwise.
	Then, by \cref{lem:Type},
	$\delta_\mathcal{B}$ is an M-convex function that can be represented as $\delta_\mathcal{B}(x) = \sum_{B \in \mathcal{B}} \sum_{i,j \in B, i < j} \infty \cdot x_ix_j$ on $\HP$.
	The set of non-isolated connected components of $G_{\delta_\mathcal{B}}^\infty$ is equal to $\mathcal{B}$.
	We say that a function $f$ is of {\it Type~I}, {\it Type~II}, or {\it Type~III}
	if $m(G_f^\infty) \geq r + 2$, $m(G_f^\infty) = r + 1$, or $m(G_f^\infty) = r$ holds, respectively (cf., \cref{lem:Type}).
	
	(If part).
	By the $\mathcal{A}$-linear equivalence,
	$f$ is represented as $f = \sum_{X \in \mathcal{L}} c(X) \ell_{X} + h$
	for some $\mathcal{A}$-linear function $h$.
	By \cref{lem:M2 representation}~(1) and~(2),
	the functions $h$ and $\sum_{X \in \mathcal{L}} c(X) \ell_{X}$ on $\HP$ are M-convex.
	Hence $f$ is QR-M${}_2$-convex.
	
	(Only-if part).
	Let $f_1, f_2 : \{0,1\}^n \rightarrow \ER$ be any quadratic M-convex functions with $f = f_1 + f_2$.
	Since $f$ satisfies condition (R),
	$f_1$ and $f_2$ also satisfy condition (R) by $\dom{f} = \dom{f_1} \cap \dom{f_2}$.
	Let $\mathcal{B}_1$ and $\mathcal{B}_2$ be the sets of non-isolated connected components of $G_{f_1}^\infty$ and $G_{f_2}^\infty$, respectively.
	Since $f_1$ and $f_2$ are M-convex,
	each member of $\mathcal{B}_1$ (resp. $\mathcal{B}_2$) induces a complete graph in $G_{f_1}^\infty$ (resp. $G_{f_2}^\infty$) by \cref{lem:Type}.
	Hence $\dom{f_1} = \dom{\delta_{\mathcal{B}_1}}$ and $\dom{f_2} = \dom{\delta_{\mathcal{B}_2}}$ hold.
	Note that $\dom{f} = \dom{\delta_\mathcal{A}} = \dom{(\delta_{\mathcal{B}_1} + \delta_{\mathcal{B}_2})}$.

	Here the following claim holds.
	\begin{cl*}
		There exist quadratic M-convex functions $f_1$ and $f_2$ such that $f = f_1 + f_2$, $\mathcal{B}_1 \cap \mathcal{B}_2 = \emptyset$, and $\mathcal{B}_1 \cup \mathcal{B}_2 = \mathcal{A}$.
	\end{cl*}
	\begin{proof}[Proof of Claim]
		Let $f_1, f_2 : \{0,1\}^n \rightarrow \ER$ be any quadratic M-convex functions with $f = f_1 + f_2$.
		We first show that if $\mathcal{B}_1$ and $\mathcal{B}_2$ satisfy (i) for each $B \in \mathcal{B}_1 \cup \mathcal{B}_2$
		there is $A \in \mathcal{A}$ such that $B \subseteq A$, and
		(ii) each $A \in \mathcal{A}$ belongs to $\mathcal{B}_1 \cup \mathcal{B}_2$ (i.e., $\mathcal{A} \subseteq \mathcal{B}_1 \cup \mathcal{B}_2$),
		then Claim holds.
		
		Suppose that (i) and (ii) hold,
		and that some $B \in \mathcal{B}_1$ violates $\mathcal{B}_1 \cap \mathcal{B}_2 = \emptyset$ or $\mathcal{B}_1 \cup \mathcal{B}_2 = \mathcal{A}$,
		i.e., $B \in \mathcal{B}_2$ or $B \not\in \mathcal{A}$.
		Then we can modify $f_1$ so that $f_1$ is M-convex with $f = f_1 + f_2$ and $\dom{f_1}$ is changed from $\dom{\delta_{\mathcal{B}_1}}$ to $\dom{\delta_{\mathcal{B}_1 \setminus \{B\}}}$ as follows.
		
		By (i) and (ii),
		there is $A \in \mathcal{A} \cap \mathcal{B}_2$ such that $B \subseteq A$.
		If $f_1$ is of Type~II or~III,
		then $f_1 \simeq \delta_{\mathcal{B}_1}$ by \cref{lem:Type}.
		Hence we have
		\begin{align*}
		f_1 + f_2 \simeq \delta_{\mathcal{B}_1} + f_2 = \delta_{\mathcal{B}_1 \setminus \{B\}} + f_2,
		\end{align*}
		where the second equality follows from $\dom{(\delta_{\mathcal{B}_1} + f_2)} = \dom{(\delta_{\mathcal{B}_1 \setminus \{B\}} + f_2)}$ by $B \subseteq A$ and $A \in \mathcal{B}_2$.
		Thus we can modify $f_1$ so that $f_1$ is M-convex with $f = f_1 + f_2$ and $\dom{f_1} = \dom{\delta_{\mathcal{B}_1 \setminus \{B\}}}$.
		If $f_1$ is of Type~I,
		then, by \cref{lem:Type}~(I) and \cref{lem:Type I anti-ultra}, the quadratic coefficient of $f_1$ are represented as $(a_{ij}^1 + b_i + b_j)_{i,j \in [n]}$,
		where $b_i \in \R$ and $(a_{ij}^1)$ satisfies the anti-ultrametric property.
		By modifying $a_{ij}^1 (= +\infty)$ to $M$ for $i,j \in B_1$ with a sufficiently large $M$,
		we have $\dom{f_1} = \dom{\delta_{\mathcal{B}_1 \setminus \{B_1\}}}$
		and the value of $f_1 + f_2$ does not change.
		Furthermore $(a_{ij}^1)$ still satisfies the anti-ultrametric property.
		Hence $f_1$ is M-convex.
		Thus we can modify $f_1$ so that $f_1$ is M-convex with $f = f_1 + f_2$ and $\dom{f_1} = \dom{\delta_{\mathcal{B}_1 \setminus \{B\}}}$.
		By repeating the above modification for $f_1$ or $f_2$,
		we obtain the $f_1$ and $f_2$ in Claim.
		
		We finally show that (i) and (ii) hold.
		
		(i). We can easily see that,
		for every $i,j$ with $a_{ij} < + \infty$ (i.e., $i \in A_p$ and $j \in A_q$ for some distinct $p,q$),
		there is $x \in \dom{f}$ such that $x_i = x_j = 1$.
		Hence, for such $i,j$,
		there is no $B \in \mathcal{B}_1 \cup \mathcal{B}_2$ satisfying $i,j \in B$.
		Therefore we obtain (i).
		
		(ii).
		Let $E_\mathcal{A}$ and $E_\mathcal{B}$ be the edge set of $G_{\delta_\mathcal{A}}^\infty$ and of $G_{\delta_{\mathcal{B}_1} + \delta_{\mathcal{B}_2}}^\infty$,
		respectively.
		That is, $\{i,j\} \in E_\mathcal{A}$ (resp. $\{i,j\} \in E_\mathcal{B}$) if and only if $i, j \in A$ for some $A \in \mathcal{A}$ (resp. $i,j \in B$ for some $B \in \mathcal{B}_1 \cup \mathcal{B}_2$).
		By (i), we have $E_\mathcal{A} \supseteq E_\mathcal{B}$.
		Suppose, to the contrary, that $E_\mathcal{A} \supsetneq E_\mathcal{B}$.
		Then there is $\{i,j\}$ such that $i,j \in A_p$ for some $p$ and $\{i,j\} \not\in E_\mathcal{B}$.
		Let $x \in \{0,1\}^n$ be a 0-1 vector such that $x_i = x_j = 1$, $\sum_{i \in [n]} x_i = r$, and $\sum_{i \in A_q} x_i \leq 1$ for each $q$ distinct from $p$.
		Since $E_\mathcal{A} \supseteq E_\mathcal{B}$,
		we have $x \in \dom{(\delta_{\mathcal{B}_1} + \delta_{\mathcal{B}_2})}$,
		whereas $x \not\in \dom{\delta_{\mathcal{A}}}$.
		This contradicts $\dom{\delta_\mathcal{A}} = \dom{(\delta_{\mathcal{B}_1} + \delta_{\mathcal{B}_2})}$,
		and hence $E_\mathcal{A} = E_\mathcal{B}$ holds.
		Therefore we obtain (ii).
		
		This completes the proof of Claim.
	\end{proof}
	
	By Claim,
	we can take quadratic M-convex functions $f_1$ and $f_2$ satisfying $f = f_1 + f_2$, $\dom{f_1} = \dom{\delta_{\mathcal{B}_1}}$, and $\dom{f_2} = \dom{\delta_{\mathcal{B}_2}}$,
	where $\mathcal{B}_1 \cap \mathcal{B}_2 = \emptyset$ and $\mathcal{B}_1 \cup \mathcal{B}_2 = \mathcal{A}$.
	In the following,
	we show that $f = f_1 + f_2$ satisfies~\eqref{eq:f_1 desired} with some laminar $\mathcal{A}$-cut family $\mathcal{L}$ and positive weight $c$ on $\mathcal{L}$
	for each of the three cases:
	(i) both $f_1$ and $f_2$ are of Type~II or~III,
	(ii) $f_1$ is of Type~I and $f_2$ is of Type~II or~III,
	and (iii) both $f_1$ and $f_2$ are of Type~I.

	(i).
	By \cref{lem:Type}~(II) or~(III),
	we have $f_1 \simeq \delta_{\mathcal{B}_1} \simeq 0$ and $f_2 \simeq \delta_{\mathcal{B}_2} \simeq 0$.
	Hence it holds that $f = f_1 + f_2 \simeq 0$.
	Thus we obtain~\eqref{eq:f_1 desired} with $\mathcal{L} = \emptyset$.

	(ii).
	Suppose that $f_1$ is represented as $f_1(x) = \sum_{i<j} a_{ij} x_i x_j$ on $\HP$.
	Note that $a_{ij}$ is not necessarily finite.
	We can assume that $(a_{ij})$ satisfies the anti-ultrametric property and $\min_{i,j} a_{ij} = 0$.
	Indeed,
	by \cref{lem:Type I anti-ultra},
	$(a_{ij} - b_i - b_j - \alpha^*)_{i,j \in [n]}$ satisfies the anti-ultrametric property and $\min_{i,j} (a_{ij} - b_i - b_j - \alpha^*) = 0$ for some $b_i$ $(i \in [n])$ and $\alpha^* \in \R$.
	Hence
	\begin{align*}
	f_1(x) &= \sum_{i<j} (a_{ij} - b_i - b_j - \alpha^*) x_i x_j + \sum_{i}(r-1)b_i x_i + \frac{r(r-1)\alpha^*}{2}\\
	&\simeq \sum_{i<j} (a_{ij} - b_i - b_j - \alpha^*) x_i x_j
	\end{align*}
	on $\dom{\delta_\mathcal{A}}$.
	Thus we can redefine $a_{ij} \leftarrow a_{ij} - b_i - b_j - \alpha^*$ for distinct $i,j \in [n]$
	to satisfy the anti-ultrametric property and $\min_{i,j} a_{ij} = 0$.

	Since $(a_{ij})$ satisfies the anti-ultrametric property,
	by \cref{lem:anti-ultrametric},
	there are a $\mathcal{B}_1$-laminar family $\mathcal{L}_1$ and a positive weight $c_1$ on $\mathcal{L}_1$ representing $(a_{ij})$ as~\eqref{eq:a_ij laminar}.
	Hence it holds that
	\begin{align}
	f_1(x) &= \sum_{L \in \mathcal{L}_1} c_1(L) \sum_{i,j \in L, i<j} x_ix_j + \delta_{\mathcal{B}_1}(x)\notag\\
	&= \sum_{L \in \mathcal{L}_1} c_1(L) \ell_L(x) + \delta_{\mathcal{B}_1}(x)\notag\\
	&\simeq \sum_{L \in \mathcal{L}_1, L \textrm{:$\mathcal{A}$-cut}} c_1(L) \ell_L(x) + \delta_{\mathcal{B}_1}(x)\label{eq:f_1 L_1},
	\end{align}
	where the equivalence follows from \cref{lem:M2 representation}~(1).
	Let $\hat{A}_1 := \bigcup_{A \in \mathcal{B}_1} A$ be the subset of $[n]$ corresponding to $\mathcal{B}_1$.
	By the $\mathcal{B}_1$-laminarity of $\mathcal{L}_1$, every $L \in \mathcal{L}_1$ satisfies $L \supseteq B$ or $L \cap B = \emptyset$ for each $B \in \mathcal{B}_1 \subseteq \mathcal{A}$.
	Hence, by \cref{lem:sim}~(1),
	\begin{align}\label{eq:hat{A}_1}
	\ell_L \simeq \ell_{L \setminus \hat{A}_1} \qquad (L \in \mathcal{L}).
	\end{align}
	By combining~\eqref{eq:f_1 L_1} and~\eqref{eq:hat{A}_1},
	we obtain
	\begin{align}\label{eq:f_1 A}
	f_1 \simeq \sum_{L \in \mathcal{L}_1^*}c_1^*(L) \ell_L,
	\end{align}
	where
	$\mathcal{L}_1^* := \{L \setminus \hat{A}_1 \mid L \in \mathcal{L}_1,\ L :\textrm{$\mathcal{A}$-cut} \}$ and $c_1^*(L) := \sum \{ c_1(L^*) \mid L^* \setminus \hat{A}_1 = L \}$.
	Note that $\mathcal{L}_1^*$ is a laminar $\mathcal{A}$-cut family and $c_1^*$ is an aggregation of $c_1$.

	On the other hand,
	by \cref{lem:Type}~(II) or~(III), it holds $f_2 \simeq 0$.
	Hence, by~\eqref{eq:f_1 A}, it holds that
	\begin{align*}
	f = f_1 + f_2 \simeq \sum_{L \in \mathcal{L}_1^*}c_1^*(L) \ell_L.
	\end{align*}
	Thus, by the laminarity of $\mathcal{A}$-cut family $\mathcal{L}_1^*$,
	we obtain~\eqref{eq:f_1 desired} with $\mathcal{L} = \mathcal{L}^*$ and $c = c_1^*$.
	
	(iii).
	By the same argument as in~(ii),
	$f_1$ satisfies~\eqref{eq:f_1 A}
	and
	$f_2$ satisfies
	\begin{align}\label{eq:f_2 A}
	f_2 \simeq \sum_{L \in \mathcal{L}_2^*}c_2^*(L) \ell_L,
	\end{align}
	where $\hat{A}_2 := \bigcup_{A \in \mathcal{B}_2} A$, $\mathcal{L}_2^* := \{ L \setminus \hat{A}_2 \mid L \in \mathcal{L}_2,\ L :\textrm{$\mathcal{A}$-cut}\}$, and $c_2^*(L) := \sum \{ c_2(L^*) \mid L^* \setminus \hat{A}_2 = L \}$
	for a $\mathcal{B}_2$-laminar family $\mathcal{L}_2$ and a positive weight $c_2$ on $\mathcal{L}_2$.
	Note that $\mathcal{L}_2^*$ is a laminar $\mathcal{A}$-cut family.
	We have $\hat{A}_2 = [n] \setminus \hat{A}_1$ by $\mathcal{B}_1 \cap \mathcal{B}_2 = \emptyset$ and $\mathcal{B}_1 \cup \mathcal{B}_2 = \mathcal{A}$.
	
	By adding~\eqref{eq:f_1 A} and~\eqref{eq:f_2 A}, it holds
	\begin{align*}
	f_1 + f_2 \simeq \sum_{L \in \mathcal{L}_1^*}w_1^*(L) \ell_L + \sum_{L \in \mathcal{L}_2^*}w_2^*(L) \ell_L.
	\end{align*}
	Hence we obtain~\eqref{eq:f_1 desired} with $\mathcal{L} = \mathcal{L}_1^* \cup \mathcal{L}_2^*$ and $c = c_1 + c_2$,
	where $(c_1 + c_2)(L) = c_1(L)$ for $L \in \mathcal{L}_1^*$ and $(c_1 + c_2)(L) = c_2(L)$ for $L \in \mathcal{L}_2^*$.
	Here $\mathcal{L}_1^* \cup \mathcal{L}_2^*$ is a laminar $\mathcal{A}$-cut family.
	Indeed, $\mathcal{L}_1^*$ and $\mathcal{L}_2^*$ are laminar $\mathcal{A}$-cut families,
	and $L_1 \cap L_2 = \emptyset$ holds for all $L_1 \in \mathcal{L}_1^*$ and $L_2 \in \mathcal{L}_2^*$
	by $L_1 \subseteq [n] \setminus \hat{A}_1$ and $L_2 \subseteq [n] \setminus \hat{A}_2 = \hat{A}_1$.

	This completes the proof of \cref{prop:Type I + Type III}.
\end{proof}

\subsection{Proof of \cref{thm:representation}: Uniqueness}\label{subsec:unique}
In this subsection, we prove the uniqueness of $\mathcal{F}$ and $c$ up to the $\mathcal{A}$-equivalence in \cref{thm:representation}.
Let $f$ be a VCSP-quadratic function of type $\mathcal{A}$.
We denote by $\overline{\DF}$ the convex hull of $\DF$,
i.e., $\overline{\DF} = \{ x \in [0,1]^n \mid \sum_{i \in A_p} x_i = 1 \textrm{ for all }p \in [r]\}$.
The {\it convex closure} $\overline{f} : \overline{\DF} \rightarrow \R$ of $f$ is the maximum convex function satisfying $\overline{f}(x) = f(x)$ for $x \in \DF$,
which is given by
\begin{align*}
	\overline{f}(x) := \sup \left\{ \sum_{1 \leq i \leq n} u_ix_i + \gamma\ \middle|\ u \in \R^n,\ \gamma \in \R,\ f(y) \geq \sum_{1 \leq i \leq n} u_iy_i + \gamma \quad (y \in \DF) \right\}.
\end{align*}

We first give another representation of $\ell_X$ up to the $\mathcal{A}$-linear equivalence.
For an $\mathcal{A}$-cut $X$,
define
\begin{align*}
\alpha(X) &:= \textrm{the number of elements $A_p \in \mathcal{A}$ with $X \supseteq A_p$},\\
\beta(X) &:= \textrm{the number of elements $A_p \in \mathcal{A}$ with $X \cap A_p \neq \emptyset$}.
\end{align*}
Note that, for any $x \in \DF$ with $\sum_{i \in X} x_i = s$,
it holds $\sum_{i \in \ave{X} \cap X} x_i = s - \alpha(X)$ and $\sum_{i \in \ave{X} \setminus X} x_i = \beta(X) - s$.
\begin{lem}\label{lem:l_X another}
	For an $\mathcal{A}$-cut $X$,
	it holds
	\begin{align}\label{eq:l_L quadratic}
	\ell_X(x) \simeq \frac{1}{2} \sum_{\alpha(X) < k < \beta(X)}\left|k - \sum_{i \in X} x_i\right|.
	\end{align}
\end{lem}
\begin{proof}
	For the left-hand side of~\eqref{eq:l_L quadratic},
	it holds $\ell_X \simeq \left(\ell_{\ave{X} \cap X} + \ell_{\ave{X} \setminus X}\right)/2$ by \cref{lem:sim}~(2).
	For the right-hand side of~\eqref{eq:l_L quadratic},
	we can see that
	\begin{align}\label{eq:l_X another}
	\sum_{\alpha(X) < k < \beta(X)}\left|k - \sum_{i \in X} x_i\right| = \ell_{\ave{X} \cap X}(x) + \ell_{\ave{X} \setminus X}(x) \qquad (x \in \DF),
	\end{align}
	and this implies~\eqref{eq:l_L quadratic}.
	Here~\eqref{eq:l_X another} can be established as follows.
	For $x \in \DF$ with $\sum_{i \in \ave{X} \cap X} x_i = s$,
	we have
	\begin{align*}
	\sum_{\alpha(X) < k < \beta(X)}\left|k - \sum_{i \in X} x_i\right|
	&= \sum_{\alpha(X) + 1 \leq k \leq s}(s - k) + \sum_{s \leq k \leq \beta(X) -1} (k-s)\\
	&= \frac{1}{2}\left( (s - \alpha(X))(s - \alpha(X) -1) + (\beta(X) - s)(\beta(X) - s -1) \right).
	\end{align*}
	On the other hand,
	by $\sum_{i \in \ave{X} \cap X} x_i = s - \alpha(X)$ and $\sum_{i \in \ave{X} \setminus X} x_i = \beta(X) - s$,
	we have
	\begin{align*}
	\ell_{\ave{X} \cap X}(x) + \ell_{\ave{X} \setminus X}(x) &= \binom{s - \alpha(X)}{2} + \binom{\beta(X)- s}{2}\\
	&= \frac{1}{2}\left( (s - \alpha(X))(s - \alpha(X) -1) + (\beta(X) - s)(\beta(X) - s -1) \right).
	\end{align*}
\end{proof}

Suppose that $f$ is an M${}_2$-convex function.
By \cref{prop:Type I + Type III} and \cref{lem:l_X another},
$f$ is representable as
\begin{align*}
	f(x) = \sum_{L \in \mathcal{L}} \frac{c(L)}{2} \sum_{\alpha(X) < k < \beta(X)}\left|k - \sum_{i \in X} x_i\right| + \sum_{1 \leq i \leq n} u_i x_i + \gamma \qquad (x \in \DF)
\end{align*}
for some laminar $\mathcal{A}$-cut family $\mathcal{L}$, positive weight $c$ on $\mathcal{L}$,
linear coefficient $u \in \R^n$, and constant $\gamma \in \R$.
Then $\overline{f}$ is explicitly written as follows.
\begin{lem}\label{lem:overline f = overline l}
	\begin{align}\label{eq:overline f = overline l}
		\overline{f}(x) = \sum_{L \in \mathcal{L}} \frac{c(L)}{2} \sum_{\alpha(X) < k < \beta(X)}\left|k - \sum_{i \in X} x_i\right| +  \sum_{1 \leq i \leq n} u_i x_i + \gamma \qquad (x \in \overline{\DF}).
	\end{align}
\end{lem}
\begin{proof}
	We denote by $\hat{f}$ the right-hand side of \eqref{eq:overline f = overline l}.
	It is clear that $f(x) = \hat{f}(x)$ for $x \in \dom{f}$.
	Since $\hat{f}$ is piecewise linear and convex,
	$\hat{f}(z) \leq \overline{f}(z)$ for $z \in \overline{\DF}$
	by the definition of $\overline{f}$.
	Thus it suffices to show $\hat{f}(z) \geq \overline{f}(z)$ for $z \in \overline{\DF}$.
	
	Take any $z \in \overline{\DF}$.
	Then $z$ satisfies the following system of inequalities and equations for some integers $k_L$ for all $L \in \mathcal{L}$:
	\begin{align}
		&0 \leq z_i \leq 1 \qquad (i \in [n]),\label{eq:convX1}\\
		&\sum_{i \in A_p} z_i = 1 \qquad (p \in [r]),\label{eq:convX2}\\
		&k_L - 1 \leq \sum_{i \in L} z_i \leq k_L \qquad (L \in \mathcal{L})\label{eq:x_i = k_L}.
	\end{align}
	
	The coefficient matrix $M$ of the system~\eqref{eq:convX1}--\eqref{eq:x_i = k_L}
	is totally unimodular.
	Indeed, let $M'$ be the $n \times (|\mathcal{L}| + r)$ matrix whose columns are the characteristic vectors of the members of $\mathcal{L}$ and $\{ A_1, A_2, \dots, A_r \}$.
	$M$ is represented as $M = \left( I\  -I\  M'\  -M'\right)^\top$,
	where $I$ is the $n \times n$ identity matrix.
	Since $\mathcal{L}$ and $\{ A_1, A_2, \dots, A_r \}$ are laminar,
	$M'$ is totally unimodular~\cite{CSTA/E70}; see also~\cite[Theorem 41.11]{book/Schrijver03}.
	Thus $M$ is also totally unimodular.

	Let $P$ be the polyhedron defined by the system~\eqref{eq:convX1}--\eqref{eq:x_i = k_L}.
	Then $P$ is an integer polyhedron by the total unimodularity of $M$.
	Hence all extreme points $y_i$ of $P$ belong to $\DF$.
	By $z \in P$,
	we have $z = \sum_i \lambda_i y_i$ for some coefficients $\lambda_i$ of a convex combination.
	Therefore $\hat{f}(z) = \sum_i \lambda_i \hat{f}(y_i) = \sum_i \lambda_i f(y_i)$ holds,
	where the first equality follows from the linearity of $\hat{f}$ on $P$.
	Since $f(y_i) = \overline{f}(y_i)$ and $\overline{f}$ is convex,
	we obtain $\sum_i \lambda_i f(y_i) = \sum_i \lambda_i \overline{f}(y_i) \geq \overline{f}(z)$,
	and hence $\hat{f}(z) \geq \overline{f}(z)$.
\end{proof}

We are ready to show the uniqueness part of \cref{thm:representation}.
Suppose that $f$ is QR-M${}_2$-convex.
Recall that, by \cref{prop:Type I + Type III} and \cref{lem:l_X another},
$f$ is representable as
\begin{align*}
f(x) = \sum_{L \in \mathcal{L}} \frac{c(L)}{2} \sum_{\alpha(X) < k < \beta(X)}\left|k - \sum_{i \in X} x_i\right| + h
\end{align*}
for some laminar $\mathcal{A}$-cut family $\mathcal{L}$, positive weight $c$ on $\mathcal{L}$,
and $\mathcal{A}$-linear function $h$.
Furthermore we can assume that $\mathcal{L}$ is non-redundant.
By \cref{lem:overline f = overline l},
the set of nondifferentiable points of $\overline{f}$ (with respect to the set of relative interior points of $\overline{\DF}$) is given by
\begin{align*}
	\bigcup_{L \in \mathcal{L},\ \alpha(L) < k < \beta(L)} \left\{ x \in \DF \ \middle|\ \sum_{i \in L} x_i = k \right\} =: P(\mathcal{L}).
\end{align*}
Suppose, to the contrary, that there is another $(\mathcal{L}', c')$ with $\mathcal{L} \not\sim \mathcal{L}'$ or $c \not\sim c'$ that satisfies the conditions in \cref{thm:representation},
and
assume that $\mathcal{L}'$ is non-redundant, i.e., $|\mathcal{L}| = |\mathcal{L}'|$.

If $\mathcal{L} \not\sim \mathcal{L}'$,
then there is $L \in \mathcal{L}$ such that $L \not\sim L'$ for all $L' \in \mathcal{L}'$.
For a set $X \subseteq [n]$,
denote by $1_X \in \{0,1\}^n$ the characteristic vector of $X$.
We can easily see that, for $\mathcal{A}$-cut $X$ with $X \not\sim L$,
0-1 vectors $1_{A_1}, \dots, 1_{A_r}, 1_L, 1_X$ are linearly independent.
Hence, for $k$ with $\alpha(L) < k < \beta(L)$,
the dimension of $\left\{ x \in \overline{\DF} \mid \sum_{i \in L} x_i = k \right\}$
is larger than that of $\left\{ x \in \overline{\DF} \mid \sum_{i \in L} x_i = k,\ \sum_{i \in X} x_i = k' \right\}$
for each $k'$ with $\alpha(X) < k' < \beta(X)$.
This implies
$\bigcup_{\alpha(L) < k < \beta(L)} \left\{ x \in \overline{\DF} \mid \sum_{i \in L} x_i = k \right\} \not\subseteq P(\mathcal{L}')$,
and hence $P(\mathcal{L}) \neq P(\mathcal{L}')$.
Therefore $\overline{f}$ has two different sets of nondifferentiable points $P(\mathcal{L})$ and $P(\mathcal{L}')$,
a contradiction.
Hence $\mathcal{L} \sim \mathcal{L}'$ holds,
and we can assume $\mathcal{L} = \mathcal{L}'$.
If $c \not\sim c'$, i.e., $c \neq c'$,
then there is $L \in \mathcal{L}$ such that $c(L) \neq c'(L)$.
By assuming $c(L) > c'(L) (> 0)$,
we can easily see that $\overline{f} - c'(L)\ell_L$ has two different sets $P(\mathcal{L})$ and $P(\mathcal{L} \setminus \{ L \})$ of nondifferentiable points,
a contradiction.
Hence $c(L) = c'(L)$ holds for all $L \in \mathcal{L}$.

We have thus proved the uniqueness part of \cref{thm:representation}.

\section{Algorithm for \DECOMP}\label{sec:(i)}
Let $f$ be a VCSP-quadratic function of type $\mathcal{A} = \{ A_1, A_2, \dots, A_r \}$.
In this section,
we devise an $O(rn^2)$-time algorithm for \DECOMP,
where as before $n = \sum_{1\leq p \leq r} |D_p|$.

\subsection{Outline}\label{subsec:Decomposition outline}
To describe our algorithm,
we need the concept of {\it restriction} of a VCSP-quadratic function.
Recall that $f$ is represented as~\eqref{eq:f}.
For $Q \subseteq [r]$, let $\mathcal{A}_Q := \{A_p\}_{p \in Q}$ be the subfamily of $\mathcal{A}$ corresponding to $Q$
and $A_Q := \bigcup_{p \in Q} A_p$ be the subset of $[n]$ corresponding to $Q$.
The {\it restriction of $f$ to $Q$} is a VCSP-quadratic function $f_Q : \{0,1\}^{A_Q} \rightarrow \ER$ of type $\mathcal{A}_Q$ defined by
\begin{align*}
f_Q(x) :=
\begin{cases}
\displaystyle \sum_{i,j \in A_Q, i<j}a_{ij} x_i x_j + \sum_{i \in A_Q} a_i x_i & \text{if $\displaystyle \sum_{i \in A_Q} x_i = |Q|$},\\
+\infty & \text{otherwise}.
\end{cases}
\end{align*}
\begin{lem}\label{lem:restriction}
	If $f$ is QR-M${}_2$-convex,
	then so is any of its restrictions.
\end{lem}
\begin{proof}
	By \cref{lem:anti-ultrametric} and \cref{prop:Type I + Type III},
	$f$ is representable as $f = f' + \delta_\mathcal{A}$,
	where the quadratic coefficient $(a'_{ij})_{i,j \in [n]}$ of $f'$ satisfies the anti-ultrametric property.
	Then $(a'_{ij})_{i,j \in A_Q}$ also has the anti-ultrametric property.
	Hence $f_Q$ is naturally representable as $f_Q = g + \delta_{\mathcal{A}_Q}$,
	where the quadratic coefficient of $g$ is $(a'_{ij})_{i,j \in A_Q}$.
	Thus $f_Q$ is QR-M${}_2$-convex.
\end{proof}

Suppose that $f$ is QR-M${}_2$-convex.
Then $f_Q$ is also QR-M${}_2$-convex
by \cref{lem:restriction}.
By \cref{thm:representation},
$f_Q$ can be represented as
\begin{align}\label{eq:f_Q,P_Q,c_Q}
f_Q = \sum_{X \in \mathcal{F}_Q}c_Q(X) \ell_{X} + h_Q
\end{align}
for some laminarizable $\mathcal{A}_Q$-cut family $\mathcal{F}_Q$, positive weight $c_Q$ on $\mathcal{F}_Q$, and $\mathcal{A}_Q$-linear function $h_Q$,
where $\ell_{X}$ and $h_Q$ are defined on $\{0,1\}^{A_Q}$.
Furthermore such $\mathcal{F}_Q$ and $c_Q$ are uniquely determined up to $\sim$.

Our algorithm for \DECOMP\ obtains an appropriate decomposition~\eqref{eq:f_Q,P_Q,c_Q} of $f_Q$ for $Q = \{1,2\}, \{1,2,3\}, \dots, \{1,2,3,\dots,r\}$ in turn as follows.
\begin{itemize}
	\item In the initial case for $Q = \{1,2\}$,
	we can obtain the decomposition~\eqref{eq:f_Q,P_Q,c_Q}
	with $(\mathcal{F}_Q, c_Q) = (\mathcal{L}_{12},c_{12})$ by executing Algorithm~1 for $f_{12}$ (\cref{subsec:r=2}).
	\item For each $t \geq 3$, we extend $(\mathcal{F}_{[t-1]}, c_{[t-1]})$ to $(\mathcal{F}_{[t]}, c_{[t]})$ by Algorithm~2 (\cref{subsec:r >= 3}),
	where $\mathcal{F}_{[2]} = \mathcal{L}_{12}$.
	In order to construct $(\mathcal{F}_{[t]}, c_{[t]})$ from $(\mathcal{F}_{[t-1]}, c_{[t-1]})$,
	we use $(\mathcal{L}_{pt},c_{pt})$ for all $p \in [t-1]$,
	which can be obtained by executing Algorithm~1 for $f_{pt}$.
	\item We perform the above extension step for $t=3$ to $t = r$.
	Then we can say that the resulting $\mathcal{A}$-cut family $\mathcal{F}_{[r]}$ is laminarizable,
	as required.
	This is described in Algorithm~3 (\cref{subsec:r >= 3}).
\end{itemize}
Note that our algorithm may output some decomposition~\eqref{eq:(L, c)inDEC} even when $f$ is not QR-M${}_2$-convex.
In this case,
the $\mathcal{A}$-cut family $\mathcal{F}$ output by the algorithm is not laminarizable.

\subsection{Case of $r = 2$}\label{subsec:r=2}
We consider the \DECOMP\ algorithm for the case of $r = 2$,
where $\mathcal{A}$ is a bipartition of $[n]$ represented as $\{ A_1, A_2\}$.
Note that $X$ is an $\mathcal{A}$-cut if and only if $X$ satisfies $\emptyset \neq (X \cap A_1) \neq A_1$ and $\emptyset \neq (X \cap A_2) \neq A_2$,
and that two $\mathcal{A}$-cuts $X$ and $Y$ are $\mathcal{A}$-equivalent (i.e., $X \sim Y$) if and only if $X = Y$ or $X = [n] \setminus Y$ by~\eqref{eq:sim}.
Let $f$ be a VCSP-quadratic function of type $\{A_1, A_2\}$,
and $(a_{ij})_{i,j \in [n]}$ be the quadratic coefficient of $f$,
where $a_{ij} = a_{ji}$ is always assumed.

Our algorithm makes use of the simple fact that,
for any $i^* \in [n]$ and $b \in \R$,
the modification of the coefficients as $a'_{i^*j} \leftarrow a_{i^*j} - b$
(as well as $a'_{ji^*} \leftarrow a_{ji^*} - b$)
for all $j \in [n] \setminus \{i^*\}$ does not affect the QR-M${}_2$-convexity of $f$.
Indeed,
the difference between
$\sum_{i < j} a_{ij} x_i x_j$
and
$\sum_{i < j} a'_{ij} x_i x_j$ is
an $\mathcal{A}$-linear function
since, for $x \in \DF$, it holds
\begin{align*}
\sum_{i < j} a_{ij} x_i x_j
= \left(\sum_{j : j < i^*} (a_{ji^*} - b) x_jx_{i^*} + \sum_{j : j > i^*} (a_{i^*j} - b) x_{i^*}x_j\right) + \sum_{i,j \in [n] \setminus \{i^*\}, i < j} a_{ij} x_i x_j + bx_{i^*}.
\end{align*}

We repeat the above modification of coefficients
for $i^*  =1,2,\ldots, n$ with appropriate choices of $b = b_{1}, b_{2}, \ldots, b_{n} \in \R$.
Then we test for the QR-M${}_2$-convexity
with reference to the condition (CB) below on a quadratic coefficient
$(a_{ij})_{i,j \in [n]}$:
\begin{itemize}
	\item[(CB)]
	Let the distinct values of $a_{ij}$ ($i \in A_1, j \in A_2$)
	be $\alpha_1 > \alpha_2 > \cdots > \alpha_m = \min_{i<j}a_{ij}$.
	For all $\alpha \in \{\alpha_1, \alpha_2, \dots, \alpha_{m-1} \}$,
	every non-isolated connected component of $G_\alpha := (V, E_\alpha)$ is a complete bipartite graph,
	where $E_\alpha := \{ \{i,j\} \mid i\in A_1,\ j \in A_2,\  \alpha \leq a_{ij} \}$.
\end{itemize}

The following lemma gives a sufficient condition for the QR-M${}_2$-convexity of $f$ in~\eqref{eq:f}.
\begin{lem}\label{lem:(CB)}
	If $(a_{ij} - b_i-b_j)_{i,j \in [n]}$ satisfies {\rm (CB)} for some $b_1, b_2, \dots, b_n \in \R$,
	then $f$ in~{\rm \eqref{eq:f}} is QR-M${}_2$-convex.
\end{lem}
\begin{proof}
	Let $f'$ be defined by the quadratic coefficient
	$(a_{ij} - b_i - b_j)$ as in \eqref{eq:f}.
	Then $f$ is QR-M${}_2$-convex if and only if $f'$ is QR-M${}_2$-convex.
	For each $s \in [m-1]$,
	denote by $\mathcal{L}^s$ the set of non-isolated connected components $L$ of $G_{\alpha_s}$.
	Their union $\mathcal{L} := \bigcup_{s = 1}^{m-1} \mathcal{L}^s$ is a laminar family.
	For $L \in \mathcal{L}$,
	denote by $L^+$ the minimal element in $\mathcal{L} \cup \{[n]\}$ properly containing $L$.
	We define $\alpha_L$ for $L \in \mathcal{L} \cup \{[n]\}$ as follows:
	$\alpha_{[n]} := \alpha_m$ and
	$\alpha_L := \alpha_s$ if $L \in \mathcal{L}^s \setminus \mathcal{L}^{s-1}$ with $s \in [m-1]$,
	where $\mathcal{L}^0 := \emptyset$.
	Since $(a_{ij} - b_i - b_j)$ satisfies (CB),
	we have
	\begin{align*}
	\sum_{i<j} (a_{ij} - b_i -b_j) x_ix_j
	&= \sum_{L \in \mathcal{L}} (\alpha_L - \alpha_{L^+}) \ell_L(x) + \alpha_m\\
	&\simeq \sum_{L \in \mathcal{L}^*} (\alpha_L - \alpha_{L^+}) \ell_L(x),
	\end{align*}
	where $\mathcal{L}^*$ is the family of $\mathcal{A}$-cuts in $\mathcal{L}$.
	We have thus obtained a representation of $f'$ in the form of \eqref{eq:(L, c)}
	with a laminar $\mathcal{A}$-cut family $\mathcal{L}^*$ and
	a positive weight $c(L) = \alpha_L - \alpha_{L^+}$ on $\mathcal{L}^*$.
	Then $f'$ is QR-M${}_2$-convex by \cref{thm:representation}, and hence
	$f$ is QR-M${}_2$-convex.
\end{proof}

The \DECOMP\ algorithm for the case of $r = 2$
is described as Algorithm~1 below.
The validity of this algorithm (\cref{prop:M_2pq})
implies that the converse of \cref{lem:(CB)} is also true,
that is,
if $f$ is QR-M${}_2$-convex
then $(a_{ij} - b_i-b_j)_{i,j \in [n]}$ satisfies (CB) by appropriate $b_i$'s,
and that such $b_i$'s can be computed easily.
\begin{description}
	\item[Algorithm~1 (for \DECOMP\ in the case of $r = 2$):]
	\item[Input:]
	A VCSP-quadratic function $f$ of type $\{A_1, A_2\}$.
	\item[Step 0:] Define $\alpha^* := \min_{i,j \in [n], i<j} a_{ij}$.
	\item[Step 1:] For $i = 1, 2, \dots, n$,
	do the following:
	Define $b_i := \min_{j \in [n] \setminus \{i\}} a_{ij} - \alpha^*$,
	and
	update $a_{ij} \leftarrow a_{ij} - b_i$ (as well as $a_{ji} \leftarrow a_{ji} - b_i$) for $j \in [n] \setminus \{i\}$.
	Then go to next $i$.
	\item[Step 2:] Check whether $(a_{ij})_{i,j \in [n]}$ satisfies (CB) or not.
	If $(a_{ij})_{i,j \in [n]}$ does not satisfy (CB),
	then output ``$f$ is not QR-M${}_2$-convex'' and stop.
	If $(a_{ij})_{i,j \in [n]}$ satisfies (CB),
	define $\alpha_1 > \alpha_2 > \cdots > \alpha_m$ and $G_\alpha$
	as in the condition (CB).
	\item[Step 3:]
	For each $s \in [m-1]$,
	denote by $\mathcal{L}^s$ the set of non-isolated connected components $L$ of $G_{\alpha_s}$.
	Define a laminar family $\mathcal{L}$ by $\mathcal{L} := \bigcup_{s = 1}^{m-1} \mathcal{L}^s$.
	For $L \in \mathcal{L}$,
	denote by $L^+$ the minimal element in $\mathcal{L} \cup \{[n]\}$ properly containing $L$.
	Define $\alpha_L$ for $L \in \mathcal{L} \cup \{[n]\}$ by:
	$\alpha_{[n]} := \alpha_m$ and 
	$\alpha_L := \alpha_s$ if  $L \in \mathcal{L}^s \setminus \mathcal{L}^{s-1}$ with $s \in [m-1]$,
	where $\mathcal{L}^0 := \emptyset$.
	Define $c : \mathcal{L} \rightarrow \R_{++}$ by $c(L) := \alpha_L - \alpha_{L^+}$.
	\item[Step 4:]
	If both $X$ and $[n] \setminus X$ belong to $\mathcal{L}$,
	then update $c$ by
	$c(X) \leftarrow c(X) + c([n] \setminus X)$ and
	remove $[n] \setminus X$ from $\mathcal{L}$.
	We consider that the new $c$ is a weight on the new $\mathcal{L}$.
	\item[Step 5:]
	Output $\mathcal{L}$ and $c$.~\qqed
\end{description}
Note that, by Step~4,
the output $\mathcal{L}$ is non-redundant.

\begin{exmp}\label{ex:instance decompotion r=2}
	For the function $f$ in~\eqref{eq:f instance},
	we execute Algorithm~1 for the restriction $f_{12}$ to $\{1,2\}$.
	Recall that $n = 4$, $a_{13} = 3$, $a_{14} = 0$, $a_{23} = 1$, $a_{24} = 4$,
	and $a_{12} = a_{34} = +\infty$.
	
	In Step~0,
	we define $\alpha^* := 0$.
	In Step~1,
	we update $a_{23} \leftarrow 0$ and $a_{24} \leftarrow 3$ (see Figure~\ref{fig:algo1}).
	\begin{figure}
		\centering
		\includegraphics[clip, width=12cm]{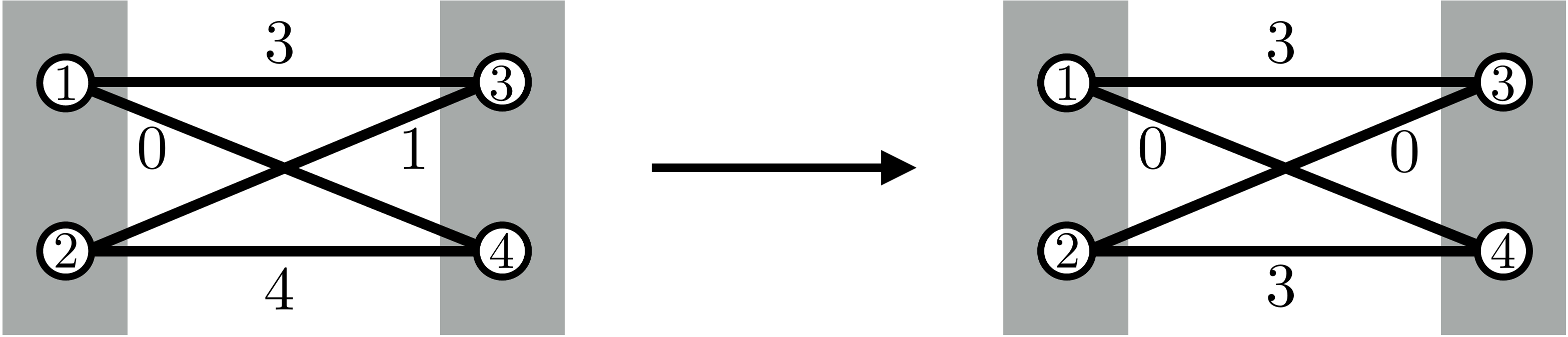}
		\caption{
			The left figure illustrates the values of $a_{13}, a_{14}, a_{23}, a_{24}$ before Step~1,
			and the right figure illustrates those values after Step~1.}
		\label{fig:algo1}
	\end{figure}
We can easily see, by Figure~\ref{fig:algo1},
$\mathcal{L} = \{ 13, 24 \}$, $\alpha_{13} = \alpha_{24} = 3$, and $\alpha_{1234} = 0$ in Step~3.
In Step~4,
we redefine $\mathcal{L}$ by $\mathcal{L} := \{24\}$ and $c : \mathcal{L} \rightarrow \R_{++}$ by $c(24) := 6$.
Then, in Step~5, we output $\mathcal{L}$ and $c$.
Note that, in Step 4,
we can also redefine $\mathcal{L}$ by $\mathcal{L} := \{13\}$ and $c : \mathcal{L} \rightarrow \R_{++}$ by $c(13) := 6$.~\qqed
\end{exmp}

\begin{prop}\label{prop:M_2pq}
	Algorithm~{\rm 1} solves \DECOMP\ in $O(n^2)$ time.
\end{prop}
For the proof of the validity of Algorithm~1,
we need the following lemma.
\begin{lem}[{\cite[Lemma~4.2]{DAM/I18}}]\label{lem:Type I key}
	Suppose that $(a_{ij})_{i,j \in [n]}$ satisfies the anti-tree metric property~{\rm \eqref{eq:Type I}}
	and let $\alpha^* := \min_{i,j \in [n], i<j} a_{ij}$.
	If $\min_{j \in [n]} a_{ij} = \alpha^*$ holds for all $i \in [n]$,
	then $(a_{ij})_{i,j \in [n]}$ satisfies the anti-ultrametric property~{\rm \eqref{eq:anti-ultrametric}}.
\end{lem}
\begin{proof}[Proof of \cref{prop:M_2pq}]
	(Validity).
We show that
\begin{itemize}
	\item
	if $f$ is not QR-M${}_2$-convex, the algorithm terminates in Step 2, with the message that
	$f$ is not QR-M${}_2$-convex, and
	\item
	if $f$ is QR-M${}_2$-convex, the algorithm terminates in Step 5, with a correct
	representation of $f$ in the form~\eqref{eq:(L, c)}.
\end{itemize}
This means, in particular, that
the algorithm for $r=2$ always detects non-QR-M${}_2$-convexity,
and never outputs a representation~\eqref{eq:(L, c)inDEC}
with a non-laminarizable family $\mathcal{F}$
if $f$ is not QR-M${}_2$-convex.

Suppose that $f$ is not QR-M${}_2$-convex.
By (the contrapositive of) Lemma~\ref{lem:(CB)},
$(a_{ij})$ in Step~2 does not satisfy (CB).
Accordingly, the algorithm terminates in Step 2, which is legitimate.

Suppose that $f$ is QR-M${}_2$-convex.
In this case, $(a_{ij})$ in Step~2 satisfies (CB),
which is shown in Claim below.
Then the algorithm terminates in Step 5 by outputting $(\mathcal{L}, c)$.
The laminar family $\mathcal{L}$, obtained in Step~3,
is an $\mathcal{A}$-cut family.
Indeed, by the operation in Step~1,
\begin{align}\label{eq:alpha^*}
\min_{j' \in A_2} a_{ij'} = \min_{i' \in A_1} a_{i'j} = \alpha^*
\end{align}
holds
for any $i \in A_1$ and $j \in A_2$.
This implies that
each $L \in \mathcal{L}$ is an $\mathcal{A}$-cut,
since otherwise $\min_{j' \in [n] \setminus \{i\}} a_{ij'} > \alpha^*$ holds
for some $i \in L$.
Therefore, by the proof of \cref{lem:(CB)}, the output $(\mathcal{L}, c)$ gives
a correct representation of $f$ in the form~\eqref{eq:(L, c)}.

It remains to prove the following claim.

\begin{cl*}
	If $f$ is QR-M${}_2$-convex,
	then $(a_{ij})_{i,j \in [n]}$ in Step~2 satisfies {\rm (CB)}.
\end{cl*}
\begin{proof}[Proof of Claim]
	Suppose that $f$ is QR-M${}_2$-convex.
	In the following,
	we prove that there is a coefficient $(\overline{a}_{ij})$ satisfying the anti-ultrametric property such that
	$a_{ij} = \overline{a}_{ij}$ for every $i \in A_1$ and $j \in A_2$,
	where it should be clear that $a_{ij} = +\infty$ if $i,j \in A_1$ or $i,j \in A_2$.
	This implies,
	by \cref{lem:anti-ultrametric}, that
	there are a laminar family $\mathcal{L}$ and a positive weight $w$ on $\mathcal{L}$ associated with $(\overline{a}_{ij})$ as~\eqref{eq:a_ij laminar}.
	Then
	$a_{ij}$ can be represented as
	\begin{align*}
	a_{ij} =
	\begin{cases}
	\{ c(L) \mid L \in \mathcal{L} {\rm \ with\ } i,j \in L \} + \alpha^* & \textrm{if $i \in A_1$ and $j \in A_2$},\\
	+\infty & \textrm{if $i,j \in A_1$ or $i,j \in A_2$}.
	\end{cases}
	\end{align*}
	Hence
	$(a_{ij})$ satisfies (CB)
	and the laminar family obtained in Step~3 coincides with the family of $\mathcal{A}$-cuts in $\mathcal{L}$.

	We now start to prove the existence of $(\overline{a}_{ij})$.
	By \cref{lem:Type I anti-ultra} and \cref{prop:Type I + Type III},
	we have $f(x) \simeq \sum_{i<j} a_{ij}' x_i x_j$,
	where $(a_{ij}')$ is a coefficient satisfying the anti-ultrametric property.
	This implies that $\sum_{i<j} a_{ij} x_i x_j - \sum_{i<j} a_{ij}' x_i x_j$ is $\mathcal{A}$-linear.
	Hence,
	for some $b_i', b_j' \in \R$ we have $a_{ij} = a_{ij}'  + b_i' + b_j'$ for every $i \in A_1, j \in A_2$.
	Let $\overline{a}_{ij} := a_{ij}' + b_i' + b_j'$ for distinct $i,j \in [n]$.
	Then $a_{ij} = \overline{a}_{ij}$ holds for any $i \in A_1, j \in A_2$, and $(\overline{a}_{ij})$ is a coefficient satisfying the anti-tree metric property~\eqref{eq:Type I}.
	
	We can redefine the coefficient $(\overline{a}_{ij})$
	so as to meet the anti-ultrametric property while maintaining $a_{ij} = \overline{a}_{ij}$ for any $i \in A_1, j \in A_2$, as follows.
	Let $\beta := \alpha^* - \min \overline{a}_{ij}$.
	Note that $\beta \geq 0$ holds
	by~\eqref{eq:alpha^*} and $a_{ij} = \overline{a}_{ij}$ for $i \in A_1, j \in A_2$.
	First suppose $\beta = 0$.
	Then $\min \overline{a}_{ij} = \alpha^*$ holds.
	Furthermore,
	we have $\min_j \overline{a}_{ij} = \alpha^*$ for every $i \in [n]$.
	Hence, by \cref{lem:Type I key},
	$(\overline{a}_{ij})$ satisfies the anti-ultrametric property,
	as required.
	
	Next suppose $\beta > 0$.
	By $\overline{a}_{ij} \geq \alpha^*$ for any $i \in A_1, j \in A_2$,
	if $\overline{a}_{i^*j^*} = \alpha^* - \beta$,
	then $i^*, j^* \in A_1$ or $i^*, j^* \in A_2$ holds.
	Without loss of generality, we assume $i^*, j^* \in A_1$.
	Since $(\overline{a}_{ij})$ satisfies~\eqref{eq:Type I},
	it holds that $\overline{a}_{i^*j^*} + \overline{a}_{kl} \geq 2 \alpha^*$ for all distinct $k,l \in A_2$.
	Hence we have $\min_{k,l \in A_2} \overline{a}_{kl} \geq \alpha^* + \beta$.
	Let $\overline{b}_i := \beta/2$ if $i \in A_1$ and $\overline{b}_i := - \beta/2$ if $i \in A_2$.
	We redefine $\overline{a}_{ij}$ as $\overline{a}_{ij} \leftarrow \overline{a}_{ij} + \overline{b}_i + \overline{b}_j$.
	Then it is easy to see that $a_{ij} = \overline{a}_{ij}$ holds for any $i \in A_1, j \in A_2$, and that $(\overline{a}_{ij})$ is a coefficient satisfying~\eqref{eq:Type I}.
	Furthermore $\alpha^* - \min \overline{a}_{ij} = 0$ holds.
	Hence, by Lemma~\ref{lem:Type I key},
	$(\overline{a}_{ij})$ satisfies the anti-ultrametric property,
	as required.
	
	This completes the proof of Claim.
	\end{proof}
	
	(Complexity).
	It is clear that Steps~0 and~1 can be done in $O(n^2)$ time,
	and that Steps~4 and~5 can be done in $O(|\mathcal{L}|) = O(n)$ time.
	
	We show that Steps~2 and~3 can be done in $O(n^2)$ time,
	improving the $O(n^3)$ time complexity of a naive implementation.
	Our approach is based on the idea used in~\cite[Section~4.2.2]{DAM/I18} (see also~\cite{IPL/CR89, JTB/WSSB77}).
	Suppose that $f$ is QR-M${}_2$-convex,
	and that we are given some $L \in \mathcal{L}$.
	We can compute in $O(|L|^2)$ time
	the (disjoint) set $\mathcal{L}'$ of all maximal members in $\mathcal{L}$
	properly contained in $L$ as follows.
	Let $L_1 := A_1 \cap L$ and $L_2 := A_2 \cap L$.
	Observe that
	$\alpha_L = \min_{j' \in L_2} a_{ij'} = \min_{i' \in L_1} a_{i'j}$
	holds for each $i \in L_1$ and $j \in L_2$.
	Choose arbitrary $i \in L_1$,
	and compute $\argmin_{j' \in L_2} a_{ij'}$.
	If $a_{ij'}$ is constant on $j' \in L_2$,
	then there is no member of $\mathcal{L}'$ containing $i$.
	Otherwise, choose $j$ from $L_2 \setminus \argmin_{j' \in L_2} a_{ij'}$, and
	compute $\argmin_{i' \in L_1} a_{i'j}$.
	Then one can see that the (unique) member $L'$ in $\mathcal{L}$
	containing $i,j$
	is equal to the union of $L_1 \setminus \argmin_{i' \in L_1} a_{i'j}$
	and $L_2 \setminus \argmin_{j' \in L_2} a_{ij'}$.
	By repeating this procedure, we obtain $\mathcal{L}'$ in $O(|L|^2)$ time.
	Thus, starting from $L = [n]$,
	we recursively apply this procedure to the $L$'s so far obtained,
	and finally get $\mathcal{L}$ (as well as $c:\mathcal{L} \rightarrow \R_{++}$) in $O(n^2)$ time in total.
	Even when $f$ is not QR-M${}_2$-convex,
	we can apply this procedure
	and detect the non-QR-M${}_2$-convexity.
	Indeed,
	define $a_{ij}'$ as $\alpha_L$ for the final $L$ containing $i,j$ in the above procedure.
	Then $a'_{ij} = a_{ij}$ holds for any $i,j$
	if and only if $(a_{ij})$ satisfies the anti-ultrametric property, i.e., $f$ is QR-M${}_2$-convex.	
\end{proof}

\subsection{Case of $r \geq 3$}\label{subsec:r >= 3}
To obtain the decomposition~\eqref{eq:f_Q,P_Q,c_Q} of the restriction $f_Q$ for $Q = \{1,2\}, \{1,2,3\}, \dots, \{1,2,3,\dots,r\}$ in turn,
we need to extend $(\mathcal{F}_{[t-1]}, c_{[t-1]})$
to $(\mathcal{F}_{[t]}, c_{[t]})$
with the use of  $(\mathcal{L}_{pt}, c_{pt})$ ($p \in [t-1]$) for $t = 3, \dots, r$.
Algorithm~2 below corresponds to this extension procedure.

We explain the idea of the extension 
for $t = r$, i.e., from $(\mathcal{F}_{[r-1]}, c_{[r-1]})$
to $(\mathcal{F}_{[r]}, c_{[r]})$.
Suppose that we are given an $\mathcal{A}_{[r-1]}$-cut family $\mathcal{F}'$ and a positive weight $c'$ on $\mathcal{F}'$ satisfying, for $f' := f_{[r-1]}$,
\begin{align}\label{eq:f'}
f' = \sum_{X \in \mathcal{F}'} c'(X) \ell_X + h'
\end{align}
for some $\mathcal{A}_{[r-1]}$-linear function $h'$.

The extension procedure consists of two phases.
In the first phase, we construct
an $\mathcal{A}$-cut family $\mathcal{F}$ and a positive weight $c$ on $\mathcal{F}$
that represent $f$ as
\begin{align}\label{eq:initial f}
f = \sum_{X \in \mathcal{F}} c(X) \ell_X + h
\end{align}
for some $\mathcal{A}$-linear function $h$.
In this representation, however,
the family $\mathcal{F}$ is not necessarily laminarizable
even when $f$ is QR-M${}_2$-convex.
In the second phase we modify
$(\mathcal{F}, c)$ in (\ref{eq:initial f})
to another pair $(\mathcal{F}\sp{*}, c\sp{*})$
such that $\mathcal{F}\sp{*}$ is laminarizable when $f$ is QR-M${}_2$-convex.
The key operation of the second phase is called a ``composition'' operation.

The first phase is easy and straightforward.
Suppose that we have a decomposition~\eqref{eq:f'}
for $f'$ in terms of $(\mathcal{F}', c')$.
For $p =1,2,\ldots, r-1$,
we apply Algorithm~1 to $f_{pr}$
to obtain the decomposition~\eqref{eq:f_Q,P_Q,c_Q} of $f_{pr}$
in terms of $(\mathcal{L}_{pr}, c_{pr})$.
If Algorithm~1 should detect the non-QR-M${}_2$-convexity of $f_{pr}$ for some $p \in [r-1]$,
then
$f$ is not QR-M${}_2$-convex by \cref{lem:restriction}, and therefore,
we can give up our construction immediately.
Otherwise, we merge $(\mathcal{F}', c')$
and $(\mathcal{L}_{pr}, c_{pr})$
($p \in [r-1]$)
to obtain a representation of $f$.
Let
$\mathcal{F} := \mathcal{F}' \cup \bigcup_{p \in [r-1]} \mathcal{L}_{pr}$,
which is an $\mathcal{A}$-cut family,
and define a positive weight $c$ on $\mathcal{F}$
by $c(X) := c'(X)$ for $X \in \mathcal{F}'$ and $c(X) := c_{pr}(X)$ for $X \in \mathcal{L}_{pr}$.
Then, with the notation
$x|_Q := (x_i)_{i \in A_Q} \in \{0,1\}^{A_Q}$
for $x = (x_1,x_2,\dots,x_n) \in \{0,1\}^n$ and $Q \subseteq [r]$,
we have
\begin{align*}
f(x) &\simeq
\begin{cases}
\displaystyle\sum_{i,j \in A_{[r-1]}, i<j} a_{ij} x_i x_j+ \sum_{p \in [r-1]} \sum_{i,j \in A_{pr}, i<j} a_{ij}x_ix_j & \textrm{if $\displaystyle\sum_i x_i = r$},\\
+\infty & \textrm{otherwise}
\end{cases}\\
&\simeq f'(x|_{[r-1]}) + \sum_{p \in [r-1]} f_{pr}(x|_{pr})\\
&\simeq \sum_{X \in \mathcal{F}} c(X) \ell_X.
\end{align*}
Thus the representation (\ref{eq:initial f}) for $f$ is obtained.
Recall that we do not impose laminarizability on $\mathcal{F}$
even when $f$ is QR-M${}_2$-convex.
As the above argument shows, no substantial computation is required in the first phase.

The second phase consists of modifying
$(\mathcal{F}, c)$ in~\eqref{eq:initial f}
to another pair $(\mathcal{F}\sp{*}, c\sp{*})$
with the additional property that
$\mathcal{F}\sp{*}$
is laminarizable when $f$ is QR-M${}_2$-convex.
For this modification we introduce a ``composition'' operation.
Before entering into a formal description, we illustrate this modification
for simple examples in Figures~\ref{fig:algo3-1} and~\ref{fig:algo3-2}.
In Figure~\ref{fig:algo3-1}, the given family
$\mathcal{F} = \{ 24 , 15, 35 \}$
at the left is not laminar and the resulting family
$\mathcal{F}\sp{*} = \{ 135 , 24 \}$
at the right is laminar;
the new $\mathcal{A}$-cut $X^* = 135$ is
constructed by our algorithm by combining $24$, $15$, and $35$.
In Figure~\ref{fig:algo3-2}, the given family
$\mathcal{F}
= \{ 135 , 24 , 28, 37, 57 \}$
at the left is not laminarizable and the resulting family
$\mathcal{F}\sp{*}
= \{ 135, 24 , 1357, 37 \}$
at the right is not laminar but laminarizable;
the new $\mathcal{A}$-cut $X^* = 1357$
is constructed by our algorithm by combining 
$135$, $28$, $37$, and $57$.
\begin{figure}
	\centering
	\includegraphics[clip, width=12cm]{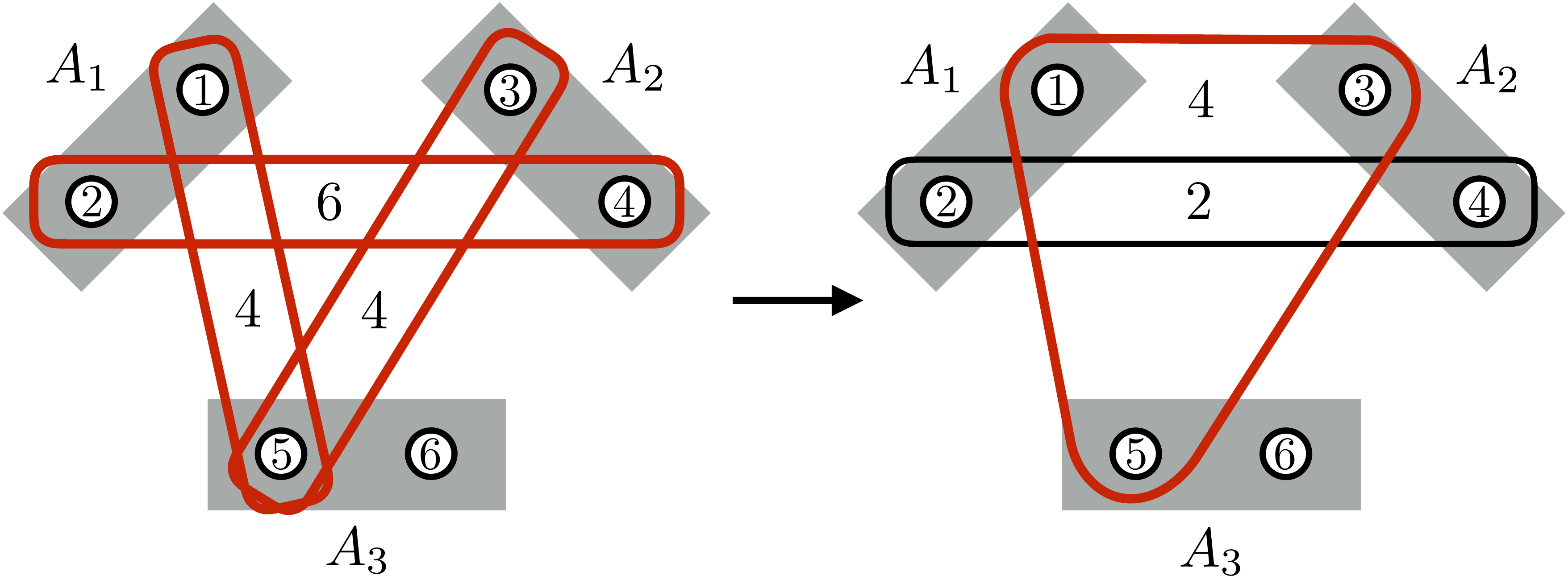}
	\caption{
		The left figure illustrates $(\mathcal{L}_{12}, c_{12})$, $(\mathcal{L}_{13}, c_{13})$, and $(\mathcal{L}_{23}, c_{23})$,
		and the right figure illustrates $(\mathcal{F}_{123}, c_{123})$.
	A pair $(15,35)$ is a composable tuple to $24$
	since $135$ satisfies $135 \sim_{12} 24$, $135 \sim_{13} 15$, and $135 \sim_{23} 35$}
	\label{fig:algo3-1}
\end{figure}
\begin{figure}
	\centering
	\includegraphics[clip, width=12cm]{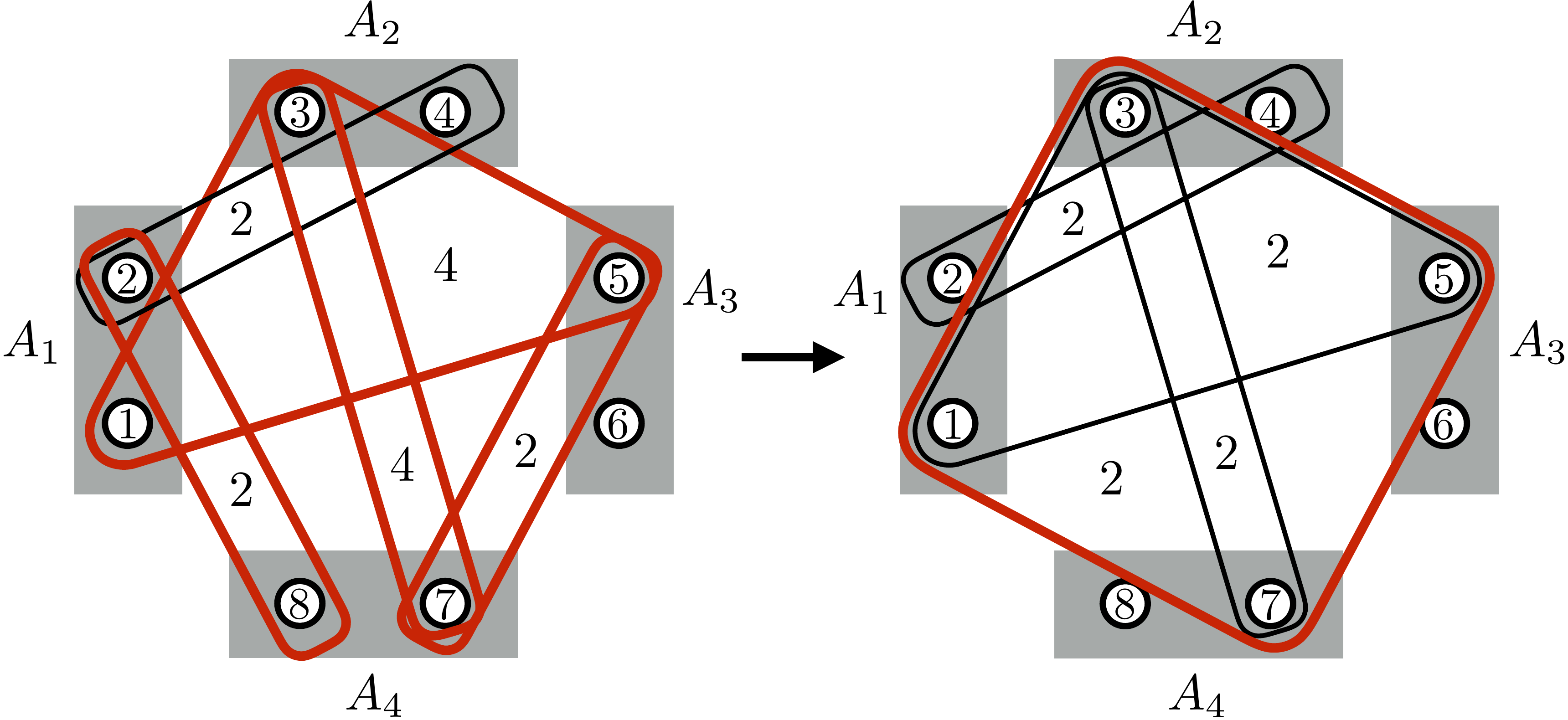}
	\caption{
		The left figure illustrates $(\mathcal{F}_{123}, c_{123})$, $(\mathcal{L}_{14}, c_{14})$, $(\mathcal{L}_{24}, c_{24})$, and $(\mathcal{L}_{34}, c_{34})$
		and the right figure illustrates $(\mathcal{F}, c)$.
		A triple $(28,37,57)$ is a composable tuple to $135$
		since $1357$ satisfies $1357 \sim_{123} 135$, $1357 \sim_{14} 28$, $1357 \sim_{24} 37$,
		and $1357 \sim_{34} 57$.
		Note that the output family $\mathcal{F}$ (described in the right) is the same as the family described in the left in Figure~\ref{fig:Laminarization}.}
	\label{fig:algo3-2}
\end{figure}

In order to explain the composition operation,
we introduce the {\it $\mathcal{A}_Q$-equivalence} $\sim_Q$
by generalizing the characterization of $\sim$ in~\eqref{eq:sim}.
For a nonempty $Q \subseteq [r]$,
we define $\sim_Q$ for $\mathcal{A}$-cuts $X$ and $Y$ by:
\begin{align*}
X \sim_Q Y \Leftrightarrow \{ \ave{X}_Q \cap X, \ave{X}_Q \setminus X \} = \{ \ave{Y}_Q \cap Y, \ave{Y}_Q \setminus Y \},
\end{align*}
where $\ave{X}_Q := \ave{X} \cap A_Q$ and $\ave{Y}_Q := \ave{Y} \cap A_Q$.
See~\eqref{eq:ave} for the notation $\ave{X}$ of the cutting support of $X$.
Note, for $\mathcal{A}_Q$-cuts $X$ and $Y$, $X \sim Y$ if and only if $X \sim_Q Y$.

Let us start the description of the composition operation.
Suppose that $X_0$ is an $\mathcal{A}_{[r-1]}$-cut
and let $\{ p_1, p_2, \dots, p_k \}$ be the set of indices $p \in [r-1]$ with $\ave{X_0} = A_{\{p_1,p_2,\dots,p_k\}}$.
We say that $(X_1, X_2, \dots, X_k)$ is a {\it composable tuple to $X_0$}
if
\begin{itemize}
	\item $\ave{X_i}$ is an $\mathcal{A}_{p_ir}$-cut (i.e., $\ave{X_i} = A_{p_ir}$) for each $i \in [k]$, and
	\item there is an $\mathcal{A}$-cut $X^*$ satisfying $X^* \sim_{[r-1]} X_0$ and $X^* \sim_{p_ir} X_i$ for $i \in [k]$.
\end{itemize}
We say that $X^*$ in the second condition is a {\it composition of $X_0$ and $(X_1, X_2, \dots, X_k)$}.
Note that a composition $X^*$ is uniquely determined up to $\sim$.
Then it holds
\begin{align}\label{eq:l_X_i = l_X^*}
\ell_{X_0}  + \ell_{X_1} + \cdots + \ell_{X_k} \simeq \ell_{\ave{X_0} \cap X^*} + \ell_{\ave{X_1} \cap X^*} + \cdots + \ell_{\ave{X_k} \cap X^*} \simeq \ell_{X^*},
\end{align}
where the first equivalence follows from \cref{lem:sim}~(2)
and the second follows form the definition of $\ell_X$.
Let $\lambda$ be a positive value with
$\lambda = \min \{ c(X_0), c(X_1), \dots, c(X_k) \}$.
By substituting~\eqref{eq:l_X_i = l_X^*} into~\eqref{eq:initial f},
we obtain
\begin{align*}
f \simeq \left( \lambda \ell_{X^*} +  \sum_{X \in \{ X_0, \dots, X_k \}} (c(X) - \lambda) \ell_{X} + \sum_{X \in \mathcal{F} \setminus \{ X_0, \dots, X_k \}} c(X) \ell_{X}\right),
\end{align*}
and the above formula provides a new decomposition of $f$.
For example,
in Figure~\ref{fig:algo3-2},
we combine $X_0 = 135, X_1 = 28, X_2 = 37, X_3 = 57$
into $X^* = 1357$ with $\lambda = 2$.

The formal description of Algorithm~2 is the following.
It is noted that,
if $\mathcal{F}$ is a non-redundant laminarizable $\mathcal{A}$-cut family,
then $|\mathcal{F}|$ is at most $2n = 2|A_{[r]}|$ (see e.g., \cite[Theorem 3.5]{book/Schrijver03}).
\begin{description}
	\item[Algorithm 2 (for extending $f'$ to $f$):]
	\item[Input:] A VCSP-quadratic function $f$ of type $\mathcal{A}$ and restriction $f' := f_{[r-1]}$ given as~\eqref{eq:f'} with $(\mathcal{F}', c')$,
	where $\mathcal{F}'$ is non-redundant and satisfies $|\mathcal{F}'| \leq 2|A_{[r-1]}|$.
	\item[Output:] 
	Either detect the non-QR-M${}_2$-convexity of $f$,
	or obtain an expression of $f$ as
	\begin{align}\label{eq:algo output}
	\sum_{X \in \mathcal{F}} c(X) \ell_{X} + \delta_\mathcal{A} + h
	\end{align}
	with a non-redundant $\mathcal{A}$-cut family $\mathcal{F}$ satisfying $|\mathcal{F}| \leq 2n = 2|A_{[r]}|$ and a positive weight $c$ on $\mathcal{F}$,
	where $h$ is $\mathcal{A}$-linear.
\end{description}
\begin{description}
	\item[Step 1:] 
	For each $p \in [r-1]$,
	execute Algorithm~1 for $f_{pr}$.
	If Algorithm~1 returns ``$f_{pr}$ is not QR-M${}_2$-convex'' for some $p \in [r-1]$,
	then output ``$f$ is not QR-M${}_2$-convex''
	and stop.
	Otherwise,
	for all $p \in [r-1]$,
	obtain $\mathcal{L}_{pr}$ and $c_{pr}$.
	Let $\mathcal{F} := \emptyset$.
	\item[Step 2:]
	While $\mathcal{F}' \neq \emptyset$, do the following:
	Let $X_0$ be an element of $\mathcal{F}'$ such that $\ave{X_0}$ is maximal.
	Let $\{ p_1, p_2, \dots, p_k \}$ be the set of indices $p \in [r-1]$ with $\ave{X_0} = A_{\{p_1,p_2,\dots,p_k\}}$.
	\begin{itemize}
		\item
		If there exists a composable tuple $(X_1, X_2, \dots, X_k)$ to $X_0$ such that $X_i \in \mathcal{L}_{p_ik}$ for $i = 1,2,\dots k$,
		then define $\lambda := \min \{c'(X_0), c_{p_1r}(X_1), \dots, c_{p_kr}(X_k)\}$ and update as
		\begin{align*}
		\mathcal{F} &\leftarrow \mathcal{F} \cup \{X^*\},\\
		c(X^*) &\leftarrow \lambda,\\
		c'(X_0) &\leftarrow c'(X_0) - \lambda,\\
		c_{p_ir}(X_i) &\leftarrow c_{p_ir}(X_i) - \lambda \qquad (i \in [k]),
		\end{align*}
		where $X^*$ is a composition of $X_0$ and $(X_1, X_2, \dots, X_k)$.
		Then remove $X_0$ from $\mathcal{F}'$ if $c'(X_0) = 0$,
		and $X_i$ from $\mathcal{L}_{p_ir}$ if $c_{p_ir}(X_i) = 0$.
		\item
		Otherwise, update as
		\begin{align*}
		\mathcal{F} &\leftarrow \mathcal{F} \cup \{X_0\},\\
		c(X_0) &\leftarrow c'(X_0),\\
		\mathcal{F}' &\leftarrow \mathcal{F}' \setminus \{X_0\}.
		\end{align*}
	\end{itemize}
	\item[Step 3:]
	Update as
	\begin{align*}
	\mathcal{F} &\leftarrow \mathcal{F} \cup \bigcup_{p \in [r-1]} \mathcal{L}_{pr},\\
	 c(X) &\leftarrow c_{pr}(X) \qquad (p \in [r-1], X \in \mathcal{L}_{pr}).
	\end{align*}
	If $|\mathcal{F}| \leq 2n$,
	then output $\mathcal{F}$ and $c$.
	Otherwise, output ``$f$ is not QR-M${}_2$-convex.''~\qqed
\end{description}

\begin{exmp}\label{ex:instance decompotion}
	Let $f$ be the VCSP-quadratic function in~\eqref{eq:f instance}.	
	We first see how Algorithm~2 runs for $f_{123}$ with the input $(\mathcal{L}_{12} = \{ 24 \}, c_{12}(24) = 6)$.
	By executing Algorithm~1 for $f_{13}$ and $f_{23}$,
	we obtain $(\mathcal{L}_{13} = \{ 15 \}, c_{13}(15) = 4)$ and $(\mathcal{L}_{23} = \{ 35\}, c_{23}(35) = 4)$.
	In Step~2,
	we compose $15, 35, 24$ to $135$ as in Figure~\ref{fig:algo3-1}.
	Then we obtain a family $\mathcal{F}_{123} := \{ 135, 24 \}$
	and a positive weight $c_{123}$ on $\mathcal{F}_{123}$
	defined by $c_{123}(135) := 4$ and $c_{123}(24) := 2$.
	We cannot execute a composition operation any more.
	Hence Algorithm~2 outputs $(\mathcal{F}_{123}, c_{123})$.

	Next we see how Algorithm~2 runs for $f = f_{1234}$ with the input $(\mathcal{F}_{123}, c_{123})$.
	By executing Algorithm~1 for $f_{14}$, $f_{24}$, and $f_{34}$,
	we obtain $(\mathcal{L}_{14} = \{ 28 \}, c_{14}(28) = 2)$,
	$(\mathcal{L}_{24} = \{ 37 \}, c_{24}(37) = 4)$, and
	$(\mathcal{L}_{34} = \{57\}, c_{34}(57) = 2)$.
	In Step~2,
	we compose $135, 28,37,57$ to $1357$ as in Figure~\ref{fig:algo3-1}.
	Then we obtain a family $\mathcal{F} := \{ 1357,135,24,37 \}$
	and a positive weight $c$ on $\mathcal{F}$
	defined by $c(X) := 2$ for all $X \in \mathcal{F}$.
	Here we remark that
	we choose a composable tuple $(28,37,57)$ to $135$
	though $(28,37)$ is also a composable tuple to $24$.
	This is because $\ave{135} \supsetneq \ave{24}$; see Step~2.
	We cannot execute a composition operation any more.
	Hence Algorithm~2 outputs $(\mathcal{F}, c)$.~\qqed
\end{exmp}

The following proposition shows that Algorithm~2 works as expected.
\begin{prop}\label{prop:c_LL}
	The following hold:
	\begin{description}
		\item[{\rm (1)}] If Algorithm~{\rm 2} outputs $(\mathcal{F}, c)$,
		then $\mathcal{F}$ is non-redundant and the function~{\rm \eqref{eq:algo output}} for $(\mathcal{F}, c)$ is equal to $f$.
		\item[{\rm (2)}] If $f$ is QR-M${}_2$-convex and $\mathcal{F}'$ is laminarizable,
		then Algorithm~{\rm 2} outputs $(\mathcal{F}, c)$
		and $\mathcal{F}$ is laminarizable.
		\item[{\rm (3)}] Algorithm~{\rm 2} runs in $O(n^2)$ time provided $|A_r| \leq \min \{ |A_1|, |A_2|, \dots, |A_{r-1}| \}$.
	\end{description}
\end{prop}

For the proof of \cref{prop:c_LL} (2),
we need the following lemma.
\begin{lem}\label{lem:index restriction}
	Suppose that $f$ is QR-M${}_2$-convex,
	$\mathcal{F}$ is a laminarizable $\mathcal{A}$-cut family,
	and $c$ is a positive weight on $\mathcal{F}$,
	where $(\mathcal{F}, c)$ represents $f$ as in~\eqref{eq:(L, c)}.
	For $Q \subseteq [r]$,
	let $\mathcal{G} := \{ X \cap A_Q \mid X \in \mathcal{F}, X \cap A_Q: \mathcal{A}_Q{\rm \mathchar`-cut} \}$
	and $d$ be the positive weight on $\mathcal{G}$ defined by $d(Y) := \sum\{ c(X) \mid X \in \mathcal{F},\ X \cap A_Q = Y \}$.
	Then
	$\mathcal{F}_Q$ and $c_Q$ in~\eqref{eq:f_Q,P_Q,c_Q}
	satisfy $\mathcal{F}_Q \sim \mathcal{G}$ and $c_Q \sim d$.
\end{lem}
\begin{proof}
	For an $\mathcal{A}$-cut $X$ and $Q \subseteq [r]$,
	let $(\ell_X)_Q$ be the restriction of $\ell_X$ to $\{ 0,1 \}^{A_Q}$.
	Note that $(\ell_X)_Q + \delta_{\mathcal{A}_Q}$ is linear on $\dom{\delta_{\mathcal{A}_Q}}$ if and only if $X \cap A_Q$ is not an $\mathcal{A}_Q$-cut.
	Therefore it holds that
	\begin{align*}
	f_Q &\simeq \sum_{X \in \mathcal{F}} c(X) (\ell_X)_Q\\
	&\simeq \sum_{Y \in \mathcal{G}} \ell_Y \cdot \sum\{ c(X) \mid X \in \mathcal{F},\ X \cap A_Q = Y \}\\
	&= \sum_{Y \in \mathcal{G}} d(Y) \ell_Y.
	\end{align*}
	Furthermore, since $\mathcal{F}$ is laminarizable,
	so is $\mathcal{G}$.
	By the uniqueness of $\mathcal{F}_Q$ and $c_Q$ up to $\sim$ (\cref{thm:representation}),
	we obtain $\mathcal{F}_Q \sim \mathcal{G}$ and $c_Q \sim d$.
\end{proof}

We are now ready to show \cref{prop:c_LL}.
\begin{proof}[Proof of \cref{prop:c_LL}]
	(1).
	By the argument before the formal description of Algorithm~2,
	we can say that if Algorithm~2 outputs $\left( \mathcal{F}, c \right)$,
	then it constructs some decomposition of $f$.
	Hence the equality holds.
	The non-redundancy of $\mathcal{F}$ is clear by its construction.
	
	(2).
	Since $f$ is QR-M${}_2$-convex,
	so is $f_{pr}$ for $p \in [r-1]$.
	Hence, by \cref{prop:M_2pq}, Algorithm~2 does not output ``$f$ is not QR-M${}_2$-convex'' in Step~1.
	Let $\mathcal{F}^*$ be a non-redundant laminarizable $\mathcal{A}$-cut family
	and $c^*$ be a positive weight on $\mathcal{F}^*$
	that satisfy~\eqref{eq:(L, c)} for the given QR-M${}_2$-convex function $f$.
	We extend $c^*$ to a nonnegative weight on $2^{[n]}$ by defining $c^*(X) := 0$ for $X \notin \mathcal{F}$.
	We can assume that if $X \in \mathcal{F}$ and $Y \in \mathcal{F}^*$ satisfies $X \sim Y$
	then it holds $X = Y$.
	It suffices to prove (i) $c(X^*) = c^*(X^*)$ for $X^*$ obtained in the first half of Step~2,
	(ii)  $c(X_0) = c^*(X_0)$ for $X_0$ obtained in the latter half of Step~2,
	and (iii) $c(X) = c^*(X)$ for $X$ obtained in Step~3.
	Indeed, the properties (i)--(iii) imply $\mathcal{F} \subseteq \mathcal{F}^*$.
	Since $\mathcal{F}^*$ is laminarizable, so is $\mathcal{F}$
	and $|\mathcal{F}| \leq 2n$.
	Hence Algorithm~2 outputs $(\mathcal{F}, c)$ in Step~3.
	By the uniqueness of $\mathcal{F}^*$ under $\sim$ (\cref{thm:representation}),
	we can say $\mathcal{F} = \mathcal{F}^*$ and $c = c^*$.

	(i).
	Let $\lambda := \min\{ c'(X_0), c_{p_1r}(X_1), c_{p_2r}(X_2), \dots, c_{p_kr}(X_k) \}$.
	We prove $c^*(X^*) = \lambda$.
	It is easy to see that $c^*(X^*) \leq \lambda$ holds
	since, by \cref{lem:index restriction},
	we have $c'(X_0) \geq c^*(X^*)$ and $c_{p_i r}(X_i) \geq c^*(X^*)$ for $i \in [k]$.
	
	Suppose, to the contrary, that $c^*(X^*) < \lambda$ holds.
	Then the following holds:
	\begin{cl*}
		Every $Y_0 \in \mathcal{F}^*$ with $Y_0 \sim_{[r-1]} X^*$ satisfies $Y_0 \sim X^*$.
	\end{cl*}
	On the other hand,
	by \cref{lem:index restriction} with $c^*(X^*) < \lambda \leq c'(X_0)$,
	there must exist $Y_0 \in \mathcal{F}^*$ satisfying $Y_0 \sim_{[r-1]} X^*$ and $Y_0 \not\sim X^*$.
	This contradicts the statement of Claim,
	and hence $c^*(X^*) = \lambda$ holds, as required.
	
	We now prove Claim.
	\begin{proof}[Proof of Claim]
	Take any $Y_0 \in \mathcal{F}^*$ with $Y_0 \sim_{[r-1]} X^*$.
	By $c_{p_ir}(X_i) > c^*(X^*)$ ($i \in [k]$) and \cref{lem:index restriction},
	for every $i \in [k]$
	there is $Y \in \mathcal{F}^*$ with $Y \sim_{p_ir} X_i$ and $Y \not\sim X^*$.
	Take $Y \in \mathcal{F}^*$ satisfying $Y \not\sim X^*$ with $\{i \in [k] \mid Y \sim_{p_ir} X_i \}$ maximal among elements $Y' \in \mathcal{F}^*$ satisfying $Y' \not\sim X^*$.
	Let $I := \{i \in [k] \mid Y \sim_{p_ir} X_i \} (\neq \emptyset)$.
	By the maximality of $\ave{X_0}$ and $Y \not\sim X^*$,
	we have $[k] \setminus I \neq \emptyset$;
	otherwise $\ave{Y} \cap A_{[r-1]} \supsetneq \ave{X_0}$, contradicting the maximality of $\ave{X_0}$.

	Choose an arbitrary $j \in [k] \setminus I$.
	Then there is $Y_j \in \mathcal{F}^*$ with $Y_j \sim_{p_jr} X_j$ and $Y_j \not\sim X^*$.
	Furthermore, by the maximality of $I$,
	there is $i \in I$ such that $Y_j \not\sim_{p_ir} X_i$.
	Hence $Y_j \not\sim_{p_i} Y \sim_{p_i} X^*$ holds.
	In the following, we denote $Y$ by $Y_i$.

	Since $Y_i, Y_j, Y_0 \in \mathcal{F}^*$ and $\mathcal{F}^*$ is laminarizable,
	so is $\{Y_i, Y_j, Y_0\}$.
	Hence, by executing appropriate transformations for $\{Y_i, Y_j, Y_0\}$,
	we can make it laminar.
	We also denote the resulting laminar family by $\{Y_i, Y_j, Y_0\}$.
	We can assume $Y_i \cap A_{p_i} = Y_0 \cap A_{p_i} (\neq \emptyset)$ and $Y_j \cap A_{p_j} = Y_0 \cap A_{p_j} (\neq \emptyset)$.
	Indeed, $Y_i \cap A_{p_i} \neq Y_0 \cap A_{p_i}$ means $([n] \setminus Y_i) \cap A_{p_i} = Y_0 \cap A_{p_i}$.
	By the laminarity of $Y_i$ and $Y_0$, we have $Y_i \cap Y_0 = \emptyset$.
	Hence $\{[n] \setminus Y_i, Y_0\}$ is also laminar.
	Furthermore, note $Y_i \cap A_r = Y_j \cap A_r (\neq \emptyset)$ by $Y_i \sim_{p_ir} X_i$ and $Y_j \sim_{p_jr} X_j$.
	
	By $Y_0 \sim_{p_ip_j} X^* \not\sim_{p_ip_j} Y_i$ and laminarity,
	it holds that $Y_0 \cap A_{p_j} \supsetneq Y_i \cap A_{p_j}$ or $Y_0 \cap A_{p_j} \subsetneq Y_i \cap A_{p_j}$.
	Assume $Y_0 \cap A_{p_j} = Y_j \cap A_{p_j} \subsetneq Y_i \cap A_{p_j}$ (the argument for the other case is similar).
	Hence, by $Y_0 \cap Y_i \neq \emptyset$ and $Y_j \cap Y_i \neq \emptyset$,
	we have $Y_0 \subsetneq Y_i \supsetneq Y_j$.
	By $Y_j \not\sim_{p_i} Y_i \sim_{p_i} X$ and $Y_i \supsetneq Y_j$,
	we have $Y_0 \cap A_{p_i} = Y_i \cap A_{p_i} \supsetneq Y_j \cap A_{p_i}$.
	Hence $Y_i \supsetneq Y_0 \supsetneq Y_j$ holds.
	By $Y_i \cap A_r = Y_j \cap A_r$,
	it holds that $Y_i \cap A_r = Y_j \cap A_r = Y_0 \cap A_r$.
	This means $Y_0 \sim X^*$.
	\end{proof}
	
	(ii).
	By \cref{lem:index restriction}, it holds that
	\begin{align*}
	c'(X_0) &= \sum \{ c^*(Y) \mid Y \in \mathcal{F}^*,\ Y \sim_{[r-1]} X_0 \}\\
	&= c^*(X_0) + \sum \{ c^*(Y) \mid Y \in \mathcal{F}^*,\ \ave{Y} \supseteq A_r,\ X_0 \sim_{[r-1]} Y \}.
	\end{align*}
	Here the second term must be zero.
	Otherwise, by \cref{lem:index restriction},
	we would have found $X_1, X_2, \dots, X_k$ in Step~2.
	Therefore $c'(X_0) = c^*(X_0)$ holds.
	Thus we obtain $c(X_0) = c'(X_0) = c^*(X_0)$.
	
	(iii).
	We can show $c(X) = c_{pr}(X) = c^*(X)$ for any $p \in [r-1]$ and $X \in \mathcal{L}_{pr}$ by a similar argument as for (ii).

	(3).
	Note that $|\mathcal{F}'| = O(|A_{[r-1]}|)$ and $|\mathcal{L}_{pr}| = O(|A_{pr}|)$ for any $p \in [r-1]$.
	By the assumption $|A_r| \leq \min \{ |A_1|, |A_2|, \dots, |A_{r-1}| \}$,
	it holds that $r|A_r| = O(n)$.
	Step~1 can be done in $O(\sum_{p \in [r-1]}(|A_p| + |A_r|)^2) = O(\sum_{p \in [r-1]}|A_p|^2) = O(n^2)$ time by \cref{prop:M_2pq}.
	In Step~2,
	we first need to sort the elements in $\mathcal{F}'$ with respect to set-inclusion ordering in $O(|A_{[r-1]}| \log |A_{[r-1]}|) = O(n \log n)$ time (this is done only once).
	In each iteration,
	we search for $\{ X_1, X_2, \dots, X_k \}$ satisfying the conditions described in Step~2.
	This can be done in $O(|\bigcup_p \mathcal{L}_{pr}|) = O(n + r|A_r|) = O(n)$ time by using the structure of $\mathcal{L}_{p_ir}$ as follows.
	
	We first construct $\mathcal{F}_i$ from $\mathcal{L}_{p_ir}$ as $\mathcal{F}_i := \mathcal{G}_i \cup \overline{\mathcal{G}_i}$
	for all $i \in [k]$ in $O(|\bigcup_p \mathcal{L}_{pr}|) = O(n)$ time,
	where
	\begin{align*}
	\mathcal{G}_i &:= \{ X_i \cap A_r \mid X_i \in \mathcal{L}_{p_ir},\ X_i \cap A_{p_i} = X_0 \cap A_{p_i} \},\\
	\overline{\mathcal{G}_i} &:= \{ A_r \setminus X_i \mid X_i \in \mathcal{L}_{p_ir},\ A_{p_i} \setminus X_i = X_0 \cap A_{p_i} \}.
	\end{align*}
	Note that $F \cup (X_0 \cap A_{p_i}) \in \mathcal{L}_{p_ir}$ if $F \in \mathcal{G}_i$
	and $(A_r \setminus F) \cup (A_{p_i} \setminus X_0) \in \mathcal{L}_{p_ir}$ if $F \in \overline{\mathcal{G}_i}$.
	We can easily see that there exists $\{ X_1, X_2, \dots, X_k \}$ satisfying the conditions in Step~2 if and only if $\bigcap_{i \in [k]} \mathcal{F}_i \neq \emptyset$.
	By the laminarity of $\mathcal{L}_{p_ir}$,
	$\mathcal{F}_i$ is a chain,
	and can be represented as $\mathcal{F}_i = \{ F_i^1, F_i^2, \dots, F_i^{q_i} \}$ for $i \in [k]$,
	where $F_i^1 \supsetneq F_i^2 \supsetneq \cdots \supsetneq F_i^{q_i}$
	(this chain can be obtained while constructing $\mathcal{L}_{p_ir}$ in Algorithm~1).
	If $\bigcap_{i \in [k]} \mathcal{F}_i \neq \emptyset$, we can obtain $F \in \bigcap_{i \in [k]} \mathcal{F}_i$ in $O\left(\sum_i|\mathcal{F}_i|\right) = O(n)$ time.
	Indeed, take the maximal elements $F_1^1, F_2^1, \dots, F_k^1$ in $\mathcal{F}_1, \mathcal{F}_2, \dots, \mathcal{F}_k$, respectively.
	If all $i$ satisfy $F_i^1 = \bigcap_j F_j^1$,
	then output $\bigcap_j F_j^1$.
	Otherwise, for each $i$ with $F_i^1 \supsetneq \bigcap_j F_j^1$,
	update $\mathcal{F}_i \leftarrow \mathcal{F}_i \setminus \{F_i^1\}$,
	and do the same thing.
	By repeating this procedure,
	we can verify $\bigcap_{i \in [k]} \mathcal{F}_i = \emptyset$ or obtain $F \in \bigcap_{i \in [k]} \mathcal{F}_i$.
	From this $F$ in $\bigcap_{i \in [k]} \mathcal{F}_i$,
	we can easily construct the desired $X_i$ as
	\begin{align*}
	X_i =
	\begin{cases}
	F \cup (X_0 \cap A_{p_i}) & \textrm{if $F \in \mathcal{G}_i$},\\
	(A_r \setminus F) \cup (A_{p_i} \setminus X_0) & \textrm{if $F \in \overline{\mathcal{G}_i}$}
	\end{cases}
	\end{align*}
	for each $i \in [k]$.
	Thus we can find $\{ X_1, X_2, \dots, X_k \}$ satisfying the conditions in Step~2 in $O(n)$ time.
	
	Furthermore we can calculate $\min \{ c'(X_0), c_{p_1r}(X_1), c_{p_2r}(X_2), \dots, c_{p_kr}(X_k) \}$ in $O(k) = O(n)$ time.
	Since $|\mathcal{F}'| + |\bigcup_p \mathcal{L}_{pr}|$ decreases at least by one in each iteration in Step~2,
	the number of iterations in Step~2 is bounded by $O(|\mathcal{F}'| + |\bigcup_{p}\mathcal{L}_{pr}|) = O(n)$.
	Hence Step~2 can be done in $O(n^2)$ time.
	
	Step~3 can be done in $O(|\bigcup_p \mathcal{L}_{pr}| + n) = O(n)$ time.
	Hence the running-time of Algorithm~2 is bounded by $O(n^2)$.
\end{proof}

Our proposed algorithm for \DECOMP\ can be summarized as follows.
\begin{description}
	\item[Algorithm~3 (for \DECOMP):]
	\item[Step 0:]
	Rename $A_1, A_2, \dots, A_r$ so as to satisfy $|A_1| \geq |A_2| \geq \cdots \geq |A_r|$.
	\item[Step 1:]
	Execute Algorithm~1 for the restriction $f_{12}$.
	If Algorithm~1 returns ``$f_{12}$ is not QR-M${}_2$-convex,''
	then output ``$f$ is not QR-M${}_2$-convex''
	and stop.
	Otherwise, obtain $\mathcal{L}_{12}$ and $c_{12}$.
	\item[Step 2:]
	For $t = 3, \dots, r$,
	execute Algorithm~2 for $(\mathcal{F}_{[t-1]},c_{[t-1]})$,
	where $\mathcal{F}_{[2]} = \mathcal{L}_{12}$ and $c_{[2]} = c_{12}$.
	If Algorithm~2 returns ``$f_{[t]}$ is not QR-M${}_2$-convex,''
	output ``$f$ is not QR-M${}_2$-convex''
	and stop.
	Otherwise, obtain $(\mathcal{F}_{[t]},c_{[t]})$.
	\item[Step 3:]
	Output $(\mathcal{F}_{[r]},c_{[r]})$.~\qqed
\end{description}

\begin{thm}\label{thm:algo 3}
	Algorithm~{\rm 3} solves \DECOMP\ in $O(rn^2)$ time.
\end{thm}
\begin{proof}
	Step~0 can be done in $O(r\log r)$ time.
	Since the running-time of Algorithm~2 for $t$ is bounded by $O(|A_{[t+1]}|^2) = O(n^2)$ by \cref{prop:c_LL}~(3),
	the running-time of Algorithm~3 is bounded by $O(rn^2)$.
	
	The validity of Algorithm~3 can be proved as follows.
	Suppose that Algorithm~3 stops at Step~1 or Step~2.
	By \cref{prop:M_2pq} and \cref{prop:c_LL}~(2),
	$f$ is not QR-M${}_2$-convex.
	Hence Algorithm~3 works correctly.
	
	Suppose that Algorithm~3 reaches Step~3.
	Since $f_{[2]} \simeq_{[2]} \sum_{X \in \mathcal{F}_{[2]}} c_{[2]}(X) \ell_{X}$ by \cref{prop:M_2pq},
	we obtain $f_{[t]} \simeq_{[t]} \sum_{X \in \mathcal{F}_{[t]}} c_{[t]}(X) \ell_{X}$
	for all $t = 3, \dots, r$
	by \cref{prop:c_LL}~(1),
	where $\simeq_Q$ denotes the $\mathcal{A}_Q$-linear equivalence
	and this notation is used for $Q = [2], \dots, [r]$ here.
	Thus
	we have $f \simeq \sum_{X \in \mathcal{F}_{[r]}} c_{[r]}(X) \ell_{X}$ holds.
	Furthermore, if $f$ is QR-M${}_2$-convex,
	then $\mathcal{F}_{[2]} (= \mathcal{L}_{12})$ is laminarizable by \cref{lem:restriction} and \cref{prop:M_2pq}.
	Hence $\mathcal{F}_{[3]}, \dots, \mathcal{F}_{[r]}$ are laminarizable by \cref{prop:c_LL}~(2).
	Thus Algorithm~3 works correctly.
\end{proof}

\section{Algorithm for \LAM}\label{sec:(ii)}
For a VCSP-quadratic function $f$ of type $\mathcal{A} = \{A_1, A_2, \dots, A_r\}$,
suppose that we have obtained a non-redundant $\mathcal{A}$-cut family $\mathcal{F}$ by solving \DECOMP.
The next step for solving \TEST\ is to check for the laminarizability of $\mathcal{F}$.

Recall that a pair $X, Y \subseteq [n]$ is said to be {\it crossing}
if $X \cap Y$, $[n] \setminus (X \cup Y)$, $X \setminus Y$, and $Y \setminus X$ are all nonempty.
An $\mathcal{A}$-cut family $\mathcal{G}$ is said to be {\it cross-free}
if there is no crossing pair in $\mathcal{G}$.
From a cross-free $\mathcal{A}$-cut family $\mathcal{G}$,
we can easily construct a laminar $\mathcal{A}$-cut family $\mathcal{A}$-equivalent to $\mathcal{G}$
by switching $X \mapsto [n] \setminus X$ for appropriate $X \in \mathcal{G}$ (see e.g.,~\cite[Section~2.2]{book/KorteVygen10});
this can be done in $O(|\mathcal{G}|)$ time.
Furthermore, if $\mathcal{F}$ is laminarizable,
then we can always construct a cross-free family $\mathcal{A}$-equivalent to $\mathcal{F}$
without using transformation $X \mapsto [n] \setminus X$.
Thus our goal is to construct a cross-free family $\mathcal{A}$-equivalent to the input family $\mathcal{F}$
by repeating appropriate transformations for $X \in \mathcal{F}$ as
$X \mapsto X \cup A_p$ or $X \mapsto X \setminus A_p$ with some $A_p$ satisfying $\ave{X} \cap A_p = \emptyset$.
Recall that $\ave{X}$ denote the cutting support of $X$ defined in~\eqref{eq:ave}.

In this section,
we devise a polynomial-time algorithm for constructing a desired cross-free family.
Our algorithm makes use of weaker notions of cross-freeness,
called 2- and 3-local cross-freeness.
The existence of a cross-free family is characterized by the existence of a {\it 2-locally cross-free family} (\cref{subsec:LCF}).
The existence of a 2-locally cross-free family can be checked easily by solving a 2-SAT problem.
If a 2-locally cross-free family exists,
then a {\it 3-locally cross-free family} also exists,
and can be constructed in polynomial time (\cref{subsec:Const 3LCF}).
From a 3-locally cross-free family,
we can construct a desired cross-free family in polynomial time via the {\it uncrossing operation} (\cref{subsec:fLCF}).
Thus we solve \LAM.

\subsection{Preliminaries}\label{subsec:outline}
We use the following notations and terminologies.
For $X \in \mathcal{F}$, let $\overline{X} := [n] \setminus X$;
note $X \sim \overline{X}$ by~\eqref{eq:sim}.
For $\mathcal{A}$-cuts $X, Y, Z$, we define $\ave{XY} := \ave{X} \cap \ave{Y}$ and $\ave{XYZ} := \ave{X} \cap \ave{Y} \cap \ave{Z}$.
For $X \in \mathcal{F}$ and $Q \subseteq [r]$ with $A_Q \subseteq \ave{X}$,
the {\it partition line of $X$ on $A_Q$} is a bipartition $\{ X \cap A_Q, \overline{X} \cap A_Q \}$ of $A_Q$.
For $A \subseteq [n]$,
if $X \cap A \subseteq Y \cap A$ holds,
we say {\it $X \subseteq Y$ on $A$}.

Without loss of generality, we can assume the following:
\begin{itemize}
	\item
	$|\mathcal{F}|$ is at most $2n$.
	\item
	For distinct $X, Y \in \mathcal{F}$ with $\ave{XY} \neq \emptyset$,
	one of
	$X \subseteq Y$, $X \subseteq \overline{Y}$, $X \supseteq Y$, and $X \supseteq \overline{Y}$ holds on $\ave{XY}$.
	\item
	For all distinct $X, Y \in \mathcal{F}$,
	both $\ave{X} \setminus \ave{Y}$ and $\ave{Y} \setminus \ave{X}$ are nonempty.
\end{itemize}
If the first or the second condition fails,
then $\mathcal{F}$ is not laminarizable.
The third condition is satisfied by the following preprocessing.
For each $X \in \mathcal{F}$,
we add a new set $A_X$ with $|A_X| = 2$ to the ground set $[n]$ and to the partition $\mathcal{A}$ of $[n]$;
the ground set will be $[n] \cup \bigcup_{X \in \mathcal{F}} A_X$ and the partition will be $\mathcal{A} \cup \{ A_X \mid X \in \mathcal{F}\}$.
Define $X_+ := X \cup \{x\}$, where $x$ is one of the two elements of $A_X$ and
$\mathcal{F}_+ := \{ X_+ \mid X \in \mathcal{F} \}$.
Note $\ave{X_+} = \ave{X} \cup A_X$ and $\ave{X_+} \setminus \ave{Y_+} \neq \emptyset$ for all $X_+, Y_+ \in \mathcal{F}_+$.
Then it is easily seen that there exists a cross-free family $\mathcal{L}$ with $\mathcal{L} \sim \mathcal{F}$ if and only if there exists a cross-free family $\mathcal{L}_+$ with $\mathcal{L}_+ \sim \mathcal{F}_+$.
Furthermore we can construct the cross-free family $\mathcal{L}$ from $\mathcal{L}_+$ by restricting $\mathcal{L}_+$ to $[n]$,
that is, $\mathcal{L} = \{ L \cap [n] \mid L \in \mathcal{L}_+ \}$.

\subsection{2-local cross-freeness}\label{subsec:LCF}
For $A \subseteq [n]$,
a pair $X, Y \subseteq [n]$ is said to be {\it crossing on $A$}
if $(X \cap Y) \cap A$, $A \setminus (X \cup Y)$, $(X \setminus Y) \cap A$, and $(Y \setminus X) \cap A$ are all nonempty.
An $\mathcal{A}$-cut family $\mathcal{G}$ is said to be {\it cross-free on $A$}
if there is no crossing pair on $A$ in $\mathcal{G}$.
An $\mathcal{A}$-cut family $\mathcal{G}$ is called {\it 2-locally cross-free}
if no $X, Y \in \mathcal{G}$ are crossing on $\ave{X} \cup \ave{Y}$.
A cross-free family is 2-locally cross-free.
We denote the ordered pair $(X, Y)$ by $XY$.

Our goal of this subsection is to construct a 2-locally cross-free family $\mathcal{F}^*$
that is
$\mathcal{A}$-equivalent to the input $\mathcal{F}$ (if it exists).
Such $\mathcal{F}^*$ consists of $X^*$
that is obtained from each $X \in \mathcal{F}$
by adding or deleting some $A_p$ not intersecting with the cutting support $\ave{X}$ of $X$,
i.e.,
$X^* = \left( X \setminus \bigcup_{p \in I} A_{p}   \right) \cup  \left(\bigcup_{p \in J} A_{p} \right)$
for some $I, J \subseteq [r]$,
where $A_p \cap \ave{X} = \emptyset$ for all $p \in I \cup J$.
By the 2-local cross-freeness,
for each ordered pair $XY$ of members $X, Y$ in $\mathcal{F}$,
either one of the following holds:
\begin{itemize}
	\setlength{\leftskip}{5mm}
	\item[($XY$:0)]
	$X^*$ contains no $A_p$ contained in $\ave{Y} \setminus \ave{X}$,
	i.e., $X^* \cap (\ave{Y} \setminus \ave{X}) = \emptyset$.
	\item[($XY$:1)]
	$X^*$ contains every $A_p$ contained in $\ave{Y} \setminus \ave{X}$,
	i.e., $X^* \supseteq \ave{Y} \setminus \ave{X}$.
\end{itemize}
It turns out that a desired 2-locally cross-free family is obtained by specifying ($XY$:0) or ($XY$:1),
called the {\it label} of $XY$,
for all ordered pairs $XY$.
We observe that the labels satisfy the following properties:
\begin{itemize}
	\item
	Suppose that $\ave{XY} \neq \emptyset$ and the partition lines of $X,Y$ on $\ave{XY}$ are different.
	Then the labels of $XY$ and $YX$ are determined uniquely by their mutual configuration.
	For example,
	if $X \subsetneq Y$ on $\ave{XY}$,
	then we have $X^* \subseteq Y^*$ on $\ave{X} \cup \ave{Y}$,
	namely, ($XY$:0) and ($YX$:1) hold.
	Also if $X \subsetneq \overline{Y}$ on $\ave{XY}$,
	then we have $X^* \cap Y^* = \emptyset$ on $\ave{X} \cup \ave{Y}$;
	($XY$:0) and ($YX$:0) hold.
	Similarly for the remaining cases,
	$X \supsetneq Y$ or $X \supsetneq \overline{Y}$ on $\ave{XY}$.
	\item
	Suppose that $\ave{XY} \neq \emptyset$ and the partition lines of $X,Y$ on $\ave{XY}$ are the same.
	In this case,
	the labels of $XY$ and $YX$ are not uniquely determined.
	If the label of $YX$ is given,
	then the label of $XY$ is determined according to the mutual configuration of $X$ and $Y$ on $\ave{XY}$.
	For example,
	suppose that we have $X = Y$ on $\ave{XY}$.
	Then ($YX$:1) implies ($XY$:0) and vice versa.
	\item
	Suppose that $X,Y, Z \in \mathcal{F}$ satisfy $\ave{YZ} \setminus \ave{X} \neq \emptyset$.
	Then the labels of $XY$ and $XZ$ must be the same.
	Indeed, if ($XY$:1) holds,
	i.e., $X^* \supseteq \ave{Y} \setminus \ave{X}$,
	then $X^* \cap (\ave{Z} \setminus \ave{X})$ is nonempty on $\ave{X} \cup \ave{Z}$.
	This implies that ($XZ$:1) holds.
\end{itemize}

An {\it LC-labeling} $s$ for $\mathcal{F}$ is a function on the set of ordered pairs of distinct members in $\mathcal{F}$ satisfying the above properties, i.e.,
\begin{align}
(s(XY), s(YX)) &=
\begin{cases}
(0, 0) & \textrm{if $X \subsetneq \overline{Y}$ on $\ave{XY}$},\\
(0, 1) & \textrm{if $X \subsetneq Y$ on $\ave{XY}$},\\
(1, 0) & \textrm{if $X \supsetneq Y$ on $\ave{XY}$},\\
(1, 1) & \textrm{if $X \supsetneq \overline{Y}$ on $\ave{XY}$},
\end{cases}\label{eq:fixed}\\
s(XY) &=
\begin{cases}
s(YX) & \textrm{if $X \subseteq \overline{Y}$ or $X \supseteq \overline{Y}$ on $\ave{XY}$},\\
1 - s(YX) & \textrm{if $X \subseteq Y$ or $X \supseteq Y$ on $\ave{XY}$},
\end{cases}\label{eq:swapped}\\
s(XY) &= s(XZ) \quad \textrm{if $\ave{YZ} \setminus \ave{X} \neq \emptyset$},\label{eq:prefixed}
\end{align}
where~\eqref{eq:fixed} and~\eqref{eq:swapped} apply only when $\ave{XY} \neq \emptyset$.
Here LC stands for Local Cross-freeness.

From the definition,
it is obvious that
any 2-locally cross-free family $\mathcal{F}^*$ that is $\mathcal{A}$-equivalent to $\mathcal{F}$ (without taking complements) gives rise to an LC-labeling $s$ for $\mathcal{F}$.
Indeed,
define $s(XY) := 0$ if $X^*$ is in case ($XY$:0)
and $s(XY) := 1$ if $X^*$ is in case ($XY$:1).
The converse is also possible.
Let $s$ be an LC-labeling for $\mathcal{F}$.
Consider the following procedure for each $X \in \mathcal{F}$:
For each $A_p \in \mathcal{A}$ with $A_p \subseteq \ave{Y} \setminus \ave{X}$ for some $Y$,
if $s(XY) = 1$, then add $A_p$ to $X$,
and if $s(XY) = 0$, then delete $A_p$ from $X$.
Let $X^s$ denote the resulting set.
Thanks to the condition~\eqref{eq:prefixed},
this procedure is independent of the choice of $Y$ and is well-defined.
Accordingly, define $\mathcal{F}^s$ by
\begin{align}\label{eq:F^s}
\mathcal{F}^s := \{ X^s \mid X \in \mathcal{F} \}.
\end{align}
Then $\mathcal{F}^s$ is indeed 2-locally cross-free.
To see this,
it suffices to consider $X,Y$ with $\ave{XY} \neq \emptyset$.
By~\eqref{eq:fixed} and~\eqref{eq:swapped},
it holds $X^s \subseteq Y^s$, $X^s \supseteq Y^s$, $X^s \cap Y^s = \emptyset$, or $(\ave{X} \cup \ave{Y}) \setminus (X^s \cup Y^s) = \emptyset$ on $\ave{X} \cup \ave{Y}$.
Thus the following holds.
\begin{prop}\label{prop:LCF}
	There exists a 2-locally cross-free family $\mathcal{A}$-equivalent to $\mathcal{F}$
	if and only if there exists an LC-labeling $s$ for $\mathcal{F}$.
	To be specific,
	$\mathcal{F}^s$ is a 2-locally cross-free family $\mathcal{A}$-equivalent to $\mathcal{F}$.
\end{prop}

In order to find an LC-labeling in a greedy fashion,
we introduce the {\it LC-graph},
which is also utilized for constructing a 3-locally cross-free family $\mathcal{A}$-equivalent to $\mathcal{F}$ in \cref{subsec:Const 3LCF}.
The {\it LC-graph} $G(\mathcal{F}) = ( V(\mathcal{F}), E_\textrm{s} \cup E_\textrm{p} )$ of the input $\mathcal{F}$ is defined by
\begin{align*}
V(\mathcal{F}) &:= \{ XY \mid X, Y \in \mathcal{F}, X \neq Y\},\\
E_\textrm{s} &:= \{ \{ XY, YX \} \mid \ave{XY} \neq \emptyset \},\\
E_\textrm{p} &:= \{ \{XY, XZ\} \mid Y \neq Z,\ \ave{YZ} \setminus \ave{X} \neq \emptyset \}.
\end{align*}
Note that the structure of LC-graph depends only on the family $\{ \ave{X} \mid X \in \mathcal{F} \}$ of cutting supports.
We call an edge $e \in E_\textrm{s}$ a {\it swapped edge},
which corresponds to~\eqref{eq:fixed} and~\eqref{eq:swapped},
and an edge $e \in E_\textrm{p}$ a {\it prefixed edge},
which corresponds to~\eqref{eq:prefixed}.
By the second assumption mentioned in \cref{subsec:outline},
exactly two types of swapped edges $e = \{ XY, YX \}$ can be distinguished;
(i) $X \subseteq Y$ or $X \supseteq Y$ on $\ave{XY}$
and (ii) $X \subseteq \overline{Y}$ or $X \supseteq \overline{Y}$ on $\ave{XY}$.
The former type of swapped edges will be called {\it flipping} (since $s(XY) = 1 - s(YX)$),
and the latter type {\it non-flipping} (since $s(XY) = s(YX)$).
See Figure~\ref{fig:instance LC graph} for an example of LC-graph.

\begin{figure}
	\centering
		\includegraphics[width=10cm]{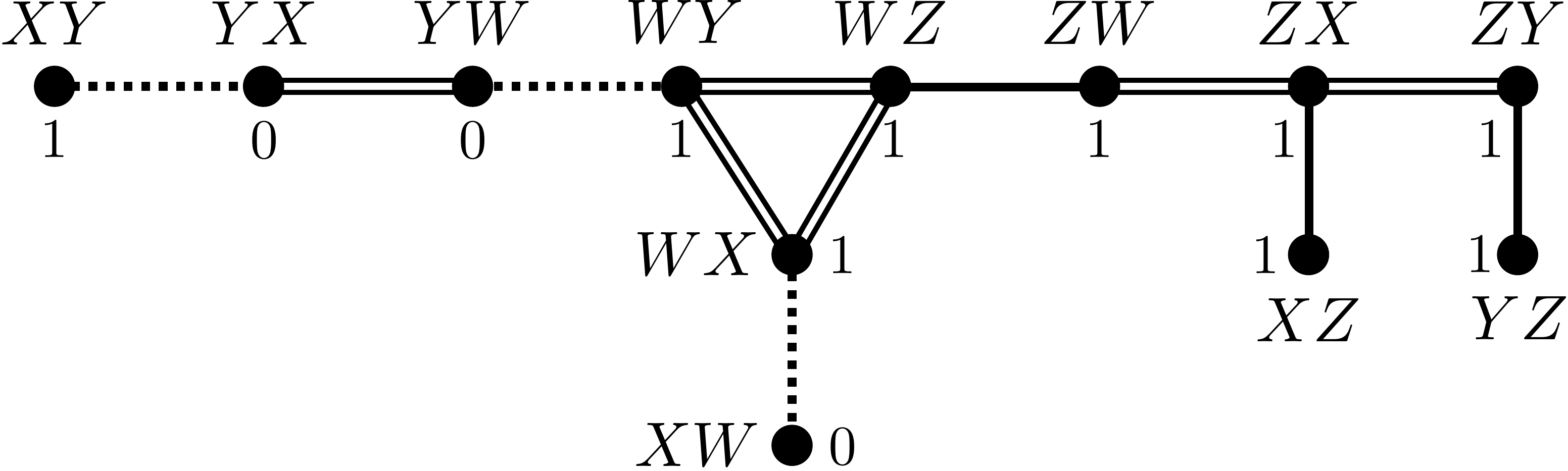}
		\caption{The LC-graph $G(\mathcal{F})$ for $\mathcal{F} = \{ X,Y,Z,W \}$ in \cref{ex:instance 2-local},
			where the edges denoted by double lines are prefixed edges,
			and the others are swapped edges.
			Flipping and non-flipping swapped edges are denoted by dotted and solid line, respectively.
			The numbers $0/1$ at the nodes denote the LC-labeling $s$ in case of setting $s(XY) = 1$.}
		\label{fig:instance LC graph}
\end{figure}

An LC-labeling is nothing but a feasible solution for the 2-SAT problem defined by the constraints~\eqref{eq:fixed}--\eqref{eq:prefixed}.
Therefore we can check the existence of an LC-labeling $s$ greedily in $O(|E_\textrm{s} \cup E_\textrm{p}|) = O(n^4)$ time.
Node $XY \in V(\mathcal{F})$ is said to be {\it fixed} if the value of an LC-labeling $s$ for $XY$ is determined as~\eqref{eq:fixed},
that is, if $\ave{XY} \neq \emptyset$ and the partition lines of $X$ and $Y$ on $\ave{XY}$ are different,
and $XY$ is said to be {\it defined} if the value of $s(XY)$ has been defined.
The algorithm is as follows.
\begin{enumerate}
	\item For each fixed node $XY$,
	define $s(XY)$ according to~\eqref{eq:fixed}.
	\item 
	In each connected component of $G(\mathcal{F})$, execute a 
	breadth-first search from a defined node $XY$,
	and define $s(ZW)$ for all reached nodes $ZW$ according to~\eqref{eq:swapped} and~\eqref{eq:prefixed}.
	If a conflict in value assignment to $s(ZW)$ is detected during this process,
	output ``there is no LC-labeling.''
	\item If there is an undefined node,
	choose any undefined node $XY$,
	and define $s(XY)$ as $0$ or $1$ arbitrarily.
	Then go to 2.
\end{enumerate}

\begin{exmp}\label{ex:instance 2-local}
	We consider the family $\mathcal{F}$ obtained in \cref{ex:instance decompotion}.
	After applying the preprocessing to $\mathcal{F}$,
	it holds
	$\mathcal{F} = \{ X, Y, Z, W \}$,
	where $X := 1357a$, $Y := 135b$, $Z := 24c$, and $W := 37d$
	with the partition $\mathcal{A} = \{ 12, 34, 56, 78, aa', bb', cc', dd' \}$ of the ground set $N := 12345678aa'bb'cc'dd'$.
	The LC-graph $G(\mathcal{F})$ is illustrated in Figure~\ref{fig:instance LC graph}.
	
	We obtain an LC-labeling $s : V(\mathcal{F}) \rightarrow \{0,1\}$ by defining, for example, $s(XY) := 1$.
	According to~\eqref{eq:fixed}--\eqref{eq:prefixed},
	the all labels are determined as
	$s(X'Y') = 0$ for $X'Y' \in \{ YX, YW, XW \}$
	and $s(X'Y') = 1$ otherwise.
	Then $X^s = 1357abb'cc'$, $Y^s = 135bcc'$,
	$Z^s = 245678aa'bb'cdd'$,
	$W^s = 123567aa'bb'cc'd$,
	and
	$\mathcal{F}^s$ is a cross-free family with $\mathcal{F}^s \sim \mathcal{F}$.
	Thus $\mathcal{F}' := \{ X^s, Y^s, N \setminus Z^s, N \setminus W^s \}$ is a laminar family with $\mathcal{F}' \sim \mathcal{F}$.
	
	Recall that the original $\mathcal{F}$ is a family of subsets of $12345678$.
	Let $\mathcal{L}$ be the family of $\mathcal{F}'$ restricted to $12345678$,
	i.e., $\mathcal{L} = \{ 1357, 135, 13, 48 \}$,
	which is the same one as the family inducing M-convex summand $f_1$ defined in~\eqref{eq:f_1 instance};
	see also Figure~\ref{fig:Laminarization}.~\qqed
\end{exmp}

\subsection{3-local cross-freeness}\label{subsec:fLCF}
An $\mathcal{A}$-cut family $\mathcal{G}$ is called {\it 3-locally cross-free} if $\mathcal{G}$ is 2-locally cross-free
and $\{X, Y, Z\}$ is cross-free on the union of the cutting supports $\ave{X} \cup \ave{Y} \cup \ave{Z}$
for all $X, Y, Z \in \mathcal{G}$ that have a nonempty intersection of the cutting supports,
i.e., $\ave{XYZ} \neq \emptyset$.
A cross-free family is 3-locally cross-free,
and a 3-locally cross-free family is 2-locally cross-free,
whereas the converse is not true (see \cref{remark:not3but2}).
We write $X \subseteq^* Y$ to mean $X \subseteq Y$ on $\ave{X} \cup \ave{Y}$.

Our objective of this subsection is to give an algorithm for
constructing a desired cross-free family
from a 3-locally cross-free family $\mathcal{A}$-equivalent to the input $\mathcal{F}$.
The algorithm consists of repeated applications of an elementary operation
that preserves 3-local cross-freeness.
The operation is defined by \eqref{eq:G'} below, and is referred to as
the {\it uncrossing operation} to $X, Y$.
By the 2-local cross-freeness of $\mathcal{G}$,
the two cases in~\eqref{eq:G'} exhaust all possibilities for $X, Y \in \mathcal{G}$.
\begin{prop}\label{prop:fLCF close}
	Suppose that $\mathcal{G}$ is 3-locally cross-free.
	For $X, Y \in \mathcal{G}$,
	define
	\begin{align}\label{eq:G'}
	\mathcal{G}' :=
	\begin{cases}
	\mathcal{G} \setminus \{X, Y\} \cup \{ X \cap Y, X \cup Y \} & {\rm if\ } X \subseteq^* Y {\rm\ or\ } Y \subseteq^* X,\\
	\mathcal{G} \setminus \{X, Y\} \cup \{ X \setminus Y, Y \setminus X \} & {\rm if\ } X \subseteq^* \overline{T} {\rm\ or\ } \overline{Y} \subseteq^* X.
	\end{cases}
	\end{align}
	Then $\mathcal{G}'$ is a 3-locally cross-free family $\mathcal{A}$-equivalent to $\mathcal{G}$.
\end{prop}
The proof of \cref{prop:fLCF close} is given at the end of this subsection.

\begin{description}
	\item[Algorithm~4 (for constructing a cross-free family):]
	\item[Input:] A 3-locally cross-free family $\mathcal{G}$.
	\item[Step 1:]
	While there is a crossing pair $X, Y$ in $\mathcal{G}$,
	apply the uncrossing operation to $X, Y$ and modify $\mathcal{G}$ accordingly.
	\item[Step 2:] Output $\mathcal{G}$.~\qqed
\end{description}

\begin{prop}\label{prop:algo 4}
	Algorithm~{\rm 4} runs in $O(n^2)$ time, and the output $\mathcal{G}$ is cross-free.
\end{prop}
\begin{proof}
	The number of crossing pairs in input $\mathcal{G}$ is at most $O(n^2)$ (since $|\mathcal{G}| = O(n)$).
	Take any $\{X, Y\} \subseteq \mathcal{G}$ which is crossing.
	Since the replacement $X \mapsto \overline{X}$ or $Y \mapsto \overline{Y}$ does not change the (non-)cross-freeness of $\{X, Y\}$, $\{X, Z\}$, and $\{Y, Z\}$ for $Z \in \mathcal{G}$,
	we can assume $X \subseteq^* Y$ or $Y \subseteq^* X$ by appropriate replacement.
	Let $\mathcal{G}'$ be the family resulting from the uncrossing operation on $X, Y$.
	Then it is easily verified that, for any $Z \in \mathcal{G} \setminus \{X, Y\}$,
	the number of crossing pairs in $\{\{ X \cap Y, Z \}, \{X \cup Y, Z\} \}$ is at most that in $\{ \{X, Z\}, \{Y,Z\} \}$.	
	Since $\{X \cap Y, X \cup Y\}$ is not crossing, the number of crossing pairs decreases at least by one.
	Furthermore, by \cref{prop:fLCF close},
	$\mathcal{G}'$ is also a 3-locally cross-free family $\mathcal{A}$-equivalent to $\mathcal{F}$.
	Eventually, we arrive at
	a cross-free family $\mathcal{A}$-equivalent to $\mathcal{F}$.
	The above process involves at most $O(n^2)$ uncrossing operations.
\end{proof}

\begin{remark}\label{remark:not3but2}
	It is worth mentioning that the uncrossing operation does not preserve
	2-local cross-freeness.
	For example,
	we define $X := 1356$, $Y := 1347$, and $Z := 1578$
	with a partition $\{12, 34, 56, 78\}$.
	Note that $\{X, Y, Z\}$ is not 3-locally cross-free but 2-locally cross-free.
	
	We consider to execute the uncrossing operation to $X, Y$.
	Then the resulting family is $\{ X \cap Y, X \cup Y, Z \}$.
	Since $X \cap Y = 13$ and $Z = 1578$,
	$\{X \cap Y, Z\}$ is crossing on $\ave{X \cap Y} \cup \ave{Z} = 123456$.~\qqed
\end{remark}

The rest of this subsection is devoted to the proof of \cref{prop:fLCF close}.
We first note the following facts,
which are also used in the proof of \cref{prop:algo 5} in \cref{subsec:Const 3LCF}.
\begin{lem}\label{lem:cross-free on S}
	Let $\mathcal{G}$ be a 2-locally cross-free family.
	A triple $\{X, Y, Z\} \subseteq \mathcal{G}$ is cross-free on $\ave{X} \cup \ave{Y} \cup \ave{Z}$
	if one of the following conditions holds:
	\begin{description}
		\item[{\rm (1)}] $\ave{XY} \neq \emptyset$, and $\{X, Y\}$ is cross-free on $\ave{X} \cup \ave{Y} \cup \ave{Z}$.
		\item[{\rm (2)}] $\ave{XY} \not\subseteq \ave{Z}$, and $\ave{XZ}$ or $\ave{YZ}$ is nonempty.
		\item[{\rm (3)}] The partition lines of $X, Y, Z$ on $\ave{XYZ}$ are not the same.
		\item[{\rm (4)}] $\ave{XY} = \ave{ZY} \neq \emptyset$,
		and there is a path $(XY, XY_1, \dots, XY_k)$ in $G(\mathcal{G})$
		such that $\{X, Y_k, Z\}$ is cross-free on $\ave{X} \cup \ave{Y_k} \cup \ave{Z}$.
	\end{description}
\end{lem}
\begin{proof}
	Let $S := \ave{X} \cup \ave{Y} \cup \ave{Z}$.
	Note that $\{X, Y\}$ is 2-locally cross-free if and only if
	so is $\{\overline{X}, Y\}$.
	Hence, by appropriate replacement $X \mapsto \overline{X}$ and/or $Y \mapsto \overline{Y}$,
	we can assume $X \subseteq^* Y$; we often use such replacement in this proof.
	
	(1).
	By symmetry,
	it suffices to show that $\{X, Z\}$ is cross-free on $S$.
	We assume $X \subseteq^* Z$ (the argument for the case of $Z \subseteq^* X$ is similar).
	There are two cases: (i) $\ave{XY} \setminus \ave{Z} \neq \emptyset$ and (ii) $(\emptyset \neq) \ave{XY} \subseteq \ave{Z}$.
	Note that $X \subseteq^* Z$ implies $Z \supseteq \ave{X} \setminus \ave{Z}$ and $X \cap (\ave{Z} \setminus \ave{X}) = \emptyset$.
	
	(i). By the 2-local cross-freeness of $\{Y, Z\}$ and $\ave{XY} \setminus \ave{Z} \neq \emptyset$,
	$Z \supseteq \ave{X} \setminus \ave{Z}$ implies $Z \supseteq \ave{Y} \setminus \ave{Z}$,
	and hence $Z \supseteq (\ave{X} \cup \ave{Y}) \setminus \ave{Z}$ holds.
	Thus $X \subseteq Z$ holds on $S$.
	
	(ii). We can assume $Y \subseteq X$ or $X \subseteq Y$ on $S$.
	Then, by $\emptyset \neq \ave{XY} = \ave{XYZ}$ and the 2-local cross-freeness of $\{Y, Z\}$,
	we have $Y \subseteq^* Z$ or $Z \subseteq^* Y$.
	If $Y \subseteq^* Z$,
	then $Z \supseteq \ave{Y} \setminus \ave{Z}$ holds on $S$.
	Hence $X \subseteq Z$ holds on $S$.
	If $Z \subseteq^* Y$,
	then $X \subseteq Y$ must hold on $S$ by $X \subseteq^* Z$.
	This means $X \subseteq^* Y$, i.e., $X \cap (\ave{Y} \setminus \ave{X}) = \emptyset$.
	Hence $X \subseteq Z$ holds on $S$.

	(2).
	We can assume $X \subseteq^* Y$ and $\ave{XZ} \neq \emptyset$.
	By $\ave{XY} \not\subseteq \ave{Z}$ and $\ave{XZ} \neq \emptyset$,
	there are two cases: (i) $\ave{XZ} \not\subseteq \ave{Y}$ or (ii) $(\emptyset \neq) \ave{XZ} \subsetneq \ave{XY}$.
	
	(i).
	$X \subseteq^* Y$ implies
	$Y \supseteq \ave{X} \setminus \ave{Y}$.
	By $\ave{XZ} \not\subseteq \ave{Y}$,
	we have $Y \cap (\ave{Z} \setminus \ave{Y}) \neq \emptyset$.
	Hence, by the 2-local cross-freeness of $\{Y, Z\}$,
	$Y$ must contain $\ave{Z} \setminus \ave{Y}$.
	Therefore, it holds that $X \subseteq Y$ on $S$;
	then we use (1) (note $\ave{XY} \neq \emptyset$).
	
	(ii).
	We assume $X \subseteq^* Z$ by the 2-local cross-freeness of $\{X, Z\}$
	(the argument for the case of $Z \subseteq^* X$ is similar).
	This implies $Z \supseteq \ave{X} \setminus \ave{Z}$.
	By $\emptyset \neq \ave{XZ} \subsetneq \ave{XY}$,
	we have $Z \cap (\ave{Y} \setminus \ave{Z}) \neq \emptyset$.
	Hence, by the 2-local cross-freeness of $\{Y, Z\}$,
	$Z$ must contain $\ave{Y} \setminus \ave{Z}$.
	Therefore, it holds that $X \subseteq Z$ on $S$;
	then we use (1).

	(3).
	Note that $\ave{XY}$, $\ave{YZ}$, and $\ave{ZX}$ are all nonempty.
	We can assume that both $X$ and $Y$ properly contain $Z$ in $\ave{XYZ}$.
	Necessarily $Z$ is disjoint from $(\ave{X} \cup \ave{Y}) \setminus \ave{Z}$
	by the 2-local cross-freeness of $\{X,Z\}$ and $\{Y,Z\}$.
	Hence $\{X, Z\}$ (or $\{Y, Z\}$) is cross-free on $S$;
	then we use (1).
	
	(4).
	We can assume $X \subseteq^* Y$ by the 2-local cross-freeness of $\{X, Y\}$.
	Then we can also assume $X \subseteq^* Z$ or $Z \subseteq^* X$.
	If $X \subseteq^* Z$,
	then $X$ does not meet $(\ave{Y} \cup \ave{Z}) \setminus \ave{X}$,
	and $\{X, Y\}$ is cross-free on $S$;
	then we use (1).
	Hence suppose $Z \subseteq^* X$.
	By $X \subseteq^* Y$ and the 2-local cross-freeness of $\{X, Y_i\}$ for $i \in [k]$,
	it must hold that $X \subseteq^* Y_i$ for $i \in [k]$.
	Since $\{X, Y_k, Z\}$ is cross-free on $\ave{X} \cup \ave{Y_k} \cup \ave{Z}$,
	it holds that $Z \subseteq X \subseteq Y_k$ on $\ave{X} \cup \ave{Y_k} \cup \ave{Z}$.
	Here $\ave{Z}$ cannot meet $\ave{Y_i} \setminus \ave{X}$,
	since
	otherwise sequence $XY, XY_1, \dots, XY_i, XZ$ also forms a path in $G(\mathcal{G})$
	and hence it holds that $X \subseteq^* Z$,
	a contradiction to $Z \subseteq^* X$.
	By this fact together with $\ave{Y_iY_{i+1}} \setminus \ave{X} \neq \emptyset$,
	we can say $\ave{Y_iY_{i+1}} \setminus \ave{Z} \neq \emptyset$.
	Hence, by $\ave{XY} = \ave{ZY}$,
	the sequence $ZY, ZY_1, \dots, ZY_k$ also forms a path in $G(\mathcal{G})$.
	By $Z \subseteq^* Y_k$ and the 2-local cross-freeness of $\{Z, Y_i\}$ for $i \in [k]$,
	we have $Z \subseteq^* Y$.
	Now $Z \subseteq^* X$ and $Z \subseteq^* Y$ hold.
	This means that $Z$ does not meet $(\ave{X} \cup \ave{Y}) \setminus \ave{Z}$,
	which implies that $\{Y, Z\}$ is cross-free on $S$;
	then we use (1).
\end{proof}

We are now ready to give the proof of \cref{prop:fLCF close}.
\begin{proof}[Proof of \cref{prop:fLCF close}]
	We only prove that if $X \subseteq^* Y$,
	then $\mathcal{G}' := \mathcal{G} \setminus \{ X, Y \} \cup \{ X \cap Y, X \cup Y \}$ is 3-locally cross-free with $\mathcal{G}' \sim \mathcal{G}$; the other case is similar.
	
	First we prove $\mathcal{G}' \sim \mathcal{G}$,
	that is,
	we show $X \sim X \cap Y$ and $Y \sim X \cup Y$.
	By $X \subseteq^* Y$, we have $X = X \cap Y$ on $\ave{X} \cup \ave{Y}$ and $Y = X \cup Y$ on $\ave{X} \cup \ave{Y}$.
	Furthermore,
	for any $p \in [r]$ with $A_p \cap (\ave{X} \cup \ave{Y}) = \emptyset$,
	$X \cap Y \supseteq A_p$ or $(X \cap Y) \cap A_p = \emptyset$ holds and $X \cup Y \supseteq A_p$ or $(X \cup Y) \cap A_p = \emptyset$ holds.
	This means $X \sim X \cap Y$ and $Y \sim X \cup Y$;
	then $\ave{X} = \ave{X \cap Y}$ and $\ave{Y} = \ave{X \cup Y}$ follow.
	
	Next we show that $\mathcal{G}'$ is 2-locally cross-free.
	Since the partition lines of $X$ and $Y$ are the same as those of $X \cap Y$ and $X \cup Y$,
	$\{ X \cap Y, X \cup Y \}$ is also cross-free on $\ave{X} \cup \ave{Y}$.
	Hence $\{X \cap Y, X \cup Y\}$ is 2-locally cross-free.
	In the following,
	we prove that
	$\{ X \cap Y, X \cup Y, Z \}$ is 2-locally cross-free for each $Z \in \mathcal{G} \setminus \{X,Y \}$.

	If $\{X, Y\}$ is cross-free on $\ave{X} \cup \ave{Y} \cup \ave{Z}$,
	then the partition lines of $X$ and $Y$ on $\ave{X} \cup \ave{Y} \cup \ave{Z}$ are the same as those of $X \cap Y$ and $X \cup Y$.
	Hence, by the 2-local cross-freeness of $\mathcal{G}$,
	we obtain that $\{X \cap Y, X \cup Y, Z\}$ is also 2-locally cross-free.
	Therefore, it suffices to deal with the cases of (i) $\ave{XZ} = \ave{YZ} = \emptyset$,
	(ii) $\ave{XZ} \neq \emptyset$ and $\ave{XY} = \ave{YZ} = \emptyset$,
	(iii) $\ave{YZ} \neq \emptyset$ and $\ave{XY} = \ave{XZ} = \emptyset$,
	and (iv) $\ave{XY} = \ave{YZ} = \ave{ZX} \neq \emptyset$.
	Indeed, for other cases,
	$\{X, Y, Z\}$ is cross-free on $\ave{X} \cup \ave{Y} \cup \ave{Z}$
	by \cref{lem:cross-free on S} (2),
	reducing to the cross-free case above.
	
	(i).
	By the 2-local cross-freeness of $\mathcal{G}$,
	we have both [$X \supseteq \ave{Z}$ or $X \cap \ave{Z} = \emptyset$] and [$Y \supseteq \ave{Z}$ or $Y \cap \ave{Z} = \emptyset$].
	Hence both [$(X \cap Y) \supseteq \ave{Z}$ or $(X \cap Y) \cap \ave{Z} = \emptyset$] and [$(X \cup Y) \supseteq \ave{Z}$ or $(X \cup Y) \cap \ave{Z} = \emptyset$] hold.
	Therefore $\{X \cap Y, X \cup Y, Z\}$ is 2-locally cross-free.
	
	(ii) and (iii).
	By symmetry, we show (ii) only.
	By $X \subseteq^* Y$,
	we have $Y \supseteq \ave{X} \setminus \ave{Y}$.
	By $\ave{XZ} \neq \emptyset$ and $\ave{XY} = \ave{YZ} = \emptyset$,
	it holds that $Y \cap (\ave{Z} \setminus \ave{Y}) \neq \emptyset$.
	By the 2-local cross-freeness of $\{Y, Z\}$,
	$Y$ must contain $\ave{Z} \setminus \ave{Y}$.
	Therefore $X \subseteq Y$ holds on $\ave{X} \cup \ave{Y} \cup \ave{Z}$, reducing to the cross-free case.
	
	(iv). $\ave{XY} = \ave{YZ} = \ave{ZX} \neq \emptyset$ implies $\ave{XYZ} \neq \emptyset$.
	Hence, by the 3-local cross-freeness of $\mathcal{G}$,
	$\{X, Y, Z\}$ is cross-free on $\ave{X} \cup \ave{Y} \cup \ave{Z}$, reducing to the cross-free case.

	Finally,
	we show that $\mathcal{G}'$ is 3-locally cross-free.
	Take distinct $S, T, U \in \mathcal{G}'$ with $\ave{STU} \neq \emptyset$.
	If $\{ S, T, U\} \cap \{ X \cap Y, X \cup Y \} = \emptyset$,
	then $\{ S, T, U\}$ does not change in the construction of $\mathcal{G}'$.
	Hence $\{ S, T, U\}$ is cross-free on $\ave{S} \cup \ave{T} \cup \ave{U}$.
	If $|\{ S, T, U\} \cap \{ X \cap Y, X \cup Y \}| = 1$,
	then $\{ S, T, U\} \setminus \{ X \cap Y, X \cup Y \}$ is cross-free on $\ave{S} \cup \ave{T} \cup \ave{U}$.
	By the 2-local cross-freeness of $\mathcal{G}'$ shown above and \cref{lem:cross-free on S} (1),
	$\{ S, T, U\}$ is also cross-free on $\ave{S} \cup \ave{T} \cup \ave{U}$.
	If $|\{ S, T, U\} \cap \{ X \cap Y, X \cup Y \}| = 2$ (assume $S = X \cap Y$ and $T = X \cup Y$),
	then the partition lines of $X \cap Y$ and $X \cup Y$ on $\ave{X} \cup \ave{Y} \cup \ave{U}$ do not change in the construction of $\mathcal{G}'$,
	since $\{X, Y, U\}$ is cross-free on $\ave{X} \cup \ave{Y} \cup \ave{U}$.
	Thus $\{ X \cap Y, X \cup Y, U\}$ is cross-free on $\ave{X} \cup \ave{Y} \cup \ave{U}$ = $\ave{X \cap Y} \cup \ave{X \cup Y} \cup \ave{U}$.
	This completes the proof of \cref{prop:fLCF close}.
\end{proof}

\subsection{Constructing 3-locally cross-free family}\label{subsec:Const 3LCF}
Our final task is to show that,
for an input $\mathcal{F}$ that is $\mathcal{A}$-equivalent to a 2-locally cross-free family,
we can always construct a 3-locally cross-free family in polynomial time.
Specifically, we use the LC-graph $G(\mathcal{F})$ introduced in \cref{subsec:LCF},
and construct an LC-labeling $s$ with the property that
the family $\mathcal{F}^s$ in~\eqref{eq:F^s}
transformed from $\mathcal{F}$ by $s$
is 3-locally cross-free.
While the existence of an LC-labeling is guaranteed by
the assumed $\mathcal{A}$-equivalence of $\mathcal{F}$ to a 2-locally cross-free family (\cref{prop:LCF}),
we need to exploit a certain intriguing structure inherent
in an LC-graph before we can construct such a special LC-labeling.

\cref{lem:cross-free on S} indicates that, more often than not,
a triple $X, Y, Z$ in any 2-locally cross-free family is cross-free on $\ave{X} \cup \ave{Y} \cup \ave{Z}$.
To construct a 3-locally cross-free family, particular cares are needed for those triples $X, Y, Z$ with $\ave{XY} = \ave{YZ} = \ave{ZX} \neq \emptyset$
for which there exists no path $(XY, XY_1, \dots, XY_k)$ satisfying $\ave{XY} \neq \ave{XY_k} \neq \emptyset$.
Indeed,
suppose that $\ave{XY}$, $\ave{YZ}$, and $\ave{XZ}$ are nonempty.
If $\ave{XY} \neq \ave{YZ}$,
then it holds that $\ave{XY} \not\subseteq \ave{Z}$ or $\ave{YZ} \not\subseteq \ave{X}$.
Hence, by \cref{lem:cross-free on S}~(2),
$\{X,Y,Z\}$ is cross-free on $\ave{X} \cup \ave{Y} \cup \ave{Z}$.
If $\ave{XY} = \ave{YZ} = \ave{ZX} \neq \emptyset$ and there is a path $(XY, XY_1, \dots, XY_k)$ satisfying $\ave{XY} \neq \ave{XY_k} \neq \emptyset$,
then, by the above argument for $\ave{XY_k} \neq \ave{YZ}$,
$\{X, Y_k, Z\}$ is cross-free on $\ave{X} \cup \ave{Y_k} \cup \ave{Z}$.
Hence, by \cref{lem:cross-free on S}~(4),
$\{X,Y,Z\}$ is cross-free on $\ave{X} \cup \ave{Y} \cup \ave{Z}$.

This motivates the notion of {\it special nodes} and {\it special connected components} in the LC-graph $G(\mathcal{F})$.
For distinct $X,Y \in \mathcal{F}$, define
\begin{align*}
R(XY) &:= \{ Z \in \mathcal{F} \mid \textrm{There is a path } (XY, XY_1, \dots, XZ) \textrm{ using only prefixed edges}\},\\
R^*(XY) &:= \{ Z \in R(XY) \mid \ave{XZ} \neq \emptyset \}.
\end{align*}
We say that a node $XY$ (or an ordered pair of $X$ and $Y$) with $\ave{XY} \neq \emptyset$ is {\it special} if $\ave{XZ} = \ave{XY}$ holds for all $Z \in R^*(XY)$.
For $X, Y \in \mathcal{F}$ with $XY$ and $YX$ both being special,
let $v(XY)$ denote the connected component (as a set of nodes) containing $XY$ (and $YX$) in $G(\mathcal{F})$.
We call such a component {\it special}.
Let $v^*(XY)$ denote the set of nodes $ZW$ in $v(XY)$ with $\ave{ZW} \neq \emptyset$.

A special component has an intriguing structure;
the proof is given at the end of this section.
\begin{prop}\label{prop:special connected component}
	If both $XY$ and $YX$ are special,
	then the following hold.
	\begin{description}
		\item[{\rm (1)}] $v(XY) = (R^*(XY) \times R(YX)) \cup (R^*(YX) \times R(XY))$.
		\item[{\rm (2)}] $v^*(XY) = (R^*(XY) \times R^*(YX)) \cup (R^*(YX) \times R^*(XY))$.
		\item[{\rm (3)}] If $ZW \in v^*(XY)$,
		then $ZW$ is special and $\ave{ZW} = \ave{XY}$.
	\end{description}
\end{prop}
For a special component $v = v(XY)$,
we call $\ave{XY}$ the {\it center} of $v$;
this is well-defined by \cref{prop:special connected component}~(3).
For $Q \subseteq [r]$,
the set $\mathcal{C}$ of all special components whose center coincides with $A_Q$
is called the {\it $Q$-flower} if the size $|\mathcal{C}|$ is at least two.
The following proposition gives a concrete representation of the $Q$-flower;
the proof is given at the end of this section.
\begin{prop}\label{prop:set of special nodes}
	A $Q$-flower is given as
	\begin{align*}
		\{ v(X_iX_j) \mid 1 \leq i < j \leq p \}
	\end{align*}
	for some $p \geq 3$ and distinct $X_1, X_2, \dots, X_p \in \mathcal{F}$ such that $R(X_iX_j) = R(X_{i'}X_j)$ for all $i$ and $i' < j$, and $R(X_iX_j) \cap R(X_{i'}X_{j'}) = \emptyset$ for all distinct $j, j' \in [p]$, $i < j$, and $i' < j'$.
\end{prop}
The above $X_1, X_2, \dots, X_p$ are called the {\it representatives} of the $Q$-flower.

\begin{exmp}\label{ex:Q-flower}
	Let $\mathcal{F} = \{S,T,U,V,X,Y,Z\}$ be the $\mathcal{A}$-cut family illustrated in Figure~\ref{fig:3-local},
	its LC-graph $G(\mathcal{F})$ being illustrated in Figure~\ref{fig:Q-flower}.
	In $G(\mathcal{F})$,
	there are six special components $v(ST) (= v(SV))$, $v(SX)$, $v(TX) (= v(VX))$, $v(XY)$, $v(XZ)$, and $v(YZ)$.
	We can see that $\{ v(ST), v(SX), v(TX) \}$ is the $\{1\}$-flower
	and $\{v(XY), v(XZ), v(YZ)\}$ is the $\{2\}$-flower.~\qqed
\end{exmp}
\begin{figure}
			\centering
			\includegraphics[width=9cm]{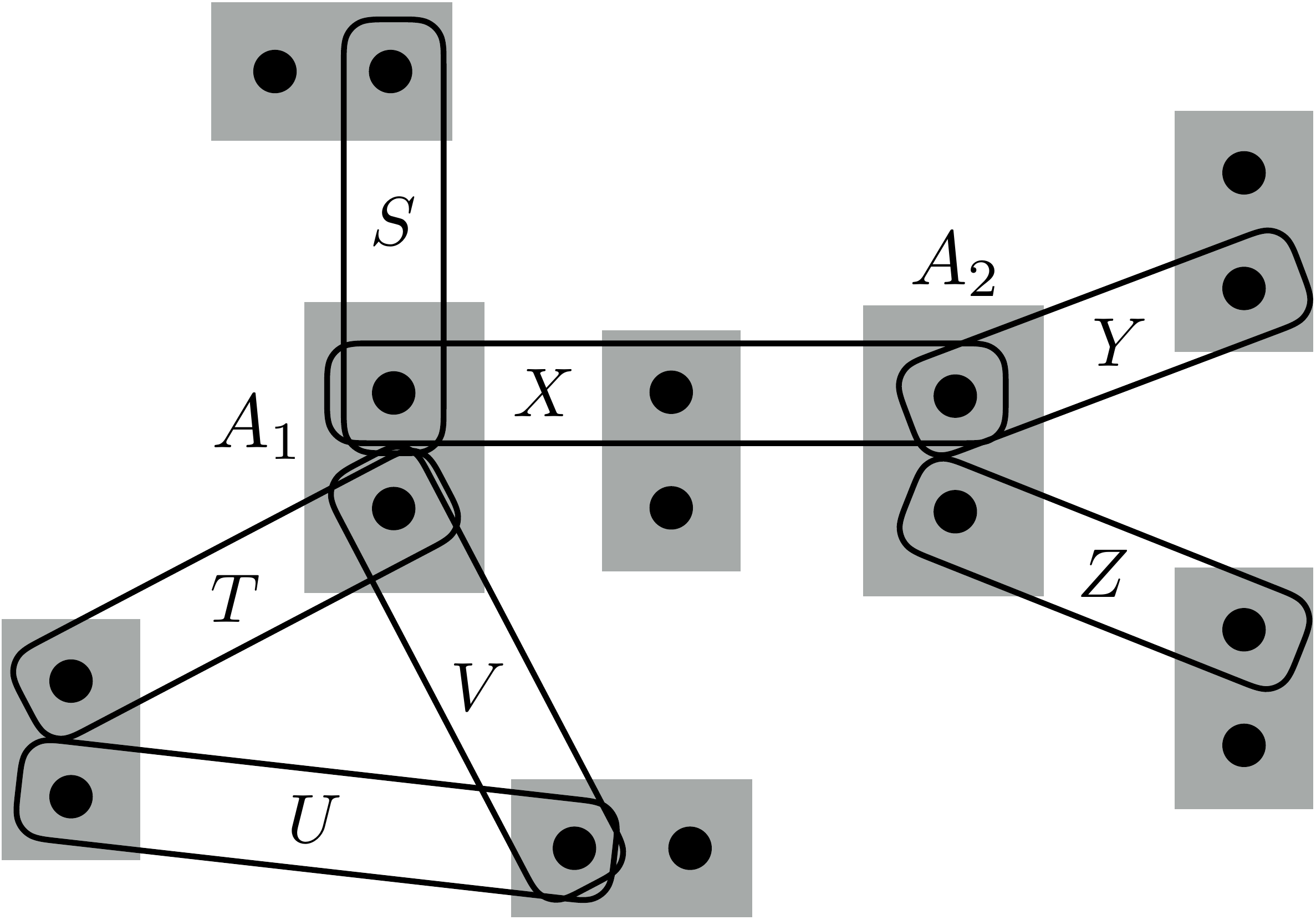}
			\caption{Black nodes indicate elements of $[n]$,
				gray rectangles indicate elements of $\mathcal{A}$, and solid curves indicate elements of $\mathcal{F} = \{ S,T,U,V,X,Y,Z \}$.
				It holds that $A_1 = \ave{ST} = \ave{TX} = \ave{SX}$ and $A_2 = \ave{XY} = \ave{YZ} = \ave{XZ}$.}
			\label{fig:3-local}
\end{figure}
\begin{figure}
			\centering
			\includegraphics[width=16cm]{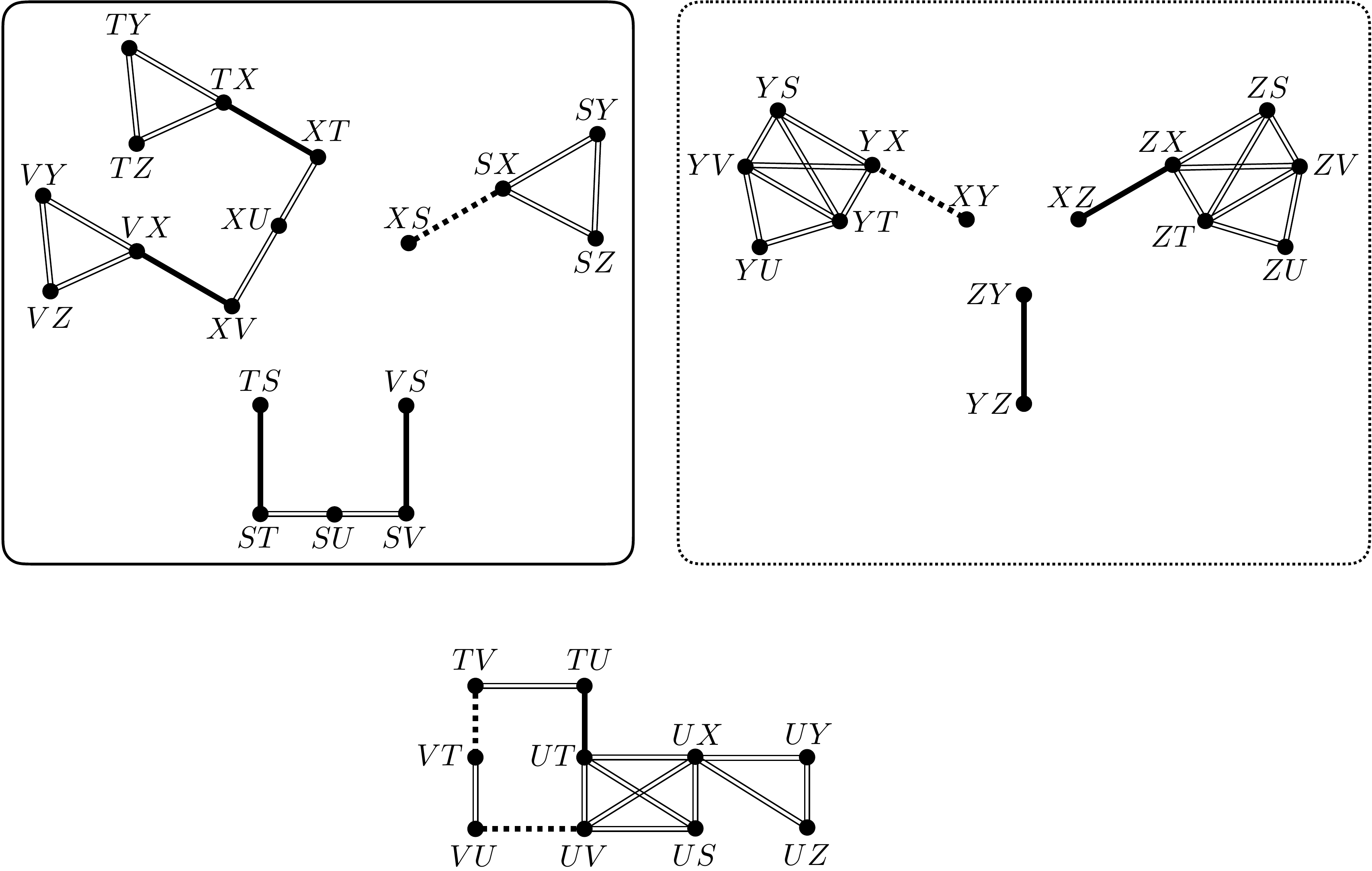}
			\caption{
				The LC-graph $G(\mathcal{F})$ for $\mathcal{F} = \{ S,T,U,V,X,Y,Z \}$ defined in Figure~\ref{fig:3-local}.
				$\{1\}$-flower (resp. $\{2\}$-flower) consists of the connected components included in the left solid curve (resp. the right dotted curve).
			}
			\label{fig:Q-flower}
\end{figure}

A component $v$ is said to be {\it fixed} if $v$ contains a fixed node,
and {\it free} otherwise.
A special component $v(XY)$ in the $Q$-flower is free if and only if
the partition lines of $X'$ and $Y'$ on $A_Q$ are the same for all $X' \in R^*(YX)$ and $Y' \in R^*(XY)$.
A {\it free $Q$-flower} is a maximal set of free components in the $Q$-flower such that the partition lines on $A_Q$ are the same.
Now the set of free components of the $Q$-flower is partitioned to free $Q$-flowers each of which is represented as
\begin{align*}
	\{ v(X_{i_s}X_{i_t}) \mid 1 \leq s < t \leq q \}
\end{align*}
with a subset $\{X_{i_1}X_{i_2}, \dots, X_{i_q}\}$ of the representatives.
A free $Q$-flower (for some $Q \subseteq [r]$) is also called a {\it free flower}.

We now provide a polynomial-time algorithm to construct a 3-locally cross-free family $\mathcal{F}^s$
by defining an appropriate LC-labeling $s$.
\begin{description}
	\item[Algorithm 5 (for constructing a 3-locally cross-free family):]
	\item[Step 0:] Determine whether there exists a 2-locally cross-free family $\mathcal{A}$-equivalent to $\mathcal{F}$.
	If not,
	then output ``$\mathcal{F}$ is not laminarizable'' and stop.
	\item[Step 1:]
	For all fixed nodes $XY$,
	define $s(XY)$ according to~\eqref{eq:fixed}.
	By a breath-first search, define $s$ on all other nodes in fixed components appropriately.
	\item[Step 2:] For each component $v$ which is free and not special,
	take any node $XY$ in $v$.
	Define $s(XY)$ as $0$ or $1$ arbitrarily, and define $s(ZW)$ appropriately for all nodes $ZW$ in $v$.
	Then all the remaining (undefined) components are special and free.
	\item[Step 3:]
	For each free flower, which is assumed to be represented as $\{ v(X_i X_j) \mid 1 \leq i < j \leq q \}$,
	do the following:
	\begin{description}
		\item[3-1:] Define the value of $s(X_iX_j)$ for $i,j \in [q]$ with $i < j$ so that $\{X_1^s, X_2^s, \dots, X_q^s\}$ is cross-free on $\bigcup_{i \in [q]}\ave{X_i}$;
		such a labeling is given, for example, as
		\begin{align}\label{eq:s}
			s(X_iX_j) :=
			\begin{cases}
			0 & \textrm{if $X_i = X_1$ on $A_Q$},\\
			1 & \textrm{if $X_i = \overline{X}_1$ on $A_Q$},
			\end{cases}
		\end{align}
		where $A_Q$ is the center of the free flower.
		\item[3-2:] Define $s(ZW)$ appropriately for all $ZW \in v(X_iX_j)$.
	\end{description}
	\item[Step 4:]
	Output $\mathcal{F}^s$.~\qqed
\end{description}

\begin{exmp}\label{ex:algo 5}
	We consider $\mathcal{F} = \{S,T,U,V,X,Y,Z\}$ in Figure~\ref{fig:3-local}
	and its LC-graph $G(\mathcal{F})$ in Figure~\ref{fig:Q-flower}.
	We execute Algorithm~5 for $G(\mathcal{F})$.
	
	We can easily determine that there exists a 2-locally cross-free family $\mathcal{A}$-equivalent to $\mathcal{F}$,
	and that there is no fixed node in $G(\mathcal{F})$.
	In Step~2, there is one component $v$ which is free and not special in $G(\mathcal{F})$ (the one at the bottom in Figure~\ref{fig:Q-flower}).
	We take, say, $TV \in v$ and define $s(TV) := 1$.
	Then,
	by~\eqref{eq:swapped} and~\eqref{eq:prefixed},
	we have $s(X'Y') = 0$ for $X'Y' \in \{ VT, VU \}$ and $s(X'Y') = 1$ for other nodes in $v$.
	We consider Step~3.
	Two flowers ($\{1\}$-flower and $\{2\}$-flower) exist in $G(\mathcal{F})$ (see \cref{ex:Q-flower}).
	Hence,
	for the $\{1\}$-flower $\{ v(ST), v(SX), v(TX) \}$,
	we define $s(ST) = s(SX) = s(TX) := 0$,
	and for the $\{2\}$-flower $\{v(XY), v(XZ), v(YZ)\}$,
	we define $s(XY) = s(XZ) = s(YZ) := 0$.
	Then we define the other values according to~\eqref{eq:swapped} and~\eqref{eq:prefixed}.
	Thus we can construct an LC-labeling inducing a 3-locally cross-free family.~\qqed
\end{exmp}

\begin{prop}\label{prop:algo 5}
	The output $\mathcal{F}^s$ is 3-locally cross-free, and Algorithm~{\rm 5} runs in $O(n^4)$ time.
\end{prop}
\begin{proof}
	We show the 3-local cross-freeness of $\mathcal{F}^s$.
	Recall that $\mathcal{F}^s$ is 2-locally cross-free
	and $\ave{X^s} = \ave{X}$ for $X^s \in \mathcal{F}^s$ (and $X \in \mathcal{F}$).
	Take any triple $\{X^s, Y^s, Z^s\}$ with $\ave{XYZ} \neq \emptyset$.
	It suffices to deal with the case of $\ave{XY} = \ave{YZ} = \ave{ZX} \neq \emptyset$ by \cref{lem:cross-free on S}~(2).
	If $XY$ is not special,
	there is a path $(XY, XY_1, \dots, XY_k)$ in $G(\mathcal{F})$ such that $\emptyset \neq \ave{XY_k} \neq \ave{XY} = \ave{XZ}$.
	Here $\ave{XY_k} \not\subseteq \ave{Z}$ or $\ave{XZ} \not\subseteq \ave{Y_k}$ holds.
	By \cref{lem:cross-free on S}~(2),
	$\{X^s, Y_k^s, Z^s\}$ is cross-free on $\ave{X} \cup \ave{Y_k} \cup \ave{Z}$.
	Hence, by \cref{lem:cross-free on S}~(4),
	$\{X, Y, Z\}$ is cross-free on $\ave{X} \cup \ave{Y} \cup \ave{Z}$.
	Therefore, we assume that $XY, YX, YZ,ZY, ZX,XZ$ are special.
	
	We can suppose that $XY, YZ, ZX$ belong to special components of the $Q$-flower $\{ v(X_iX_j) \mid 1 \leq i < j \leq p \}$,
	i.e.,
	$\ave{XY} = \ave{YZ} = \ave{ZX} = A_Q$.
	By \cref{prop:set of special nodes},
	we can assume $X \in R^*(X_kX_i)$, $Y \in R^*(X_iX_j)$, and $Z \in R^*(X_jX_k)$ for distinct $i,j,k \in [p]$ with $i < j < k$.
	
	Suppose that $v(X_iX_j)$, $v(X_iX_k)$, or $v(X_jX_k)$ is fixed.
	Then we can assume that there is $\hat{X} \in R^*(X_kX_i)$ such that
	the partition lines of $\hat{X}, Y, Z$ are not the same.
	By \cref{lem:cross-free on S}~(3),
	$\{ \hat{X}^s, Y^s, Z^s \}$ is cross-free on $\ave{\hat{X}} \cup \ave{Y} \cup \ave{Z}$.
	Furthermore, since there is a path $(YX= YX_0, YX_1, \dots, YX_k = Y\hat{X})$,
	$\{X^s, Y^s, Z^s\}$ is cross-free on $\ave{X} \cup \ave{Y} \cup \ave{Z}$
	by \cref{lem:cross-free on S}~(4).
	
	Suppose that $v(X_iX_j)$, $v(X_iX_k)$, and $v(X_jX_k)$ are free.
	Then $v(X_iX_j), v(X_iX_k), v(X_jX_k)$ are contained in the same free $Q$-flower.
	By the definition of $s$ (cf.~\eqref{eq:s}),
	$\{ X_i^s, X_j^s, X_k^s \}$ is cross-free on $\ave{X_i} \cup \ave{X_j} \cup \ave{X_k}$.
	By applying repeatedly \cref{lem:cross-free on S} (4), $\{X^s, Y^s, Z^s\}$ is cross-free on $\ave{X} \cup \ave{Y} \cup \ave{Z}$.

	Finally we see the running-time of Algorithm~5.
	By the argument at the end of \cref{subsec:LCF},
	Step~0 can be done in $O(n^4)$ time.
	We can also obtain an appropriate value of each $s(XY)$ in Steps~1--3 in $O(n^4)$ time.
	From $s$, we can construct $\mathcal{F}^s$ in $O(|V(\mathcal{F})|) = O(n^2)$ time.
	Thus the running-time of Algorithm~5 is bounded by $O(n^4)$.
\end{proof}

By \cref{prop:algo 4,prop:algo 5},
we obtain the following theorem.
\begin{thm}\label{thm:algo 4 5}
	Algorithms~{\rm 4} and~{\rm 5} solve \LAM\ in $O(n^4)$ time.
\end{thm}

The rest of this section is devoted to proving \cref{prop:special connected component,prop:set of special nodes}.
First we show a key lemma about special nodes.
\begin{lem}\label{lem:special}
	If $XY$ is special and $\ave{XY} = \ave{X'Y}$ for some $X'$,
	then $R(X'Y) \subseteq R(XY)$, and $\ave{X'Z} \supseteq \ave{XZ}$ for any $Z \in R(X'Y)$.
\end{lem}
\begin{proof}
	We prove $Y_k \in R(XY)$ and $\ave{X'Y_k} \supseteq \ave{XY_k}$ by induction on the length $k$ of a path $(X'Y = X'Y_0, X'Y_1, \dots, X'Y_k)$.
	For $k = 0$, we have $Y_0 = Y \in R(XY)$ and $\ave{X'Y_0} \supseteq \ave{XY_0}$.
	For the induction step, suppose that $Y_{k} \in R(XY)$ and $\ave{X' Y_k} \supseteq \ave{X Y_k}$ for $k \geq 0$.
	Since a prefixed edge $\{X' Y_k, X' Y_{k+1}\}$ exists,
	we have $\ave{Y_k Y_{k+1}} \setminus \ave{X' Y_k} \neq \emptyset$.
	Then $Y_{k+1} \neq X$ holds.
	Indeed,
	if $Y_{k+1} = X$,
	then $\ave{X Y_k} \setminus \ave{X' Y_k} \neq \emptyset$,
	a contradiction to $\ave{X' Y_k} \supseteq \ave{X Y_k}$.
	By $\ave{X' Y_k} \supseteq \ave{X Y_k}$,
	we obtain $\ave{Y_k Y_{k+1}} \setminus \ave{X Y_k} \neq \emptyset$.
	Hence there is a prefixed edge $\{X Y_k, X Y_{k+1}\}$.
	This means $Y_{k+1} \in R(XY)$.
	
	Suppose, to the contrary, that $\ave{X' Y_{k+1}} \not\supseteq \ave{X Y_{k+1}}$, i.e.,
	$\ave{X Y_{k+1}} \setminus \ave{X' Y_{k+1}} \neq \emptyset$ holds.
	Note that $\ave{X Y_{k+1}} \setminus \ave{X' Y_{k+1}} = \ave{X Y_{k+1}} \setminus \ave{X X'}$ holds.
	Furthermore, by $\ave{XY} = \ave{X'Y}$,
	we obtain $\ave{XX'} \supseteq \ave{XY}$.
	Hence we have $\ave{X Y_{k+1}} \setminus \ave{XY} \neq \emptyset$.
	However, since $XY$ is special and $Y_{k+1} \in R(XY)$,
	it must hold that $\ave{XY_{k+1}} = \ave{XY}$ or $\ave{XY_{k+1}} = \emptyset$;
	this is a contradiction.
\end{proof}

\begin{proof}[Proof of \cref{prop:special connected component}]
	First we show the following three claims.
	\setcounter{thm1}{0}
	\begin{cl}\label{cl:empty}
		$R(XY) \cap R(YX) = \emptyset$.
	\end{cl}
	\begin{proof}
		Suppose, to the contrary, that $R(XY) \cap R(YX) \neq \emptyset$.
		For each $Z \in R(XY)$,
		we have $\ave{XZ} \subseteq \ave{YZ}$
		since $\ave{XZ} = \ave{XY}$ or $\ave{XZ} = \emptyset$.
		
		Let $Z \in R(XY) \cap R(YX)$ be an element such that
		the length $k$ of a path $(YX, \dots, YZ)$ in $G(\mathcal{F})$ is shortest.
		If $k \geq 2$,
		there is a prefixed edge $\{ YZ_k, YZ_{k-1} \}$ and $Z_{k-1} \neq X$.
		That is $\ave{Z_k Z_{k-1}} \setminus \ave{YZ_k} \neq \emptyset$.
		By $\ave{XZ_k} \subseteq \ave{YZ_k}$,
		we obtain $\ave{Z_k Z_{k-1}} \setminus \ave{XZ_k} \neq \emptyset$.
		Hence a prefixed edge $\{XZ_{k}, XZ_{k-1} \}$ exists.
		This means $Z_{k-1} \in R(XY) \cap R(YX)$,
		which contradicts the minimality of $Z = Z_k$.
		Therefore a prefixed edge $\{YX, YZ\}$ exists for some $Z \in R(XY) \cap R(YX)$.
		That is, $\ave{XZ} \setminus \ave{XY} \neq \emptyset$.
		Hence we obtain $\emptyset \neq \ave{XZ} \neq \ave{XY}$.
		This contradicts the assumption that $XY$ is special.
	\end{proof}
	\begin{cl}\label{cl:one swapped =}
		For any $Y' \in R^*(XY)$,
		it holds that
		$R(YX) = R(Y'X)$, and $\ave{Y'Z} = \ave{YZ}$ for any $Z \in R(YX) = R(Y'X)$.
	\end{cl}
	\begin{proof}
		If $R^*(XY) = \{Y\}$,
		the proof is trivial.
		Suppose $R^*(XY) \setminus \{Y\} \neq \emptyset$.
		Take any $Y' \in R^*(XY) \setminus \{Y\}$.
		Then there is a swapped edge $\{ XY', Y'X \}$.
		Therefore, for all $Z \in R(Y'X)$, $XY$ and $Y'Z$ are connected.
		Since $XY$ is special, it holds that $\ave{YX} = \ave{Y'X}$.
		Since $YX$ is special and $\ave{YX} = \ave{Y'X}$,
		by \cref{lem:special},
		we have $R(Y'X) \subseteq R(YX)$ and $\ave{Y'Z} \supseteq \ave{YZ}$ for all $Z \in R(Y'X)$.
		
		In the following,
		we prove that, for each $Z \in R(YX)$, it holds that $Z \in R(Y'X)$ and $\ave{Y'Z} \subseteq \ave{YZ}$,
		which imply $R(Y'X) = R(YX)$ and $\ave{Y'Z} = \ave{YZ}$.
		We show this by induction on the length of a path $(YX = YX_0, YX_1, \dots, YX_{k+1} = YZ)$.
		For $X_0$,
		we have $R(Y'X) \ni X = X_0$ and $\ave{YX_0} = \ave{Y'X_0}$.
		Suppose $R(Y'X) \ni X_{k}$ and $\ave{YX_k} \supseteq \ave{Y'X_k}$ by induction.
		Since a prefixed edge $\{YX_k, YX_{k+1}\}$ exists,
		we have $\ave{X_k X_{k+1}} \setminus \ave{YX_k} \neq \emptyset$.
		By $\ave{YX_k} \supseteq \ave{Y'X_k}$,
		we obtain $\ave{X_k X_{k+1}} \setminus \ave{Y' X_k} \neq \emptyset$.
		Hence there is a prefixed edge $\{Y' X_k, Y' X_{k+1}\}$.
		This means $R(Y'X) \ni X_{k+1} = Z$.
		
		Suppose, to the contrary, that $\ave{YX_{k+1}} \not\supseteq \ave{Y'X_{k+1}}$, i.e.,
		$\ave{Y'X_{k+1}} \setminus \ave{YX_{k+1}} \neq \emptyset$ holds.
		Then there is a prefixed edge $\{ YX_{k+1}, YY' \}$.
		Hence we have $R(YX) \ni Y'$.
		However this contradicts $R(YX) \not\ni Y'$
		by \cref{cl:empty} and $R(XY) \ni Y'$.
		Therefore we obtain $\ave{YX_{k+1}} \supseteq \ave{Y' X_{k+1}}$.
	\end{proof}
	
	\begin{cl}\label{cl:using one swapped}
		For $ZW \in v(XY)$,
		there is a path from $XY$ or $YX$ to $ZW$ containing at most one swapped edge.
	\end{cl}
	\begin{proof}
		Suppose, to the contrary, that, for some $ZW$, all paths from $XY$ to $ZW$ and from $YX$ to $ZW$ use at least two swapped edges.
		Take such a path $P$ with a minimum number of swapped edges.
		Denote the number of swapped edges in $P$ by $k (\geq 2)$.
		Without loss of generality, we assume that $P$ is a path from $XY$ to $ZW$.
		By $k \geq 2$,
		$P$ has a subpath $(XY = X_0Y_0, \dots, X_0Y_1, Y_1X_0, \dots, Y_1X_1, X_1Y_1)$.
		Note that $Y_1 \in R^*(XY)$, and $X_1 \in R^*(Y_1X) = R^*(YX)$ by \cref{cl:one swapped =}.
		Hence there is a path from $YX$ to $X_1Y_1$ using only one swapped edge.
		Indeed, $(YX = Y_0X_0, \dots, Y_0X_1, X_1Y_0, \dots, X_1Y_1)$ is such a path.
		This means that there is a path from $YX$ to $ZW$ with $k-1$ swapped edges,
		a contradiction to the minimality of $P$.
	\end{proof}
	
	We are now ready to show the statement of \cref{prop:special connected component}~(1).
	If $Z \in R^*(XY)$ and $W \in R(YX)$,
	then there is a path in $G(\mathcal{F})$ such as $(XY, \dots, XZ, ZX, \dots, ZW)$ since $R(YX) = R(ZX)$ by \cref{cl:one swapped =}, implying $ZW \in v(XY)$.
	Conversely, if $ZW \in v(XY)$,
	then there is a path from $XY$ or $YX$ to $ZW$ with at most one swapped edge by \cref{cl:using one swapped}.
	We may assume that there is such a path $P$ from $XY$ to $ZW$.
	If $P$ has no swapped edge,
	then $Z = X \in R^*(YX)$ and $W \in R(XY)$ hold.
	If $P$ has exactly one swapped edge,
	then $Z \in R^*(XY)$ and $W \in R(ZX) = R(YX)$ by \cref{cl:one swapped =}.
	Thus we obtain \cref{prop:special connected component}~(1).
	
	Next we show \cref{prop:special connected component}~(3).
	If $ZW \in v(XY)$,
	then there is a path from $XY$ or $YX$ to $ZW$ with at most one swapped edge by \cref{cl:using one swapped}.
	We may assume that there is such a path $P$ from $XY$ to $ZW$.
	If $P$ has no swapped edge,
	then $\ave{ZW} = \ave{XY}$ or $\ave{ZW} = \emptyset$ holds since $XY$ is special.
	If $P$ has one swapped edge,
	then $\ave{ZW} = \ave{YW}$ holds by \cref{cl:one swapped =}
	and $\ave{YW} = \ave{XY}$ or $\ave{YW} = \emptyset$ holds since $YX$ is special.
	Therefore, if $ZW \in v^*(XY)$, then $\ave{ZW} = \ave{XY}$,
	and $ZW$ is obviously special.
	Thus we obtain \cref{prop:special connected component}~(3).
	
	Finally we show \cref{prop:special connected component}~(2).
	For every $Z \in R^*(XY)$ and $W \in R^*(YX)$,
	we have $\ave{Z} \supseteq \ave{XY} \subseteq \ave{W}$ by (3).
	Hence $\ave{ZW} \neq \emptyset$, implying $ZW \in v^*(XY)$.
	Conversely,
	let $ZW \in v^*(XY)$.
	By \cref{prop:special connected component}~(1),
	we may assume $Z \in R^*(XY)$ and $W \in R(YX)$.
	Since $\ave{ZW} \neq \emptyset$,
	$\ave{ZW} = \ave{XY}$ holds by \cref{prop:special connected component}~(3).
	Hence $\ave{W} \cap \ave{Y} \supseteq \ave{XY} \neq \emptyset$.
	This means $W \in R^*(YX)$.
\end{proof}

\begin{proof}[Proof of \cref{prop:set of special nodes}]
	Let $v(X_1X_2)$ be a special connected component with $\ave{X_1X_2} = A_Q$.
	Take any special connected component $v(Y_1Y_2)$ with $\ave{Y_1Y_2} = A_Q$.
	It suffices to show that, (i) if $R(X_1X_2) \cap R(Y_1Y_2) \neq \emptyset$, we have $R(X_1X_2) = R(Y_1Y_2)$ (this implies $v(Y_1Y_2) = v(Y_1X_2)$ by $Y_2 \in R^*(X_1X_2)$),
	and (ii) if $R(X_1X_2) \cap R(Y_1Y_2) = \emptyset$, there exists a special connected component $v(X_2Y_2)$ with $\ave{X_2Y_2} = A_Q$, $R(X_2Y_2) = R(Y_1Y_2)$, and $R(Y_2X_2) = R(X_1X_2)$.
	
	(i).
	If there exists $Z \in R^*(X_1X_2) \cap R^*(Y_1Y_2)$,
	then $X_1Z$ and $Y_1Z$ are special and $\ave{X_1Z} = \ave{Y_1Z} (= A_Q)$.
	Hence, by \cref{lem:special},
	we have $R(X_1Z) \subseteq R(Y_1Z)$ and $R(X_1Z) \supseteq R(Y_1Z)$,
	i.e., $R(X_1Z) = R(Y_1Z)$.
	This implies $R(X_1X_2) = R(X_1Z) = R(Y_1Z) = R(Y_1Y_2)$, as required.
	Thus, in the following, we show that there exists $Z \in R^*(X_1X_2) \cap R^*(Y_1Y_2)$.
	
	Suppose, to the contrary, $R^*(X_1X_2) \cap R^*(Y_1Y_2) = \emptyset$.
	Note that $R^*(X_1X_2) \cap R^*(Y_1Y_2) = \emptyset$ implies $R^*(X_1X_2) \cap R(Y_1Y_2) = R(X_1X_2) \cap R^*(Y_1Y_2) = \emptyset$.
	Indeed, each $Z \in R^*(X_1X_2) \cap R(Y_1Y_2)$ satisfies $Z \supseteq A_Q$ by $Z \in R^*(X_1X_2)$.
	Hence $Z \in R^*(Y_1Y_2)$ holds by $Z \in R(Y_1Y_2)$ and $\ave{Y_1Z} \neq \emptyset$.
	Let $Z \in R(X_1X_2) \cap R(Y_1Y_2) = (R(X_1X_2) \cap R(Y_1Y_2)) \setminus (R^*(X_1X_2) \cup R^*(Y_1Y_2))$ be an element such that
	the length of a path $(X_1X_2 = X_1Z_0, X_1Z_1, \dots, X_1Z_k = X_1Z)$ is shortest; by the assumption, $k \geq 1$.
	Since a prefixed edge $\{X_1Z_k, X_1Z_{k-1}\}$ exists,
	we have $\ave{Z_k Z_{k-1}} \setminus \ave{X_1Z_k} \neq \emptyset$.
	Furthermore, by $\ave{X_1Z_k} = \ave{Y_1Z_k} = \emptyset$,
	we obtain $\ave{Z_k Z_{k-1}} \setminus \ave{Y_1Z_k} \neq \emptyset$.
	This means that a prefixed edge $\{Y_1Z_k, Y_1Z_{k-1}\}$ exists and $Z_{k-1} \in R(X_1X_2) \cap R(Y_1Y_2)$ holds, a contradiction to the minimality of $k$.
	
	(ii).
	First we show that $\ave{X_2'Y_2'} = A_Q$ or $\ave{X_2'Y_2'} = \emptyset$ holds for any $X_2' \in R(X_1X_2)$ and $Y_2' \in R(Y_1Y_2)$.
	Since, for any $Z \in R(X_1X_2) \cup R(Y_1Y_2)$, $\ave{Z} \supseteq A_Q$ or $\ave{Z} \cap A_Q = \emptyset$ holds by \cref{prop:special connected component}~(3),
	we have $\ave{X_2'Y_2'} \supseteq A_Q$ or $\ave{X_2'Y_2'} \cap A_Q = \emptyset$ for each $X_2' \in R(X_1X_2)$ and $Y_2' \in R(Y_1Y_2)$ with $\ave{X_2'Y_2'} \neq \emptyset$.
	Suppose, to the contrary, that there exist $X_2' \in R(X_1X_2)$ and $Y_2' \in R(Y_1Y_2)$ with $\emptyset \neq \ave{X_2'Y_2'} \neq A_Q$.
	Then $\ave{X_2'Y_2'} \supsetneq A_Q$ or $\ave{X_2'Y_2'} \cap A_Q = \emptyset$ holds.
	Hence we have $\ave{X_2'Y_2'} \setminus \ave{X_1X_2'} \neq \emptyset$ by $\ave{X_1X_2'} = A_Q$ or $\ave{X_1X_2'} = \emptyset$.
	This means that there is a prefixed edge $\{X_1X_2', X_1Y_2'\}$
	and $R(X_1X_2) \cap R(Y_1Y_2) \neq \emptyset$ holds, a contradiction.
	
	By $\ave{X_2} \supseteq A_Q \subseteq \ave{Y_2}$, we have $\ave{X_2Y_2} \neq \emptyset$.
	Hence, by the above argument, we obtain $\ave{X_2Y_2} = A_Q$.
	Furthermore $Y_1Y_2$ is special and $\ave{Y_1Y_2} = \ave{X_2Y_2}$ holds.
	By \cref{lem:special},
	we obtain $R(X_2Y_2) \subseteq R(Y_1Y_2)$.
	By $\ave{X_2Z} = A_Q$ or $\ave{X_2Z} = \emptyset$ for every $Z \in R(X_2Y_2) \subseteq R(Y_1Y_2)$,
	it holds that $X_2Y_2$ is special.
	Furthermore, since $X_2Y_2$ is special,
	we also obtain $R(X_2Y_2) \supseteq R(Y_1Y_2)$ by \cref{lem:special}.
	Hence $R(X_2Y_2) = R(Y_1Y_2)$ holds.
	By a similar argument,
	$Y_2X_2$ is special and
	$R(Y_2X_2) = R(X_1X_2)$ holds.
	Thus, a special component $v(X_2Y_2)$ with $\ave{X_2Y_2} = A_Q$ exists.
\end{proof}

\section*{Acknowledgments}
We thank the referees for helpful comments.
The first author's research was partially supported by JSPS KAKENHI Grant Numbers 25280004,
26330023, 17K00029.
The second author's research was supported by JSPS Research Fellowship for Young Scientists.
The third author's research was supported by The Mitsubishi Foundation,
CREST, JST, Grant Number JPMJCR14D2, Japan, and
JSPS KAKENHI Grant Number 26280004.
The last author's research was supported by a Royal Society University 
Research Fellowship. This project has received funding from the European 
Research Council (ERC) under the European Union's Horizon 2020 research 
and innovation programme (grant agreement No 714532). The paper reflects 
only the authors' views and not the views of the ERC or the European 
Commission. The European Union is not liable for any use that may be 
made of the information contained therein.

\end{document}